\newif\iffullversion
\fullversiontrue
\iffullversion
\documentclass[11pt]{article}
\usepackage[a4paper,margin=2cm]{geometry}
\else
\documentclass[10pt, conference,
compsocconf,letterpaper]{IEEEtran}
\fi

\usepackage[T1]{fontenc}

\usepackage{amsthm} 

\usepackage{amsmath}
\usepackage{amssymb}
\usepackage{stmaryrd}
\usepackage[utf8]{inputenc}
\usepackage{xspace}
\usepackage{proof}
\usepackage{color}
\usepackage{paralist}
\usepackage{thm-restate}

\usepackage{hyperref}

\usepackage{listings} %applied pi listings

\usepackage{syntax} %for grammar specification
\usepackage{enumitem} % to interrupt and resume an itemize

\usepackage{graphicx} %to be able to crop the pdf in latex
\usepackage{tikz} %for pictures appearing in the proof
\usetikzlibrary{decorations.pathmorphing,decorations.pathreplacing,decorations.shapes}
\tikzset{ snake/.style={decorate, decoration={zigzag, pre length =
      10mm, post length = 10mm,} }, }

\newcommand{\downloadlink}{http://sapic.gforge.inria.fr/}

\newcommand{\trans}[1]{\ensuremath{{\xrightarrow{#1}}}\xspace}

\newcommand{\set}[1]{\{\,#1\,\}}

\DeclareMathOperator\E E

\newcommand{\inlineapip}[1]{
\lstinline[columns=fullflexible,breaklines=true]{#1}
}

\newcommand{\code}[1]{\textsf{#1}\xspace} %How we want code to look like
\newcommand{\apip}[1]{\textsf{#1}\xspace} %How we want code to look like

\newtheorem{definition}{Definition}

\newtheorem{lemma}{Lemma}

\newtheorem{theorem}{Theorem}

\newtheorem{proposition}{Proposition}

\newtheorem{remark}{Remark}

\newtheoremstyle{mycase}{}{}{}{}{\bf}{.}{.5em}{\thmnote{#1:}~#3}
\theoremstyle{mycase}
\newtheorem{mycase}{Case}

\numberwithin{subcase}{mycase}

\newtheoremstyle{component}{}{}{}{}{\itshape}{.}{.5em}{\thmnote{#3}#1}
\theoremstyle{component}
\newtheorem*{component}{}

\DeclareMathOperator*{\dom}{\Dom}

\newcommand\setN{\mathbb N}

\newcommand{\mset}[1]{\ensuremath{\{\,#1\,\}^\#\xspace}}
\newcommand{\msetminus}{\ensuremath{\setminus^{\#}\xspace}}
\newcommand{\mcup}{\ensuremath{\cup^{\#}\xspace}}
\newcommand{\msubset}{\ensuremath{\subset^{\#}\xspace}}
\newcommand{\msubseteq}{\ensuremath{\subseteq^{\#}\xspace}}
\newcommand{\melem}{\ensuremath{\in^{\#}\xspace}}

\newcommand{\ie}{i.\,e.}

\newcommand{\eg}{e.\,g.}

\newcommand{\calE}{\ensuremath{\mathcal{E}}\xspace}
\newcommand{\calF}{\ensuremath{\mathcal{F}}\xspace}

\newcommand{\calL}{\ensuremath{\mathcal{L}}\xspace}
\newcommand{\calM}{\ensuremath{\mathcal{M}}\xspace}
\newcommand\calP{\ensuremath{\mathcal{P}}\xspace}

\newcommand\calS{\ensuremath{\mathcal{S}}\xspace}

\newcommand{\pair}[2]{\ensuremath{\langle #1, #2 \rangle}}
\newcommand{\fst}{\ensuremath{\mathsf{fst}}}
\newcommand{\snd}{\ensuremath{\mathsf{snd}}}

\newcommand{\msgsort}{\ensuremath{\mathit{msg}}\xspace}
\newcommand{\freshsort}{\ensuremath{\mathit{fresh}}\xspace}
\newcommand{\pubsort}{\ensuremath{\mathit{pub}}\xspace}
\newcommand{\tempsort}{\ensuremath{\mathit{temp}}\xspace}

\newcommand{\PN}{\ensuremath{\mathit{PN}}\xspace}
\newcommand{\FN}{\ensuremath{\mathit{FN}}\xspace}
\newcommand{\Vars}{\ensuremath{\mathcal{V}}\xspace}

\newcommand{\names}{\ensuremath{\mathit{names}}\xspace}

\newcommand{\Sign}{\ensuremath{\Sigma}\xspace}
\newcommand{\ET}{\ensuremath{E}\xspace}
\newcommand{\Terms}{\ensuremath{\mathcal{T}}\xspace}
\newcommand{\Mess}{\ensuremath{\mathcal{M}}\xspace}
\newcommand{\FSign}{\ensuremath{\Sigma_\mathit{fact}}\xspace}
\newcommand{\GroundFacts}{\ensuremath{\mathcal{G}}\xspace}

\newcommand{\Subst}{\ensuremath{\sigma}\xspace}

\newcommand{\Names}{\ensuremath{\calE}\xspace}
\newcommand{\StoreA}{\ensuremath{\calS}\xspace}
\newcommand{\StoreB}{\ensuremath{\calS^\mathrm{MS}}\xspace}
\newcommand{\Processes}{\ensuremath{\calP}\xspace}
\newcommand{\ActiveLocks}{\ensuremath{\calL}\xspace}
\newcommand{\piout}{\ensuremath{\mathrm{out}}\xspace}
\newcommand{\piin}{\ensuremath{\mathrm{in}}\xspace}

\newcommand{\apipstate}[1]{\ensuremath{
(\Names_{#1},\allowbreak\StoreA_{#1},\allowbreak\StoreB_{#1},\allowbreak\Processes_{#1},\allowbreak\Subst_{#1},\allowbreak\calL_{#1}) }}

\newcommand{\msrewrite}[1]{\mathrel{-\hspace{-2pt}[#1]\hspace{-4pt}\to}}
\newcommand{\emptyrule}{\ensuremath{[]}\xspace}
\newcommand{\msr}[3]{\ensuremath{[#1] \msrewrite{#2} [#3]}}

\newcommand{\ri}{\ensuremath{\mathit{ri}}\xspace}
\newcommand{\prems}{\ensuremath{\mathit{prems}}\xspace}
\newcommand{\actions}{\ensuremath{\mathit{actions}}\xspace}
\newcommand{\conclusions}{\ensuremath{\mathit{conclusions}}\xspace}
\newcommand{\ginsts}{\ensuremath{\mathit{ginsts}}\xspace}
\newcommand{\pfacts}{\ensuremath{\mathit{pfacts}}\xspace}
\newcommand{\lfacts}{\ensuremath{\mathit{lfacts}}\xspace}
\newcommand{\tracesmsr}{\ensuremath{\mathit{traces}^\mathit{msr}}\xspace}
\newcommand{\execmsr}{\ensuremath{\mathit{exec}^\mathit{msr}}\xspace}
\newcommand{\tr}{\ensuremath{\mathit{tr}}\xspace}
\newcommand{\Tr}{\ensuremath{\mathit{Tr}}\xspace}
\newcommand{\temp}{\ensuremath{\mathit{temp}}\xspace}

\newcommand{\vars}{\ensuremath{\mathit{vars}}\xspace}
\newcommand{\st}{\ensuremath{\mathit{st}}\xspace}
\newcommand{\tracespi}{\ensuremath{\mathit{traces}^\mathit{pi}}\xspace}

\newcommand{\Dom}{\ensuremath{\mathbf{D}\xspace}}

\newcommand\sem[1]{{\llbracket {#1} \rrbracket}}

\newcommand{\state}{{\sf state}}
\newcommand{\semistate}{{\sf state}^{\sf semi}}
\newcommand{\dummy}{{\sf !Dum}}

\newcommand{\RepNonce}{\ensuremath{\mathit{RepNonce}}\xspace}
\newcommand{\ProtoNonce}{\ensuremath{\mathit{ProtoNonce}}\xspace}

\newcommand{\Eq}{\ensuremath{\mathrm{Eq}}\xspace}
\newcommand{\Fr}{\ensuremath{\mathsf{Fr}}\xspace}
\newcommand{\K}{\ensuremath{!\mathsf{K}}\xspace}
\newcommand{\In}{\ensuremath{\mathsf{In}}\xspace}
\newcommand{\Out}{\ensuremath{\mathsf{Out}}\xspace}
\newcommand{\Msg}{\ensuremath{\mathsf{Msg}}\xspace}
\newcommand{\Ack}{\ensuremath{\mathsf{Ack}}\xspace}
\newcommand{\Event}{\ensuremath{\mathrm{Event}}\xspace}
\newcommand{\ChannelInEvent}{\ensuremath{\mathrm{InEvent}}\xspace}

\newcommand{\ReservedFacts}{\ensuremath{\mathcal{F}_{\mathit{res}}}\xspace}

\newcommand{\hide}{\ensuremath{\mathit{hide}}\xspace}
\newcommand{\filter}{\ensuremath{\mathit{filter}}\xspace}

\newcommand{\bijp}{\ensuremath{\leftrightsquigarrow_{P}\xspace}}

\newcommand{\Ass}{\ensuremath{\alpha}\xspace}
\newcommand{\AssInit}{\ensuremath{\alpha_\mathit{init}}\xspace}
\newcommand{\AssEq}{\ensuremath{\alpha_\mathit{eq}}\xspace}
\newcommand{\AssNotEq}{\ensuremath{\alpha_\mathit{noteq}}\xspace}
\newcommand{\AssSetIn}{\ensuremath{\alpha_\mathit{in}}\xspace}
\newcommand{\AssSetNotIn}{\ensuremath{\alpha_\mathit{notin}}\xspace}
\newcommand{\AssLock}{\ensuremath{\alpha_\mathit{lock}}\xspace}
\newcommand{\AssIn}{\ensuremath{\alpha_\mathit{inev}}\xspace}

\newcommand{\theactualrule}[1]{\text{Please redefine the command
theactualrule.}}
\newcommand{\underscorethingy}[1]{\text{Please redefine the command
underscorethingy.}}

\newcommand{\Fresh}{\textsc{Fresh}\xspace}
\newcommand{\MD}{\textsc{MD}\xspace}
\newcommand{\MDOutPubFreshAppl}{\ensuremath{\set{ \textsc{MDOut},\allowbreak \textsc{MDPub},\allowbreak \textsc{MDFresh},\allowbreak \textsc{MDAppl}}}}
\newcommand{\MDOutPubFreshApplFresh}{\ensuremath{\set{ \textsc{MDOut},\allowbreak \textsc{MDPub},\allowbreak \textsc{MDFresh},
\allowbreak \textsc{MDAppl},
\allowbreak \textsc{Fresh}
}}}

\lstdefinelanguage{apip}{
  morekeywords=[1]{
	out, in, if, then, else, event, insert, delete, lookup, as, in, lock, unlock
    },
  sensitive=true,
  morecomment=[l]{\#},
  morecomment=[l]{//},
  morestring=[b]',
  morestring=[s]{`}{'},
}

\lstdefinestyle{apip}{
    language={apip}, 
	stringstyle=\rm\textup,
	keywordstyle=[1]\sf,
 	keywordstyle=[2]\rm\upshape,
	literate={||}{{$\mid$}}1
	{new}{{$\nu$}}1
	{<}{{$\langle$}}1
	{>}{{$\rangle$}}1
	{(}{{$($}}1
	{)}{{$)$}}1
}

\newtheorem{example}{Example}

\iffullversion
\usepackage[affil-it]{authblk}
\else
\IEEEoverridecommandlockouts
\fi

\begin{document}

% \date{May 8, 2013}
\iffullversion
\author{Steve Kremer} 
\affil{INRIA~Nancy - Grand'Est \& Loria, France}

\author{Robert K\"unnemann\thanks{Most of this work was carried out
    when the author was affiliated to INRIA~Paris - Rocquencourt,
    France}%
} \affil{Department of Computer Science, TU Darmstadt, Germany}%

\else
\author{%
  \IEEEauthorblockN{Steve Kremer}
  \IEEEauthorblockA{INRIA~Nancy - Grand'Est \& Loria, France}
  \and 
  \IEEEauthorblockN{Robert K\"unnemann}
  \IEEEauthorblockA{Department of Computer
    Science, TU Darmstadt, Germany}%
  \IEEEauthorblockA{INRIA~Paris - Rocquencourt, France}%
}
\fi

\title{
Automated analysis of security protocols with global state%
\iffullversion \\ (Full version) %
\else
\thanks{The full version of this paper including all proofs is available
at \url{\downloadlink}} %
\fi }

\iffullversion \date{} \fi
\maketitle

%TODO put something here.
% \category{H.4}{Information Systems Applications}{Miscellaneous}
% %A category including the fourth, optional field follows...
% \category{D.2.8}{Software Engineering}{Metrics}[complexity measures, performance measures]
% \terms{Theory}
% \keywords{ACM proceedings, \LaTeX, text tagging}
% \author{Steve Kremer \and Robert K\"{u}nnemann}

\begin{abstract} 

  Security APIs, key servers and protocols that need to keep the
  status of transactions, require to maintain a global, non-monotonic
  state, e.g., in the form of a database or register. However, most
  existing automated verification tools do not support the analysis of
  such stateful security protocols -- sometimes because of fundamental
  reasons, such as the encoding of the protocol as Horn clauses, which
  are inherently monotonic.  A notable exception is the recent tamarin
  prover which allows specifying protocols as multiset rewrite (msr)
  rules, a formalism expressive enough to encode state. As multiset
  rewriting is a ``low-level'' specification language with no direct
  support for concurrent message passing, encoding protocols correctly
  is a difficult and error-prone process.

  We propose a process calculus which is a variant of the applied pi
  calculus with constructs for manipulation of a global state by
  processes running in parallel. We show that this language can be
  translated to msr rules whilst preserving all security properties
  expressible in a dedicated first-order logic for security
  properties. The translation has been implemented in a prototype tool
  which uses the tamarin prover as a backend. We apply the tool to
  several case studies among which a simplified fragment of PKCS\#11,
  the Yubikey security token, and an optimistic contract signing
  protocol.%
\end{abstract}

\section{Introduction}

Automated analysis of security protocols has been extremely
successful. Using automated tools, flaws have been for instance
discovered in the Google Single Sign On
Protocol~\cite{BreakingGoogle-FMSE2008}, in commercial security tokens
implementing the PKCS\#11 standard~\cite{BCFS-ccs10}, and one may also
recall Lowe's attack~\cite{Lowe1996} on the Needham-Schroeder public
key protocol 17 years after its publication. While efficient tools
such as ProVerif~\cite{blanchetcsfw01}, AVISPA~\cite{armando05AVISPA}
or Maude-NPA~\cite{MaudeNPA09} exist, these tools fail to analyze
protocols that require \emph{non-monotonic global state}, i.e., some
database, register or memory location that can be read and altered by
different parallel threads. In particular ProVerif, one of the most
efficient and widely used protocol analysis tools, relies on an
abstraction that encodes protocols in first-order Horn clauses. This
abstraction is well suited for the monotonic knowledge of an attacker
(who never forgets), makes the tool extremely efficient for verifying
an unbounded number of protocol sessions and allows to build on
existing techniques for Horn clause resolution.  However, Horn clauses
are inherently monotonic: once a fact is true it cannot be set to
false anymore.  As a result, even though ProVerif's input language, a
variant of the applied pi calculus~\cite{AF-popl01}, allows a priori
encodings of a global memory, the abstractions performed by ProVerif
introduce false attacks.
%\iffullversion 
In the ProVerif user manual~\cite[Section 6.3.3]{proverif} such an encoding of memory
cells and its limitations are indeed explicitly discussed:
% \begin{quote}
  {\it
  ``Due to the abstractions performed by ProVerif, such a cell is
  treated in an approximate way: all values written in the cell are
  considered as a set, and when one reads the cell, ProVerif just
  guarantees that the obtained value is one of the written values (not
  necessarily the last one, and not necessarily one written before
  the read).''}
%\end{quote}
%\fi
Some work~\cite{ARR-csf11,Modersheim-ccs10,DKRS-csf11} has
nevertheless used ingenious encodings of mutable state in Horn
clauses, but these encodings have limitations that we discuss below.

A prominent example where non-monotonic global state appears are
security APIs, such as the RSA PKCS\#11 standard~\cite{PKCS11}, IBM's
CCA~\cite{IBMCCAGuide} or the trusted platform module
(TPM)~\cite{TCG}. They have been known to be vulnerable to logical
attacks for some time~\cite{longley92automatic,bond01api} and formal
analysis has shown to be a valuable tool to identify attacks and find
secure configurations. % Several models have been proposed, some of which
% require manual inspection~\cite{Froschle2010Reasoning-with-}, which is
% tedious and error-prone.
One promising paradigm for analyzing security APIs is to regard them
as a participant in a protocol and use existing analysis
tools. However, Herzog~\cite{herzog06applying} already identified not
accounting for mutable global state as a major barrier to the
application of security protocol analysis tools to verify security
APIs.  Apart from security APIs many other protocols need to maintain
databases: key servers need to store the status of keys, in optimistic
contract signing protocols a trusted party maintains the status of a
contract, RFID protocols maintain the status of tags and more
generally websites may need to store the current status of
transactions.

\paragraph{Our contributions} 
We propose a tool for analyzing protocols that may involve
non-monotonic global state, relying on Schmidt {\it et al.}'s tamarin
tool~\cite{SMCB-csf12,SMCB-cav13} as a backend.  We designed a new
process calculus that extends the applied pi calculus by defining, in
addition to the usual constructs for specifying concurrent processes,
constructs for explicitly manipulating global state. This calculus
serves as the tool's input language.  The heart of our tool is a
translation from this extended applied pi calculus to a set of
multiset rewrite rules that can then be analyzed by tamarin which we
use as a backend. We prove the correctness of this translation and
show that it preserves all properties expressible in a dedicated first
order logic for expressing security properties. As a result, relying
on the tamarin prover, we can analyze protocols without bounding the
number of sessions, nor making any abstractions. Moreover it allows to
model a wide range of cryptographic primitives by the means of
equational theories.  As the underlying verification problem is
undecidable, tamarin may not terminate. However, it offers an
interactive mode with a GUI which allows to manually guide the tool in
its proof. Our specification language includes support for private
channels, global state and locking mechanisms (which are crucial to
write meaningful programs in which concurrent threads manipulate a
common memory). The translation has been carefully engineered in order
to favor termination by tamarin. We illustrate the tool on several
case studies: a simple security API in the style of PKCS\#11, a
complex case study of the Yubikey security device, as well as several
examples analyzed by other tools that aim at analyzing stateful
protocols. In all of these case studies we were able to avoid
restrictions that were necessary in previous works.

\paragraph{Related work} The most closely related work is the
StatVerif tool by Arapinis {\it et al.}~\cite{ARR-csf11}. They propose an
extension of the applied pi calculus, similar to ours, which is
translated to Horn clauses and analyzed by the ProVerif tool. Their
translation is sound but allows for false attacks, limiting the scope
of protocols that can be analyzed. Moreover, StatVerif can only handle
a finite number of memory cells: when analyzing an optimistic contract
signing protocol this appeared to be a limitation and only the status
of a single contract was modeled, providing a manual proof to justify
the correctness of this abstraction. Finally, StatVerif is limited to
the verification of secrecy properties. As illustrated by the Yubikey
case study, our work is more general and we are able to analyze
complex injective correspondance properties.

M\"odersheim~\cite{Modersheim-ccs10} proposed a language with support
for sets together with an abstraction where all objects that belong to
the same sets are identified. His language, which is an extension of
the low level AVISPA intermediate format, is compiled into Horn
clauses that are then analyzed, \eg,
using ProVerif. His approach is tightly
linked to this particular abstraction limiting the scope of
applicability. M\"odersheim also discusses the need for a more
high-level specification level which we provide in this work.

There has also been work tailored to particular
applications. In~\cite{DKS-jcs09}, Delaune {\it et al.}~show by a
dedicated hand proof that for analyzing PKCS\#11 one may bound the
message size. Their analysis still requires to artificially bound the
number of keys. Similarly in spirit, Delaune \emph{et
  al.}~\cite{DKRS-csf11} give a dedicated result for analyzing
protocols based on the TPM and its registers. However, the number of
reboots (which reinitialize registers) needs to be limited.

Guttman~\cite{Guttman-jar12} also extended the strand space model by
adding support for state. While the protocol execution is modeled
using the classical strand spaces model, state is modeled by a
multiset of facts, and manipulated by multiset rewrite rules.  The
extended model has been used for analyzing by hand an optimistic
contract signing protocol. As of now, protocol analysis in the strand
space model with state has not been mechanized yet.

In the goal of relating different approaches for protocol analysis
Bistarelli {\it et al.}~\cite{BCLM-jcs05} also proposed a translation from a
process algebra to multiset rewriting: they do however not consider
private channels, have no support for global state and assume that
processes have a particular structure. These limitations significantly
simplify the translation and its correctness proof. Moreover their
work does not include any tool support for automated verification.

Obviously any protocol that we are able to analyze can be directly
analyzed by the tamarin prover~\cite{SMCB-csf12,SMCB-cav13} as the
rules produced by our translation could have been given directly as an
input to tamarin. Indeed, tamarin has already been used for analyzing
a model of the Yubikey device~\cite{KS-stm12}, the case studies
presented with Mödersheim's abstraction, as well as those presented
with StatVerif. It is furthermore able to reproduce the aforementioned
results on PKCS\#11~\cite{DKS-jcs09} and the TPM~\cite{DKRS-csf11} --
moreover, it does so without bounding the number of keys, security
devices, reboots, etc.  Contrary to ProVerif, tamarin sometimes
requires additional \emph{typing lemmas} which are used to guide the
proof. These lemmas need to be written by hand (but are proved
automatically). In our case studies we also needed to provide a few
such lemmas manually.  In our opinion, an important disadvantage of
tamarin is that protocols are modeled as a set of multiset rewrite
rules. This representations is very low level and far away from actual
protocol implementations, making it very difficult to model a protocol
adequately.  Encoding private channels, nested replications and locking
mechanisms directly as multiset rewrite rules is a tricky and error
prone task. As a result we observed that, in practice, the protocol
models tend to be simplified. For instance, locking mechanisms are
often omitted, modeling protocol steps as a single rule and making
them effectively atomic. Such more abstract models may obscure issues
in concurrent protocol steps and increase the risk of implicitly
excluding attacks in the model that are well possible in a real
implementation, \eg, race conditions. Using a more high-level
specification language, such as our process calculus, arguably eases
protocol specification and overcomes some of these risks.

\section{Preliminaries}

\paragraph{Terms and equational theories} As usual in symbolic
protocol analysis we model messages by abstract terms.  Therefore we
define an order-sorted term algebra with the sort \msgsort and two
incomparable subsorts \pubsort and \freshsort.  For each of these
subsorts we assume a countably infinite set of names, \FN for fresh
names and \PN for public names. Fresh names will be used to model
cryptographic keys and nonces while public names model publicly known
values.  We furthermore assume a countably infinite set of variables
for each sort $s$, $\Vars_s$ and let \Vars be the union of the set of
variables for all sorts.  We write $u:s$ when the name or variable $u$
is of sort $s$. Let \Sign be a signature, i.e., a set of function
symbols, each with an arity. We write $f/n$ when function symbol $f$
is of arity $n$.  We denote by $\Terms_\Sign$ the set of well-sorted
terms built over \Sign, \PN, \FN and \Vars. For a term $t$ we denote
by $\names(t)$, respectively $\vars(t)$ the set of names, respectively
variables, appearing in $t$. The set of ground terms, i.e., terms
without variables, is denoted by $\Mess_\Sign$.  When $\Sigma$ is
fixed or clear from the context we often omit it and simply write
$\Terms$ for $\Terms_\Sign$ and $\Mess$ for $\Mess_\Sign$.

We equip the term algebra with an equational theory \ET, that is a
finite set of equations of the form $M=N$ where $M,N \in \Terms$. From
the equational theory we define the binary relation $=_\ET$ on terms,
which is the smallest equivalence relation containing equations in
$\ET$ that is closed under application of function symbols, bijective
renaming of names and substitution of variables by terms of the same
sort. Furthermore, we require $\ET$ to distinguish different fresh
names, \ie, $\forall a,b\in\FN: a\neq b \Rightarrow a \neq_\ET b$.

\begin{example}
  Symmetric encryption can be modelled using a signature $$\Sigma =
  \set{ \mathit{senc}/2,\allowbreak \mathit{sdec}/2,\allowbreak
    \mathit{encCor}/2,\allowbreak \mathit{true}/0 }$$ and an
  equational theory defined by
  $$
%  \begin{array}{l}
    \mathit{sdec}(\mathit{senc}(m,k),k) = m \quad 
    \mathit{encCor}(\allowbreak\mathit{senc}(x,y),y) = \mathit{true}
  %\end{array}
  $$ The last
  equation allows to check whether a term can be correctly decrypted
  with a certain key.
\end{example}

For the rest of the paper we assume that \ET refers to some fixed
equational theory and that the signature and equational theory always
contain symbols and equations for pairing and projection, i.e.,
$\{ \pair{.}{.}, \fst, \snd \} \subseteq \Sign$ and equations
$\fst(\pair x y) = x$ and $\snd(\pair x y) = y$ are in \ET.
We will sometimes use $\langle x_1, x_2, \ldots, x_n\rangle$ as
a shortcut for $\langle x_1,\langle x_2,\langle \ldots,\langle x_{n-1}, x_n\rangle\ldots\rangle$.

We also use the usual notion of positions for terms. A position $p$ is a sequence of positive integers and $t|_p$ denotes the subterm of $t$ at position $p$.

\paragraph{Facts} We also assume an unsorted signature \FSign,
disjoint from $\Sign$. The set of \emph{facts} is defined as \[ \calF
:= \{ F(t_1,\ldots,t_k) \mid t_i\in\Terms_\Sign, F\in\FSign \mbox{ of
  arity }k \}.\] Facts will be used both to annotate protocols, by the
means of events, and for defining multiset rewrite rules. We partition
the signature $\FSign$ into \emph{linear} and \emph{persistent} fact
symbols.  We suppose that \FSign always contains a unary, persistent
symbol $\K$ and a linear, unary symbol $\Fr$.  Given a sequence or set
of facts $S$ we denote by $\lfacts(S)$ the multiset of all linear
facts in $S$ and $\pfacts(S)$ the set of all persistent facts in
$S$. By notational convention facts whose identifier starts with `!'
will be persistent.  \GroundFacts denotes the set of ground facts,
i.e., the set of facts that does not contain variables. For a fact $f$
we denote by $\ginsts(f)$ the set of ground instances of $f$. This
notation is also lifted to sequences and sets of facts as expected.

\paragraph{Substitutions} A substitution $\sigma$ is a partial
function from variables to terms. We suppose that substitutions are
well-typed, i.e., they only map variables of sort $s$ to terms of sort
$s$, or of a subsort of s. We denote by $\sigma = \{ ^{t_1}/_{x_1},
\ldots , ^{t_n}/_{x_n} \}$ the substitution whose domain is
$\Dom(\sigma) = \{ x_1, \ldots , x_n \}$ and which maps $x_i$ to
$t_i$. As usual we homomorphically extend $\sigma$ to apply to terms
and facts and use a postfix notation to denote its application, e.g.,
we write $t\sigma$ for the application of $\sigma$ to the term $t$. A
substitution $\sigma$ is grounding for a term $t$ if $t\sigma$ is
ground.  Given function $g$ we let $g(x) = \bot$ when $x\not\in
\Dom(x)$. When $g(x) = \bot$ we say that $g$ is undefined for $x$.  We
define the function $f:=g[a\mapsto b]$ with $\Dom(f) = \Dom(g) \cup
\set a$ as $f(a):=b$ and $f(x):=g(x)$ for $x\neq
a$.% otherwise, even if $f$ or $g$ are not substitutions.

\paragraph{Sets, sequences and multisets} We write $\setN_n$ for the
set $\{1,\ldots, n\}$. Given a set $S$ we denote by $S^*$ the set of
finite sequences of elements from $S$ and by $S^\#$ the set of finite
multisets of elements from $S$. We use the superscript $^\#$ to
annotate usual multiset operation, e.g. $S_1 \cup^\# S_2$ denotes the
multiset union of multisets $S_1, S_2$. Given a multiset $S$ we denote
by $\mathit{set}(S)$ the set of elements in $S$. The sequence
consisting of elements $e_1, \ldots, e_n$ will be denoted by $[e_1,
\ldots, e_n]$ and the empty sequence is denoted by $[]$. We denote by
$|S|$ the length, i.e., the number of elements of the sequence. We use
$\cdot$ for the operation of adding an element either to the start or
to the end, e.g., $e_1 \cdot [e_2, e_3] = [e_1, e_2, e_3] = [e_1, e_2]
\cdot e_3$. Given a sequence $S$, we denote by $\mathit{idx}(S)$ the
set of positions in $S$, i.e., $\setN_n$ when $S$ has $n$ elements,
and for $i \in \mathit{idx}(S)$ $S_i$ denotes the $i$th element of the
sequence.  Set membership modulo $E$ is denoted by $\in_E$ and defined
as $e \in_E S$ if $\exists e' \in S.\ e' =_E e$. $\subset_\ET$ and
$=_\ET$ are defined for sets in a similar way. Application of
substitutions are lifted to sets, sequences and multisets as expected.
By abuse of notation we sometimes interpret sequences as sets or
multisets\iffullversion; the applied operators should make the implicit cast clear\fi.

\iffullversion
\section{A cryptographic pi calculus with explicit state}
\else
\section[A cryptographic pi calculus with explicit state]{\texorpdfstring{A
cryptographic pi calculus\\ with explicit state}{A cryptographic pi
calculus with explicit state}}
\fi
\label{sec:calculus}
%see http://tex.stackexchange.com/questions/14008/how-to-put-a-line-break-in-section-heading

\subsection{Syntax and informal semantics}

Our calculus is a variant of the applied pi calculus~\cite{AF-popl01}. In
addition to the usual operators for concurrency, replication, communication
and name creation, it offers several constructs for reading and updating an explicit global state. The grammar for processes is described in Figure~\ref{fig:syntax}.

\begin{figure}
    \begin{grammar}
      <M,N> ::= 
      $x,y,z \in \Vars$ 
      \alt $p\in\PN$
      \alt $n\in\FN$
      \alt $f$($M_1$,\ldots,$M_n$) ($f \in \Sigma$ of arity $n$)
      
            <$P$,$Q$> ::= 
      0
      \alt $P \mid Q$   
      \alt $!$ $P$
      \alt $\nu n$; $P$
      \alt out($M,N$); $P$
      \alt in($M,N$); $P$
      \alt if $M$=$N$ then $P$ [else $Q$]
      \alt event $F$ ; $P$ \quad ($F\in\calF$)
      \alt insert $M$,$N$; $P$
      \alt delete $M$; $P$
      \alt lookup $M$ as $x$ in $P$ [else $Q$]
      \alt lock $M$; $P$
      \alt unlock $M$; $P$
      \alt $ [L] \msrewrite A [R] ; P$ \quad ($L,R,A \in \calF^*$)
    \end{grammar}
\caption{Syntax}\label{fig:syntax}
\end{figure}

$0$ denotes the terminal process. $P \mid Q$ is the parallel
execution of processes $P$ and $Q$ and $! P$ the replication of $P$,
allowing an unbounded number of sessions in protocol executions. The
construct $\nu n; P$ binds the name $n$ in $P$ and models the
generation of a fresh, random value. Processes out($M,N$); $P$ and
in($M,N$); $P$ represent the output, respectively input, of message
$N$ on channel $M$. Readers familiar with the applied pi
calculus~\cite{AF-popl01} may note that we opted for the possibility
of pattern matching in the input construct, rather than merely binding
the input to a variable $x$. The process \code{if $M$=$N$ then $P$
  else $Q$} will execute $P$ if $M =_\ET N$ and $Q$ otherwise.
The event construct is merely used for annotating processes and % does
% not influence its execution. These events
will be useful for stating security properties. For readability we
sometimes omit to write \code{else $Q$} when $Q$ is 0, as well as
trailing 0 processes.

The remaining constructs are used for manipulating state and are new
compared to the applied pi calculus. We offer two different mechanisms
for state. The first construct is \emph{functional} and allows to
associate a value to a key. The construct insert $M$,$N$ binds the
value $N$ to a key $M$. Successive inserts allow to change this
binding. The delete $M$ operation simply ``undefines'' the mapping for
the key $M$. The lookup $M$ as $x$ in $P$ else $Q$ allows to retrieve
the value associated to $M$, binding it to the variable $x$ in $P$. If
the mapping is undefined for $M$ the process behaves as $Q$. The lock
and unlock constructs allow to gain exclusive access to a resource
$M$. This is essential for writing protocols where parallel processes
may read and update a common memory. We additionally offer another
kind of global state in form of a multiset of ground facts, as opposed
to the previously introduced functional store.  This multiset can be
altered using the construct $[L] \msrewrite A [R] ; P$, which tries to
match each fact in the sequence $L$ to facts in the current multiset
and, if successful, adds the corresponding instance of facts $R$ to
the store.  The facts $A$ are used as annotations in a similar way to
events.  The purpose of this construct is to provide access to the
underlying notion of state in tamarin, but we stress that it is
distinct from the previously introduced functional state, and its use
is only advised to expert users.  We allow this ``low-level'' form of
state manipulation in addition to the \emph{functional} state, as it
offers a great flexibility and has shown useful in one of our case
studies. This style of state manipulation is similar to the state
extension in the strand space model~\cite{Guttman-jar12} and the
underlying specification language of the tamarin
tool~\cite{SMCB-csf12,SMCB-cav13}.  Note that, even though those
stores are distinct (which is a restriction imposed by our
translation% , however, we see this construct as a low-level construct
% recommended only to experienced users
), data can be moved from one to
another, for example as follows: \apip{lookup 'store1' as $x$ in $[]
  \msrewrite{~} [\mathsf{store2}(x)]$}.
  
In the following example, which will serve as our running example, we
model a security API that, even though much simplified, illustrates
the most salient issues that occur in the analysis of
security APIs such as
PKCS\#11~\iffullversion\cite{DKS-jcs09,BCFS-ccs10,Froschle2010Reasoning-with-}
\else\cite{DKS-jcs09,BCFS-ccs10}\fi.

\begin{example}
%Put names of events here..
\lstset{morekeywords=[2]{NewKey,Wrap,DecUsing}}

We consider a security device that allows the creation of keys in its
secure memory. The user can access the device via an API. If he
creates a key, he obtains a handle, which he can use to let the device
perform operations on his behalf. For each handle the device also
stores an attribute which defines what operations are permitted for
this handle.  The goal is that the user can never gain knowledge of
the key, as the user's machine might be compromised. We model the
device by the following process (we use $\mathsf{out}(m)$ as a
shortcut for $\mathsf{out}(c,m)$ for a public channel $c$):
\[
	! P_\mathit{new}	\mid\  ! P_\mathit{set} \mid\  !P_\mathit{dec} 
	\mid\  !P_\mathit{wrap}\text{, where}
\]
\begin{lstlisting} 
$P_\mathit{new}:=$  new h; new k; event NewKey(h,k); 
        insert <`key',h>,k; 
        insert <`att',h>,`dec'; out(h) 
\end{lstlisting}
In the first line, the device creates a new handle $h$ and a key $k$
and, by the means of the event NewKey$(h,k)$, logs the creation of
this key. It then stores the key that belongs to the handle by
associating the pair $\pair {\mbox{`key'}} h$ to the value of the key
$k$. In the next line, $\pair {\mbox{`att'}} h$ is associated to a
public constant $\mbox{`dec'}$. Intuitively, we use the public
constants $\mbox{`key'}$ and $\mbox{`att'}$ to distinguish two
databases. The process
% $P_\mathit{set}:=$  \lstinline[columns=fullflexible,breaklines=true]{in(h); lock h; insert <'att',h>, 'wrap'; unlock h}
\begin{lstlisting}
$P_\mathit{set}:=$  in(h);  insert <`att',$h$>, `wrap'
\end{lstlisting}
allows the attacker to change the attribute of a key from the initial value
$\mbox{`dec'}$ to another value $\mbox{`wrap'}$. 
If a handle has the $\mbox{`dec'}$ attribute set, it can be used for decryption:
\begin{lstlisting}
$P_\mathit{dec}:=$ in(<h,c>); lookup <`att',h> as $a$ in 
        if $a$=`dec' then 
            lookup <`key',$h$> as $k$ in
                if encCor($c$,$k$)=true then
                    event DecUsing(k,sdec($c$,$k$));
                    out(sdec($c$,$k$)) 
\end{lstlisting}
The first lookup stores the value
associated to $\pair {\mbox{`att'}} h$ in $a$. The value is compared
against $\mbox{`dec'}$. If the comparison and another lookup for the
associated key value $k$ succeeds, we check whether decryption succeeds
and, if so, output the plaintext.

If a key has the $\mbox{`wrap'}$ attribute set, it might be used to
encrypt the value of a second key:
\begin{lstlisting}
$P_\mathit{wrap}:=$ in(<$h_1$,$h_2$>); lookup <`att',$h_1$> as $a_1$ in 
            if $a_1$=`wrap' then
                lookup <`key',$h_1$> as $k_1$ in
                    lookup <`key', $h_2$> as $k_2$ in
                        event Wrap($k_1$,$k_2$);
                        out(senc($k_2$,$k_1$)) 
\end{lstlisting}
\end{example}

The bound names of a process are those that are bound by $\nu n$. We
suppose that all names of sort \freshsort appearing in the process are
under the scope of such a binder.  Free names must be of sort
\pubsort. A variable $x$ can be bound in three ways:
\begin{inparaenum}[\it (i)]
\item by the construct \code{lookup $M$ as $x$}, or
\item $x \in \vars(N)$ in the construct \code{in($M,N$)} and $x$ is not
  under the scope of a previous binder,
\item $x \in \vars(L)$ in the construct $ [L] \msrewrite A [R]$ and $x$ is not
  under the scope of a previous binder.
\end{inparaenum}
While the construct \code{lookup $M$ as $x$} always acts  as a binder, the
input and $ [L] \msrewrite A [R]$ constructs do not rebind an already
bound variable but perform pattern matching.  For instance in the process
\[ P = \code{in(c,$f(x)$); in(c,$g(x)$)} \]
$x$ is bound by the first input and pattern matched in the second. 
It might seem odd that lookup acts as a binder, while input does not.
We justify this decision as follows: as $P_\mathit{dec}$ and
$P_\mathit{wrap}$ in the previous example show, lookups appear often after
input was received. If lookup were to use pattern matching, the following process
\[ P = \code{in($c,x$); lookup `$\mathrm{store}$' as $x$ in $P'$} \]
might unexpectedly perform a check if `store' contains the message given by
the adversary, instead of binding the content of `store' to $x$, 
due to an undetected clash in the naming of variables.

A process is ground if it does not contain any free variables. We
denote by $P\sigma$ the application of the homomorphic extension of
the substitution $\sigma$ to $P$. As usual we suppose that the
substitution only applies to free variables.
We sometimes interpret the syntax tree of a process as a term and
write $P|_p$ to refer to the subprocess of $P$ at position $p$ (where
$\mid$, \code {if} and \code{lookup} are interpreted as binary
symbols, all other constructs as unary).

\subsection{Semantics}

\paragraph{Frames and deduction} Before giving the formal semantics
of our calculus we introduce the notions of frame and deduction. A
\emph{frame} consists of a set of fresh names $\tilde n$ and a
substitution $\sigma$ and is written $\nu \tilde
n. \sigma$. Intuitively a frame represents the sequence of messages
that have been observed by an adversary during a protocol execution
and secrets $\tilde n$ generated by the protocol, a priori unknown to
the adversary. Deduction models the capacity of the adversary to
compute new messages from the observed ones.

\begin{definition}[Deduction]
  \label{def:apip-ded}
  We define the deduction relation $\nu \tilde n . \Subst \vdash t$ as
  the smallest relation between frames and terms defined by the
  deduction rules in \autoref{fig:deduc}.
  
\end{definition}

\begin{figure*}
  $$
  \begin{array}{c @{\qquad}c}
  \infer[ \textsc{Dname}]{ \nu \tilde n . \Subst \vdash a }{ a\in\FN\cup\PN & a\notin\tilde n }
  &
  \infer[ \textsc{DEq}]{\nu \tilde n . \Subst \vdash t' }{\nu \tilde n . \Subst \vdash t & t=_E t'}
  \\[2mm]
  \infer[ \textsc{DFrame}]{\nu \tilde n . \Subst \vdash x\Subst}{x \in
    \Dom(\Subst) }
  & 
  \infer[ \textsc{DAppl}]{\nu \tilde n . \Subst \vdash f(t_1,\ldots,t_n) }
  {\nu \tilde n . \Subst \vdash t_1  \cdots 
    \nu \tilde n . \Subst \vdash t_n & f\in\Sigma^k }
  \end{array}
  $$
  \caption{Deduction rules.}
  \label{fig:deduc}
\end{figure*}

\begin{example}
  If one key is used to wrap a second key, then, if the intruder learns
  the first key, he can deduce the second. For $\tilde n=k_1,k_2$ and
  $\sigma=\set{^{\mathit{senc}(k_2,k_1)}/_{x_1}, ^{k_1}/_{x_2}}$,
  $\nu\tilde n.\sigma\vdash k_2$, as witnessed by the proof tree given
  in Figure~\ref{fig:examplededuction}.
\end{example}

\begin{figure*}

\[
\infer
 {\nu \tilde n . \sigma \vdash k_2 }{  
\infer
{
\nu\tilde n.\sigma\vdash  \mathit{sdec}(\mathit{senc}(k_2,k_1),k_1)
}
{
\infer
{\nu\tilde n.\sigma\vdash  \mathit{senc}(k_2,k_1)
}
{ x_1\in\dom(\sigma)}
&
\infer
{\nu\tilde n.\sigma\vdash  k_1}
{ x_2\in\dom(\sigma)}
}
&
\mathit{sdec}(\mathit{senc}(k_2,k_1),k_1) =_\ET k_2
 }
\]
\caption{Proof tree witnessing that $\nu \tilde n . \sigma \vdash k_2$}
\label{fig:examplededuction}
\end{figure*}

\paragraph{Operational semantics} We can now define the operational
semantics of our calculus. The semantics is defined by a labelled transition
relation between process configurations. A \emph{process
  configuration} is a 6-tuple $(\Names, \StoreA, \StoreB, \Processes, \Subst, \ActiveLocks)$
where 
\begin{itemize}
\item $\Names \subseteq \FN$ is the set of fresh names generated by the processes;
\item $\StoreA: \Mess_\Sign \to \Mess_\Sign$ is a partial function
  modeling the functional store;
\item $\StoreB \subseteq \GroundFacts^\#$ is a multiset of ground
  facts and models the multiset of stored facts;
\item $\Processes$ is a multiset of ground processes representing the
  processes executed in parallel;
\item $\Subst$ is a ground substitution modeling the messages output
  to the environment;
\item $\ActiveLocks \subseteq  \Mess_\Sign$ is the set of currently
  acquired locks.
\end{itemize}

The transition relation is defined by the rules described in
Figure~\ref{fig:operationalsemantics}. Transitions are labelled by
sets of ground facts. For readability we omit empty sets and 
brackets around singletons, i.e., we write $\to$ for
$\stackrel{\emptyset}\longrightarrow$ and $\stackrel{f}\longrightarrow$ for $\stackrel{\set f}\longrightarrow$. We write $\to^*$
for the reflexive, transitive closure of $\to$ (the transitions that
are labelled by the empty sets) and write $\stackrel{f}
\Rightarrow$ for $\to^* \stackrel{f}\rightarrow \to^*$. We can now
define the set of traces, i.e., possible executions, that a process
admits.

\begin{definition}[Traces of $P$]
Given a ground process $P$ we define the \emph{set of traces of $P$} as
$$
\begin{array}{l}
 \tracespi(P) =  \left\{[F_1,\ldots,F_n] \mid 
     (\emptyset,\emptyset, \emptyset, \{ P\}, \emptyset, \emptyset) \phantom{\stackrel{F_1}{\Longrightarrow} }\right.\\
   \qquad      \stackrel{F_1}{\Longrightarrow} 
      (\Names_1,\StoreA_1, 
      \StoreB_1, \Processes_1, \Subst_1,
     \ActiveLocks_1)\\
      \qquad \left. 
        \stackrel{F_2}{\Longrightarrow} 
     \ldots
     \stackrel{F_n}{\Longrightarrow} 
     (\Names_n,\StoreA_n, \StoreB_n, \Processes_n, \Subst_n,  \ActiveLocks_n)
   \right\}
\end{array}
$$
\end{definition}

\begin{figure*}
  \iffullversion \small \fi
  \newlength{\addlinespace}
  \setlength{\addlinespace}{1mm}
  {\bf Standard operations:}
  $$
  \begin{array}{rcl}
    (\Names, \StoreA, \StoreB, \Processes \mcup \{0\}, \Subst, \ActiveLocks) 
    & \longrightarrow &  (\Names, \StoreA, \StoreB, \Processes , \Subst, \ActiveLocks) \\[\addlinespace]
    (\Names, \StoreA, \StoreB, \Processes \mcup \{P|Q\}, \Subst, \ActiveLocks) 
    &\longrightarrow&
    (\Names, \StoreA, \StoreB, \Processes \mcup \{P,Q\}, \Subst, \ActiveLocks) 
    \\[\addlinespace]
    (\Names, \StoreA, \StoreB, \Processes \mcup \{!P\}, \Subst, \ActiveLocks) 
    &\longrightarrow& 
    (\Names, \StoreA, \StoreB, \Processes \mcup \{!P, P\}, \Subst, \ActiveLocks) 
    \\[\addlinespace]
    (\Names, \StoreA, \StoreB, \Processes \mcup \{\nu a;P\}, \Subst, \ActiveLocks) 
    &\longrightarrow&
    (\Names\cup\{a'\}, \StoreA, \StoreB, \Processes \mcup \{P\{a'/a\}
    \}, \Subst, \ActiveLocks) \\
    \multicolumn{3}{r}{\text{if $a'$ is fresh}}
    \\[\addlinespace]
    (\Names, \StoreA, \StoreB, \Processes, \Subst, \ActiveLocks) 
    &\trans{K(M)}& 
    (\Names, \StoreA, \StoreB, \Processes, \Subst, \ActiveLocks) \quad
    {\text{if $\nu \Names.\Subst \vdash M$}}
    \\[\addlinespace]
    (\Names, \StoreA, \StoreB, \Processes \mcup \{\piout(M,N);P\},\Subst, \ActiveLocks) 
    &\trans{K(M)}&
    (\Names, \StoreA, \StoreB, \Processes \mcup \{P\},\Subst\cup \{^N / _x\}, \ActiveLocks ) 
    \\
    \multicolumn{3}{r}{\text{if
      $x$ is fresh and $\nu \Names. \Subst\vdash M$}}
    \\[\addlinespace]
    (\Names, \StoreA, \StoreB, \Processes \mcup \{\piin(M,N);P\},\Subst, \ActiveLocks) 
    &\trans{K(\pair{M}{N\tau})}&
    (\Names, \StoreA, \StoreB, \Processes \mcup \{P\tau \},\Subst, \ActiveLocks)  
    \\
    \multicolumn{3}{r} {\text{if }\exists \tau.\ \tau\text{ is
            grounding for } N, \nu \Names.\Subst\vdash M, \nu \Names.\Subst\vdash N\tau}
    \\[\addlinespace]
    % \multicolumn{1}{l}{(\Names, \StoreA, \StoreB,}
    % \\
    % \quad\quad
      (\Names, \StoreA, \StoreB, \Processes \mcup \{\piout(M,N);P, \piin(M',N');Q\},\Subst, \ActiveLocks) 
    % \begin{array}{l}
    %   (\Names, \StoreA, \StoreB, \Processes \mcup\\\quad \{\piout(M,N);P, \piin(M',N');Q\},\Subst, \ActiveLocks) \end{array}
    &\longrightarrow&
    (\Names, \StoreA, \StoreB, \Processes \cup \{P, Q \tau \},\Subst, \ActiveLocks)  
    \\
    \multicolumn{3}{r}{\text{if $M=_E M'$ and  $\exists \tau.$ $N =_E N'\tau$ and $\tau$ grounding for $N'$}}
    \\[\addlinespace]
    (\Names, \StoreA, \StoreB, \Processes \cup \{\text{if $M=N$ then $P$ else $Q$}\},\Subst, \ActiveLocks) 
    &\longrightarrow&
    (\Names, \StoreA, \StoreB, \Processes \cup \{P\}, \Subst, \ActiveLocks) 
    \quad\text{if $M=_\ET N$} 
    \\[\addlinespace]
    (\Names, \StoreA, \StoreB, \Processes \cup \{\text{if $M=N$ then $P$ else $Q$}\}, \Subst, \ActiveLocks) 
    &\longrightarrow& 
    (\Names, \StoreA, \StoreB, \Processes \cup \{Q\}, \Subst, \ActiveLocks)  
    \quad\text{if $M\neq_\ET N$}
    \\[\addlinespace]
    (\Names, \StoreA, \StoreB, \Processes \cup \{\text{event($F$); $P$}\}, \Subst, \ActiveLocks) 
    &\stackrel{F}{\longrightarrow}&
    (\Names, \StoreA, \StoreB, \Processes \cup \{P\}, \Subst, \ActiveLocks)
  \end{array}
$$
{\bf Operations on global state:}
$$
\begin{array}{rcl}
  (\Names, \StoreA, \StoreB, \Processes \mcup \{\text{insert $M,N$; $P$}\},
  \Subst, \ActiveLocks) 
  &\longrightarrow& 
  (\Names, \StoreA[M\mapsto N ], \StoreB, \Processes \mcup \{P\},  \Subst,
  \ActiveLocks)  
  \\[\addlinespace]
  (\Names, \StoreA, \StoreB, \Processes \mcup \{\text{delete $M$; $P$}\},
  \Subst, \ActiveLocks) 
  &\longrightarrow&
  (\Names, \StoreA[M \mapsto \bot], \StoreB, \Processes \mcup \{P\},  \Subst,
  \ActiveLocks)  % \\
  % \multicolumn{3}{r}{\text{ where }\StoreA'(x) = \begin{cases}
  %   \StoreA(x)  & \text{if  $x\neq_E M$} \\
  %   \bot  & \text{otherwise}
  % \end{cases}}
  \\[\addlinespace]
  (\Names, \StoreA, \StoreB, \Processes \mcup \{\text{lookup $M$ as $x$ in $P$
    else $Q$ }\},  \Subst, \ActiveLocks) 
  &\longrightarrow&
  (\Names, \StoreA, \StoreB, \Processes \mcup \{ P\{V/x\}\},  \Subst,
  \ActiveLocks)\\
  \multicolumn{3}{r}{\text{if $\StoreA(N)=_E V$ is defined and $N=_E M$}} \\[\addlinespace]
  (\Names, \StoreA, \StoreB, \Processes \mcup \{\text{lookup $M$ as $x$ in $P$
    else $Q$ }\},  \Subst, \ActiveLocks) 
  & \longrightarrow & 
  (\Names, \StoreA, \StoreB, \Processes \mcup \{ Q \},  \Subst,
  \ActiveLocks)\\
  \multicolumn{3}{r}{\text{if $\StoreA(N)$ is undefined for all $N=_E M$}}\\[\addlinespace]

  (\Names, \StoreA, \StoreB, \Processes \mcup \{\text{lock $M$; $P$}\},
  \Subst, \ActiveLocks) 
  &\longrightarrow&
  (\Names, \StoreA, \StoreB, \Processes \mcup \{P\},  \Subst, \ActiveLocks
  \cup \set{M})\\
  \multicolumn{3}{r}{\text{if $ M{\not\in}_\ET\ActiveLocks$}} \\[\addlinespace]
  % unlock
  (\Names, \StoreA, \StoreB, \Processes \mcup \{\text{unlock $M$; $P$}\},
  \Subst, \ActiveLocks )
  &\longrightarrow&
  (\Names, \StoreA, \StoreB, \Processes \mcup \{P\},  \Subst, \ActiveLocks
  \setminus \set{ M' \mid M'=_E M}
  )  \\[\addlinespace]

  (\Names, \StoreA, \StoreB, \Processes \mcup \{[ l \msrewrite a r ];\ 
    P\},  \Subst, \ActiveLocks) 
  &\stackrel{a'}{\longrightarrow}&
  (\Names, \StoreA, \StoreB \setminus \lfacts(l') \mcup r',
  \Processes \mcup \set{P\tau},  \Subst, \ActiveLocks) \\
  \multicolumn{3}{r}{
    \begin{array}{ll}\text{if }\exists
        \tau,l',a',r'. & \tau
        \text{ grounding for } l \msrewrite{a} r, l' \msrewrite{a'} r'
        =_E (l \msrewrite a r)\tau,\\
        & \lfacts(l') \msubseteq \StoreB, \pfacts(l')\subset
        \StoreB
    \end{array}}
  \end{array}
  $$
  \caption{Operational semantics}
  \label{fig:operationalsemantics}
\end{figure*}

\begin{example}
  In Figure~\ref{fig:transitions} we display the transitions that
  illustrate how the first key is created on the security device in
  our running example and witness that $[\mathrm{NewKey}(h',k')]\in\tracespi(P)$.
%   The following examples shows how the first key is created on the security device in our
%   running example.
% \begin{align*}
%  & (\emptyset, \emptyset, \emptyset, \mset{
% \underbrace{
% !P_\mathit{new}, 
% 	! P_\mathit{set} \mid !P_\mathit{dec} 
% 	\mid !P_\mathit{wrap}
% }_{=:\calP'}
% },
% \emptyset,  \emptyset )\\
% \rightarrow & (\emptyset, \emptyset, \emptyset, \mset{P_\mathit{new}}\mcup
% \calP',
% \emptyset,  \emptyset )\\
% \rightarrow & (\emptyset, \emptyset, \emptyset, \mset{ \inlineapip{ new h;
% new k; event NewKey(h,k);}\ldots}\\
% &\mcup \calP',
% \emptyset,  \emptyset )\\
% %
% \rightarrow^* & (\set{h',k'}, \emptyset, \emptyset, \mset{ \inlineapip{ event
% NewKey(h',k');}\ldots}\\
% & \mcup \calP',
% \emptyset,  \emptyset )
% \intertext{$  \trans{\mathrm{NewKey}(h',k')} $}
% %
% & (\set{h',k'}, \emptyset, \emptyset, \mset{ \inlineapip{ 
% insert <`key',h>,k;}\ldots}\\
% & \mcup \calP', \emptyset,  \emptyset )\\
% \rightarrow^* & (\set{h',k'}, \StoreA, \emptyset, \mset{
%  \inlineapip{out(h); 0}}\mcup \calP',
% \emptyset,  \emptyset )\\
% \rightarrow^* & (\set{h',k'}, \StoreA, \emptyset, \calP',
% \set{^{h'}/ _{x_1} },  \emptyset )\text{, where }
% \end{align*}
% $ \StoreA(\inlineapip{ <`key',h'>})=k'$ and
% $ \StoreA(\inlineapip{<`att',h'>})=\mbox{`dec'} $.
% Therefore, $[\mathrm{NewKey}(h',k')]\in\tracespi(P)$.
\end{example}

\begin{figure*}
  \begin{align*}
    & (\emptyset, \emptyset, \emptyset, \mset{
      \underbrace{
        !P_\mathit{new}, 
	! P_\mathit{set} \mid !P_\mathit{dec} 
	\mid !P_\mathit{wrap}
      }_{=:\calP'}
    },
    \emptyset,  \emptyset )
    \rightarrow (\emptyset, \emptyset, \emptyset, \mset{P_\mathit{new}}\mcup
    \calP',
    \emptyset,  \emptyset )\\
    \rightarrow\,\,\, & (\emptyset, \emptyset, \emptyset,
    \mset{ \inlineapip{ new {}}h;\inlineapip{ new {}} k; \inlineapip{ event NewKey}(h,k);
    \ldots}
     \mcup \calP',
     \emptyset,  \emptyset )\\
    \rightarrow^* & (\set{h',k'}, \emptyset, \emptyset, \mset{ \inlineapip{ event
        NewKey}(h',k');\ldots}
    \mcup \calP',
    \emptyset,  \emptyset )\\
    {\trans{\mathrm{NewKey}(h',k')}}\,\,\,&
    (\set{h',k'}, \emptyset, \emptyset, \mset{ \inlineapip{  insert
        <`key',}\ h'\inlineapip{>,}\ k';\ldots}
    \mcup \calP', \emptyset,  \emptyset )\\
    \rightarrow^* & (\set{h',k'}, \StoreA, \emptyset, \mset{
      \inlineapip{out}(h'); 0}\mcup \calP',
    \emptyset,  \emptyset )
    \rightarrow^*  (\set{h',k'}, \StoreA, \emptyset, \calP',
    \set{^{h'}/ _{x_1} },  \emptyset )\\
    & \text{
      where $ \StoreA(\inlineapip{<`key',} h' \inlineapip{>})=k'$ and
      $ \StoreA(\inlineapip{<`att',}h' \inlineapip{>})=\mbox{`dec'}$.
}
\end{align*}
\caption{Example of transitions modelling the creation of a key on
a PKCS\#11-like device}
\label{fig:transitions}
\end{figure*}

\section{Labelled multiset rewriting}

We now recall the syntax and semantics of labelled multiset rewriting
rules, which constitute the input language of the tamarin
tool~\cite{SMCB-csf12}.

\begin{definition}[Multiset rewrite rule]
  A labelled multiset rewrite rule $ri$ is a triple $(l,a,r)$,
  $l,a,r\in\calF^*$, written $l \msrewrite{a} r$. We call
  $l=\prems(\ri)$ the premises, $a=\actions(\ri)$ the actions, and
  $r=\conclusions(\ri)$ the conclusions of the rule.
\end{definition}

\begin{definition}[Labelled multiset rewriting system]
\label{def:msr-system}
  A labelled multiset rewriting system is a set of labelled multiset
  rewrite rules $R$, such that each rule $l\msrewrite a r\in R$ satisfies
  the following conditions:
  \begin{itemize}
  \item $l,a,r$ do not contain fresh names
  \item $r$ does not contain $\Fr$-facts
  \end{itemize}
  A labelled multiset rewriting system is called well-formed,
  if additionally 
  \begin{itemize}
  \item   for each $l'\msrewrite {a'} r' \in_\ET \ginsts(l\msrewrite
    a r)$ we have that $\cap_{r'' =_\ET r'}\names(r'') \cap \FN \subseteq
    \cap_{l'' =_\ET l'}\names(l'') \cap \FN$.
  \end{itemize}
\end{definition}

We define one distinguished rule \textsc{Fresh} which is the only rule
allowed to have $\Fr$-facts on the right-hand side
$$\textsc{Fresh}: [] \msrewrite{} [\Fr(x:\mathit{fresh})]$$

The semantics of the rules is defined by a labelled transition relation. 

\begin{definition}[Labelled transition relation] \label{def:transition-relation}

  Given a multiset rewriting system $R$ we define the \emph{labeled
    transition relation} $\rightarrow_{R} \subseteq \GroundFacts^{\#}
  \times \Processes(\GroundFacts) \times \GroundFacts^{\#}$ as 
  \[ S \stackrel{a}{\longrightarrow}_{R} ((S \setminus^\#
  \mathit{lfacts}(l)) \cup^\# r)\]
  if and only if $l \msrewrite a r
  \in_E \ginsts(R\cup\textsc{Fresh}) $, $\mathit{lfacts}(l)
  \subseteq^\# S$ and $ \mathit{pfacts}(l) \subseteq S $.
\end{definition}

\begin{definition}[Executions]
  \label{def:execution}
  Given a multiset rewriting system $R$ we define its set of
  executions as
  $$
  \begin{array}{l}
    \execmsr(R)=   \left\{ \emptyset \stackrel{A_1}{\longrightarrow}_{R} \ldots
      \stackrel{A_n}{\longrightarrow}_{R} S_n \mid  \right. \\
     \quad\forall a,i,j \colon 0 \leq i\neq j <n. \\
     ~\quad ( S_{i+1} \setminus^\# S_i)=\{\Fr(a)\} 
     \left. \Rightarrow (S_{j+1}
      \setminus^\# S_j)\neq\{\Fr(a) \} \right\}
  \end{array}
  $$
\end{definition}

The set of executions consists of transition sequences that respect
freshness, \ie, for a given name $a$ the fact $\Fr(a)$ is only added
once, or in other words the rule $\textsc{Fresh}$ is at most fired once
for each name.
We define the set of traces in a similar way as for processes.
\begin{definition}[Traces]
  The set of traces is defined as
  $$
  \begin{array}{l}
    \tracesmsr(R) =  
    \left\{\vphantom{\stackrel{A_1}{\Longrightarrow}_{R} }
      [A_1,\ldots,A_n] \mid \right.\   \forall\;  0 \leq i \leq
    n.\;  A_i \not= \emptyset \\ 
    \qquad\qquad
    \left. \text{ and }
      \emptyset
      \stackrel{A_1}{\Longrightarrow}_{R} 
      \ldots
      \stackrel{A_n}{\Longrightarrow}_{R} S_n \in \execmsr(R) \right\}
  \end{array}
  $$
  where $\stackrel{A}{\Longrightarrow}_{R}$ is defined as
  $\stackrel{\emptyset}{\longrightarrow}\!{}^*_{R}\stackrel{A}{\longrightarrow}\!{}_{R}\stackrel{\emptyset}{\longrightarrow}\!{}^*_{R}$.
\end{definition}
Note that both for processes and multiset rewrite rules the set of
traces is a sequence of sets of facts.

\section{Security Properties}

In the tamarin tool~\cite{SMCB-csf12} security properties are
described in an expressive two-sorted first-order logic. The sort
\temp is used for time points, $\Vars_\temp$ are the temporal
variables.

% \begin{definition}[Trace formulas] A trace atom is either
%   \begin{itemize}
%   \item false $\bot$,
%   \item a term equality $t_1 \approx t_2$,
%   \item a timepoint ordering $i \lessdot j$,
%   \item a timepoint equality $i \doteq j$, or
%   \item an action $F@i$ for a fact $F\in\calF$ and a timepoint $i$.
%   \end{itemize}
%   A trace formula is a first-order formula over trace atoms.
% \end{definition}

\begin{definition}[Trace formulas] A trace atom is either false
  $\bot$, a term equality $t_1 \approx t_2$, a timepoint ordering $i
  \lessdot j$, a timepoint equality $i \doteq j$, or an action $F@i$
  for a fact $F\in\calF$ and a timepoint $i$. A trace formula is a
  first-order formula over trace atoms. 
\end{definition}
As we will see in our case studies this logic is expressive enough to
analyze a variety of security properties, including complex injective
correspondence properties.

To define the semantics, let each sort $s$ have a domain $\Dom(s)$.
$\Dom(\tempsort)=\mathcal{Q}$,
$\Dom(\msgsort)=\Mess$,
$\Dom(\freshsort)=\FN$, and
$\Dom(\pubsort)=\PN$.
A function $\theta:\Vars \to \calM \cup \mathcal{Q}$ is a valuation if it
respects sorts, that is, $\theta(\Vars_s) \subset \Dom(s)$ for all sorts
$s$. If $t$ is a term, $t\theta$ is the application of the homomorphic
extension of $\theta$ to $t$.

\begin{definition}[Satisfaction relation] 
  The satisfaction relation $(\tr,\theta) \vDash \varphi$ between trace \tr, valuation $\theta$ and trace formula $\varphi$ is defined as follows:
  $$%\arraycolsep=1.4pt
  \begin{array} {lcl}
    (\tr,\theta) \vDash \bot & \multicolumn{2}{l}{\mbox{never}}\\
    (\tr,\theta) \vDash F@i & \mbox{iff} & 
    \theta(i) \in \mathit{idx}(tr)  \mbox{ and } F\theta  \in_E \tr_{\theta(i)}\\
    (\tr,\theta) \vDash i\lessdot j & \mbox{iff} &  \theta(i) < \theta(j)\\
    (\tr,\theta) \vDash i\doteq j & \mbox{iff} & \theta(i) = \theta(j)\\
    (\tr,\theta) \vDash t_1 \approx t_2 & \mbox{iff} & t_1\theta =_E t_2\theta\\
    (\tr,\theta) \vDash \neg \varphi & \mbox{iff} & \mbox{not } (\tr,\theta) \vDash \varphi\\
    (\tr,\theta) \vDash \varphi_1 \wedge \varphi_2 & \mbox{iff} &
    (\tr,\theta)\vDash \varphi_1 \mbox{ and } (\tr,\theta)\vDash \varphi_2\\
    (\tr,\theta) \vDash \exists x:s .\varphi & \mbox{iff} & \mbox{there is
} u\in \Dom(s) \mbox{ such that}\\ & & (\tr,\theta[x\mapsto u]) \vDash
    \varphi
  \end{array}
  $$
\end{definition}
%\iffullversion
When $\varphi$ is a ground formula we sometimes simply write $\tr \vDash
\varphi$ as the satisfaction of $\varphi$ is independent of the valuation.
%\fi
\begin{definition}[Validity, satisfiability]
Let $\mathit{Tr} \subseteq (\Processes(\GroundFacts))^*$ be a set of
traces. 
A trace formula $\varphi$ is said to be \emph{valid} for $\mathit{Tr}$, written
$\mathit{Tr}  \vDash^\forall \varphi$, if for any trace $\tr \in
\mathit{Tr}$ and any valuation $\theta$ we have that  $(\tr , \theta) \vDash
\varphi$.

A trace formula $\varphi$ is said to be \emph{satisfiable} for $\mathit{Tr}$, written
$\mathit{Tr}  \vDash^\exists \varphi$, if there exist a trace $\tr \in
\mathit{Tr}$ and a valuation $\theta$ such that  $(\tr , \theta) \vDash
\varphi$.

\end{definition}

Note that $\mathit{Tr} \vDash^\forall \varphi$ iff $\mathit{Tr}
\not\vDash^\exists \neg \varphi$. 
Given a multiset rewriting system
$R$ we say that $\varphi$ is valid,  written $R \vDash^\forall
\varphi$, if $\tracesmsr(R) \vDash^\forall \varphi$.  We say that
$\varphi$ is satisfied in $R$, written $R \vDash^\exists \varphi$, if
$\tracesmsr(R) \vDash^\exists \varphi$.  Similarly, given a ground
process $P$ we say that $\varphi$ is valid, written $P \vDash^\forall
\varphi$, if $\tracespi(P) \vDash^\forall \varphi$, and that $\varphi$
is satisfied in $P$, written $P \vDash^\exists \varphi$, if
$\tracespi(P) \vDash^\exists \varphi$.

\begin{example}\label{ex:property}
  The following trace formula expresses secrecy of keys generated on
  the security API, which we introduced in Section~\ref{sec:calculus}.
  % \[
  % \begin{array}{l}
  %   \neg ( \exists h,k \colon \msgsort,\ i,j \colon
  %   \mathit{temp} .\ \\
  %   \qquad\qquad\mathrm{NewKey}(h,k)@i \wedge \mathrm{K}(k)@j ) 
  % \end{array}
  % \]
  \[
    \neg ( \exists h,k \colon \msgsort,\ i,j \colon
    \mathit{temp} .\ \mathrm{NewKey}(h,k)@i \wedge \mathrm{K}(k)@j ) 
  \]

\end{example}

\section[A translation from processes into multiset rewrite rules]
{A translation from processes into multiset rewrite rules
  \label{sec:translation}}

In this section we define a translation from a process $P$ into a set
of multiset rewrite rules $\sem P$ and a translation on trace formulas
such that $P \models^\forall \varphi$ if and only if $\sem P
\models^\forall \sem \varphi$. Note that the result also holds for
satisfiability, as an immediate consequence. For a rather expressive
subset of trace formulas (see~\cite{SMCB-csf12} for the exact
definition of the fragment), checking whether $\sem P \models^\forall
\sem \varphi$ can then be discharged to the tamarin prover that we use
as a backend.

\subsection{Definition of the translation of processes}
To model the adversary's message deduction capabilities, we introduce
the  set of rules \textsc{MD} defined in Figure~\ref{fig:MD}.
% $$
% \begin{array}{rcl r}
%  \mathsf{Out}(x) & \msrewrite~ & \K(x) & (\textsc{MDOut})\\ 
%  \K(x) & \msrewrite{K(x)} & \mathsf{In}(x) & (\textsc{MDIn})\\ 
%   & \msrewrite{~} & \K(x:\mathit{pub}) &(\textsc{MDPub})\\ 
%   \Fr(x:\mathit{fresh})& \msrewrite{~} &
% \K(x:\mathit{fresh}) &(\textsc{MDFresh})
% \end{array}
% $$
% and
% $$
% \begin{array}{rcl r}
%    \K(x_1),\ldots, \K(x_k)& \msrewrite{~} & \K(f(x_1,\ldots,x_k)) \\
%  &&
%  \mbox{ for $f\in\Sign^k$} & (\textsc{MDAppl})\\ 
% \end{array}
% $$

\begin{figure*}
  $$
  \begin{array}{rclr}
    \mathsf{Out}(x) & \msrewrite{~} & \K(x) & (\textsc{MDOut})\\ 
    \K(x) & \msrewrite{K(x)}&  \mathsf{In}(x) & (\textsc{MDIn})\\ 
    & \msrewrite{~} & \K(x:\mathit{pub}) &(\textsc{MDPub})\\
    \Fr(x:\mathit{fresh}) & \msrewrite{~} & \K(x:\mathit{fresh}) &(\textsc{MDFresh})\\
    \K(x_1),\ldots, \K(x_k) &\msrewrite{~} & \K(f(x_1,\ldots,x_k)) 
     \text{ for $f\in\Sign^k$} & (\textsc{MDAppl})
  \end{array}
  $$
  \caption{The set of rules \textsc{MD}.}
  \label{fig:MD}
\end{figure*}
In order for our translation to be correct, we need to make some
assumptions on the set of processes we allow. These assumptions are
however, as we will see, rather mild and most of them without loss of
generality. First we define a set of reserved variables that will be
used in our translation and whose use we therefore forbid in the
processes.

\begin{definition}[Reserved variables and facts] % (fold)
  \label{def:reserved-names}
  The set of reserved variables is defined as the set containing the elements
  \begin{inparaitem}[]
  % \item $\mathit{sid}$,
  % \item $\mathit{rep}_i$ for any $i$,
  \item $\mathit{n}_a$ for any $a\in\FN$ and
  \item $\mathit{lock}_l$ for any $l\in\setN$.
  \end{inparaitem}
  
  The set of reserved facts $\ReservedFacts$ is defined as the set containing facts
  $f(t_1, \ldots, t_n)$ where $t_1, \ldots, t_n \in \Terms$ and $f \in 
  \{$ \begin{inparaitem}[]
    \item Init,
    \item Insert,
    \item Delete,
    \item IsIn,
    \item IsNotSet,
    \item state,
    % \item $\semistate$,
    \item Lock,
    \item Unlock,
    \item Out,
    \item Fr,
    \item In,
    \item Msg,
%    \item RepNonce,
    \item ProtoNonce,
    \item Eq,
    \item NotEq,
    \item Event,
    \item InEvent
%    \item K
    \end{inparaitem} $\}$.
  % definition reserved_names (end)
\end{definition}

Similar to~\cite{ARR-csf11}, for our translation to be sound, we
require that for each process, there exists an injective mapping
assigning to every \apip{unlock $t$} in a process a \apip{lock $t$}
that precedes it in the process' syntax tree.
% Similar to~\cite{ARR-csf11}, but slightly more general, we require for
% our translation that any \code{lock $t$} in a process is followed by
% at most one \code{unlock $t$} in each branch of the process' syntax
% tree.
Moreover, given a process \code{lock $t$; $P$} the corresponding
unlock in $P$ may not be under a parallel or replication. These
conditions allow us to annotate each corresponding pair \code{lock $t$},
\code{unlock $t$} with a unique label $l$. The annotated version
of a process $P$ is denoted $\overline P$. The formal definition of
$\overline P$ is given in Appendix~\ref{app:defs}. In case the annotation
fails, i.e., $P$ violates one of the above conditions, the process
$\overline P$ contains $\bot$.

% Given a ground process $P$, we formally define an annotation procedure in
% Definition~\ref{def:process-annotation} in the Appendix. It assigns a
% unique label $l$ to each \code{lock $t$} that occurs in the process tree, and
% annotates the closest \code{unlock $t$} (for the syntactically same $t$)
% with the BLA label $l$. This annotation might output a process that
% contains $\bot$, which means that the annotation failed. This is the case
% if \code{lock $t$} is not followed by \code{unlock $t$}, or vice versa, on
% some path to a leaf, or if there is a replication or a parallel after a \code{lock
% $t$} but before an \code{unlock $t$}.

\begin{definition}[well-formed]
\label{def:well-formed}
  A ground process $P$ is well-formed if
  \begin{itemize}
  \item no reserved variable nor reserved fact appear in $P$,
%  \item each replication $!^i$ has a unique identifier $i$,
  \item any name and variable in $P$ is bound at most once and
  \item $\overline P$ does not contain $\bot$.
  \item For each action $l \msrewrite a r$ that appears in the process, the
following holds: for each $l'\msrewrite {a'} r' \in_\ET \ginsts(l\msrewrite
    a r)$ we have that $\cap_{r'' =_\ET r'}\names(r'') \cap \FN \subseteq
    \cap_{l'' =_\ET l'}\names(l'')\cap \FN$.
  \end{itemize}
  A trace formula $\varphi$ is well-formed if no reserved variable nor reserved fact appear in $\varphi$.
\end{definition}
The two first restrictions of well-formed processes are not a loss of
generality as processes and formulas can be consistently renamed to
avoid reserved variables and $\alpha$-converted to avoid binding names or
variables several times. Also note that the second condition
is not necessarily preserved during an execution, e.g. when unfolding
a replication, $!P$ and $P$ may bind the same names. We only require
this condition to hold on the initial process for our translation to
be correct.

The annotation of locks restricts the set of protocols we can
translate, but allows us to obtain better verification results, since
we can predict which \code{unlock} is ``supposed'' to close a given
\code{lock}. This additional information is helpful for tamarin's
backward reasoning. We think that our locking mechanism captures all
practical use cases. Using our calculus' ``low-level'' multiset
manipulation construct, the user is also free to implement locks
himself, e.g., as
\[ [\mathrm{NotLocked}()] \rightarrow []; \mathit{code}; []
\rightarrow [\mathrm{NotLocked}()]\] (In this case the user does not
benefit from the optimisation we put into the translation of locks.)
Obviously, locks can be modelled both in tamarin's multiset rewriting
calculus (this is actually what the translation does) and Mödersheim's
set rewriting calculus~\cite{Modersheim-ccs10}. However, protocol
steps typically consist of a single input, followed by several
database lookups, and finally an output. In practice, they tend to be
modelled as a single rule, and are therefore atomic.  Real
implementations are however different, as several entities might be
involved, database lookups could be slow, etc. In this case, such
simplified
models could, \eg, miss race conditions. To the best of our knowledge,
StatVerif is the only comparable tool that models locks explicitly and
it has stronger restrictions.% It has the same
% restrictions, but supports only a single, global lock,  \ie, every
% \code{lock} must be followed by precisely one corresponding \code{unlock}
% in every branch of the syntax tree, and in a process \code{lock;$P$}, the
% part of the process $P$ that occurs before the next unlock, may not include
% parallel, replication, or lock. 

\begin{definition}
  \label{def:transprocess}
  Given a well-formed ground process $P$ we define the labelled
  multiset rewriting system $\sem P$ as
  $$\textsc{MD} \cup \{\textsc{Init} \}  \cup
  \sem{\overline P,[],[]}$$

  \begin{itemize}
  \item   where the rule \textsc{Init} is defined as 
    $$\textsc{Init}:  [] \msrewrite{\mathrm{Init}()} [\state_{[]}()]$$

  \item $\sem{P,p,\tilde x}$ is defined inductively for process $P$,
    position $p \in \setN^*$ and sequence of variables $\tilde x$ in Figure~\ref{fig:transprocess}.
  \item For a position $p$ of $P$ we define $\state_p$ to be
      persistent if ${P|_p} =\ !Q$ for some process $Q$;
    otherwise $\state_p$ is linear.
  \end{itemize}
\end{definition}

\begin{figure*}
  \renewcommand{\arraystretch}{1.3}
  \renewcommand{\theactualrule}[1]{\{ #1 \}}
  \renewcommand{\underscorethingy}{}

  \input{def-trans.tex}
  \caption{Translation of processes: definition of $\sem{P,p,\tilde x}$}\label{fig:transprocess}
\end{figure*}

In the definition of $\sem{P,p,\tilde x}$ we intuitively use the
family of facts $\state_p$ to indicate that the process is currently
at position $p$ in its syntax tree. A fact $\state_p$ will indeed be
true in an execution of these rules whenever some instance of $P_p$
(i.e. the process defined by the subtree at position $p$ of the syntax
tree of $P$) is in the multiset $\Processes$ of the process
configuration.
The translation of the zero-process, parallel and replication
operators merely use $\state_p$-facts. For instance $\sem{P \mid Q, p,
  \tilde x}$ defines the rule
$$
[\state_p(\tilde x)] \to [\state_{p\cdot 1}(\tilde x), \state_{p\cdot 2}(\tilde x)]\\
$$
which intuitively states that when a process is at position $p$
(modelled by the fact $\state_p(\tilde x)$ being true) then the
process is allowed to move both to $P$ (putting $\state_{p\cdot
  1}(\tilde x)$ to true) and $Q$ (putting $\state_{p\cdot 2}(\tilde
x)$ to true). The translation of $\sem{P \mid Q, p, \tilde x}$ also
contains the set of rules $\sem{P, p\cdot 1, \tilde x}\cup \sem{Q,
  p\cdot 2, \tilde x}$ expressing that after this transition the
process may behave as $P$ and $Q$, i.e., the processes at positions
$p\cdot 1$, respectively $p\cdot 2$, in the process tree.  Also note
that the translation of $!P$ results in a persistent fact as $!P$
always remains in $\Processes$.
The translation of the construct $\nu\,a$ translates the name $a$ into
a variable $n_a$, as msr rules must not contain fresh names.  Any
instantiation of this rule will substitute $n_a$ by a fresh name,
which the \Fr-fact in the premise guarantees to be new. This step is
annotated with a (reserved) action \ProtoNonce, used in the proof of
correctness to distinguish adversary and protocol nonces. Note that
the fact $\state_{p\cdot1}$ in the conclusion carries $n_a$, so that
the following protocol steps are bound to the fresh name used to
instantiate $n_a$.  The first rules of the translation of \apip{out}
and \apip{in} model the communication between the protocol and the
adversary, and vice versa.  In the case of \apip{out}, the adversary
must know the channel $M$, modelled by the fact $\In(M)$ in the rule's
premisse, and learns the output message, modelled by the fact
$\Out(N)$ in the conclusion. In the case of \apip{in}, the knowledge of
the message $N$ is additionally required and the variables of the
input message are added to the parameters of the $\state$ fact to
reflect that these variables are bound. The second and third rules of
the translations of \apip{out} and \apip{in} model an internal
communication, which is synchronous. For this reason, when the second
rule of the translation of \apip{out} is fired, the \state-fact is
substituted by an intermediate, \emph{semi-state} fact, $\semistate$,
reflecting that the sending process can only execute the next step if
the message was successfully received. The fact $\Msg(M,N)$ models
that a message is present on the synchronous channel.  Only with the
acknowledgement fact $\Ack(M,N)$, resulting from the second rule of
the translation of \apip{in}, is it possible to advance the execution
of the sending process, using the third rule in the translation of
\apip{out}, which transforms the semi-state \emph{and} the
acknowledgement of receipt into $\state_{p\cdot 1}(\ldots)$. Only now
the next step in the execution of the sending process can be
executed. The remaining rules essentially update the position in the
$\state$ facts and add labels.  Some of these labels are used to
restrict the set of executions. For instance the label Eq($M$,$N$)
will be used to indicate that we only consider executions in which $M
=_\ET N$. As we will see in the next section these restrictions will
be encoded in the trace formula.

%\iffullversion %{{{
\begin{figure*}

  $$\arraycolsep=1.4pt
  \begin{array}{rcl}
      [ ] & \msrewrite{\mathrm{Init}()} & [\state_{[]}()]\\
      {[ \state_{[]}()] } & \msrewrite{~}  & [!\state_{[1]}()]\\
      {[ !\state_{[1]}(), \Fr(h)]}  & \msrewrite{~}  & [\state_{[11]}(h)]\\
      {[ \state_{[11]}(h), \Fr(k)] } & \msrewrite{~}  & [\state_{[111]}(k, h)]\\
    {[ \state_{[111]}(k, h)]}  & \msrewrite{\mathrm{Event}(), \mathrm{NewKey}(h,
  k)}  & [\state_{[1111]}(k, h)]\\
    {[ \state_{[1111]}(k, h)]}  & \msrewrite{\mathrm{Insert}(\langle\mbox{\small 'key'},
  h\rangle, k)} & [\state_{[11111]}(k, h)]\\
    {[ \state_{[11111]}(k, h)]}  & \msrewrite{\mathrm{Insert}(\langle\mbox{\small
        'att'}, h\rangle, \mbox{\small 'dec'})}  & [\state_{[111111]}(k, h)]\\
      {[ \state_{[111111]}(k, h)]}  & \msrewrite{~} & [\Out(h), \state_{[1111111]}(k,h)]
  \end{array}
  $$
  \caption{The set of multiset rewrite rules $\sem{!P_\mathit{new}}$ (omitting
  the rules in \textsc{MD})}
  \label{fig:ex-trans}
\end{figure*}

\begin{example}
  Figure~\ref{fig:ex-trans} illustrates the above translation by
  presenting the set of msr rules $\sem{!P_\mathit{new}}$ (omitting
  the rules in \textsc{MD} already shown in Figure~\ref{fig:MD}).
% \begin{example}
%  $\sem{!P_\mathit{new}}$ gives the following set of rules:
% \begin{align*}
%  [ ] & \msrewrite{\mathrm{Init}()} [\state_{0}()]\\
%  [ \state_{0}()] & \msrewrite{} [!\state_{01}()]\\
%  [ !\state_{01}(), \Fr(h)] & \msrewrite{} [\state_{011}(h)]\\
%  [ \state_{011}(h), \Fr(k)] & \msrewrite{} [\state_{0111}(k, h)]\\
%  [ \state_{0111}(k, h)] & \msrewrite{\mathrm{Event}(), \mathrm{NewKey}(h,
% k)}\\
% &\qquad  [\state_{01111}(k, h)]\\
%  [ \state_{01111}(k, h)] & \msrewrite{\mathrm{Insert}(<\mbox{\small 'key'},
% h>, k)}\\
% &\qquad [\state_{011111}(k, h)]\\
%  [ \state_{011111}(k, h)] & \msrewrite{\mathrm{Insert}(<\mbox{\small
% 'att'}, h>, \mbox{\small 'dec'})}\\
% & \qquad [\state_{0111111}(k, h)]\\
%  [ \state_{0111111}(k, h)] & \msrewrite{} [\Out(h), \\
%  &  \qquad\qquad \state_{01111111}(k,h)]
% \end{align*}
  
  A graph representation of an example trace, generated by the
  tamarin tool, is depicted in Figure~\ref{fig:example-trace}.  Every
  box in this picture stands for the application of a multiset rewrite
  rule, where the premises are at the top, the conclusions at the
  bottom, and the actions (if any) in the middle. Every premise needs
  to have a matching conclusion, visualized by the arrows, to ensure
  the graph depicts a valid msr execution.  (This is a simplification
  of the dependency graph representation tamarin uses to perform
  backward-induction~\cite{SMCB-csf12,SMCB-cav13}.) Note that the
  machine notation for $\state_p()$ predicates omits brackets $[\, ]$
  in the position $p$ and denotes the empty sequence by `0'. We also
  note that in the current example $ !\state_{[1]}()$ is persistent
  and can therefore be used multiple times as a premise. As $\Fr(\, )$
  facts are generated by the $\textsc{Fresh}$ rule which has an empty
  premise and action, we omit instances of \textsc{Fresh} and 
  leave those premises, but only those, disconnected.

\begin{figure}[h] % (fold)
\centering

%\fbox{
% \includegraphics[scale=0.65,
% trim=1.7cm 1.5cm 1.5cm 1.3cm]
% {gfx/example-double.pdf}
% % }
% \caption{ Example trace for the translation of $!P_\mathit{new}$. }

\includegraphics[scale=0.65]
{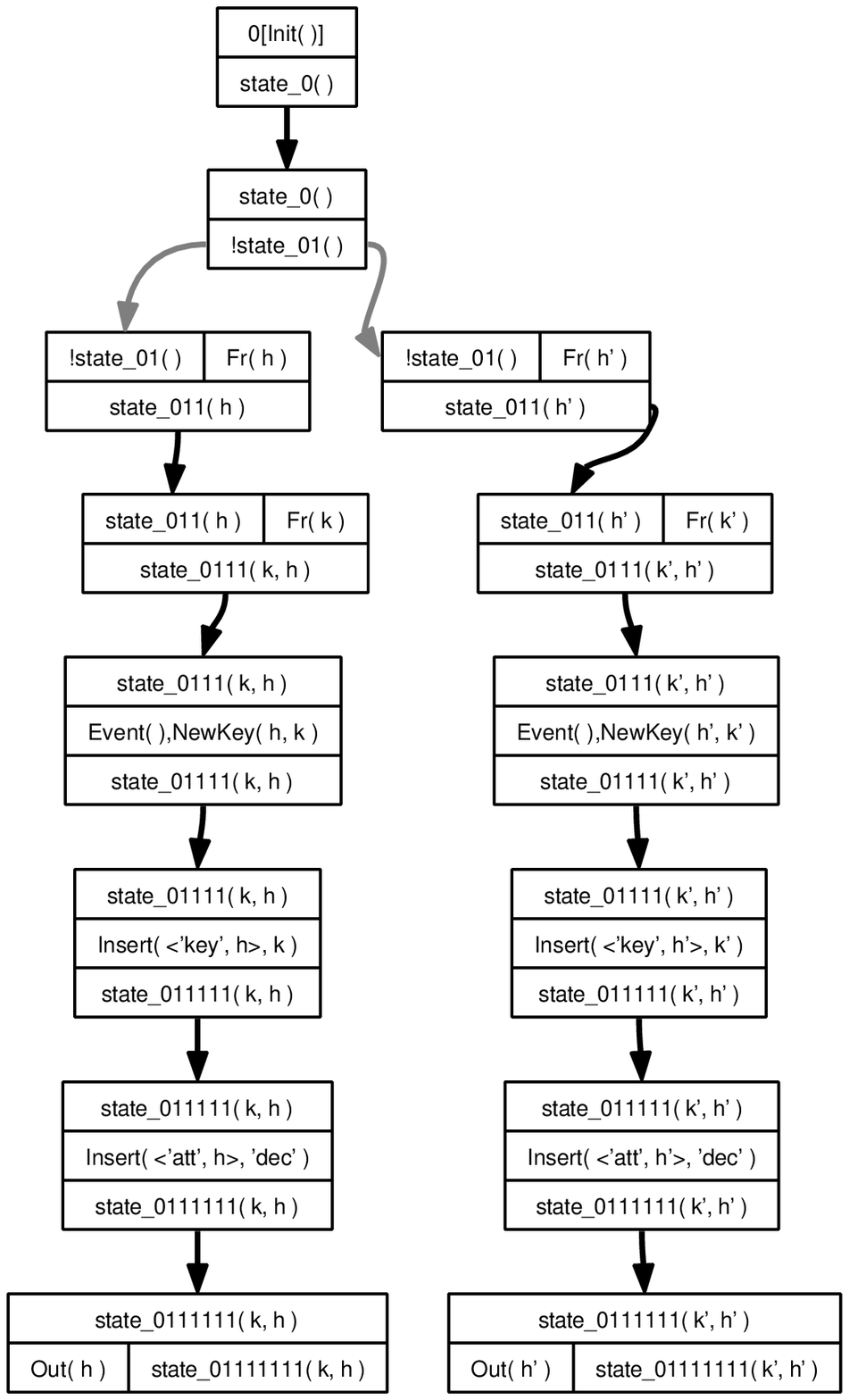}
% }
\caption{ Example trace for the translation of $!P_\mathit{new}$. }

\label{fig:example-trace}
% figure example-trace (end)
\end{figure}
\end{example} 
% \else
% \begin{example}
%   The set $\sem{!P_\mathit{new}}$ contains the following msr rules (omitting
%   the rules in \textsc{MD} already shown in Figure~\ref{fig:MD}):

%   $$\arraycolsep=0.8pt
%   \begin{array}{rcl}
%       [ ] & \msrewrite{\mathrm{Init}()} & [\state_{[]}()]\\
%       {[ \state_{[]}()] } & \msrewrite{~}  & [!\state_{[1]}()]\\
%       {[ !\state_{[1]}(), \Fr(h)]}  & \msrewrite{~}  & [\state_{[11]}(h)]\\
%       {[ \state_{[11]}(h), \Fr(k)] } & \msrewrite{~}  & [\state_{[111]}(k, h)]\\
%           {[ \state_{[111]}(k, h)]}  & \msrA\mathrm{Event}(), \\
%                                      & ~ \mathrm{NewKey}(h, k)\msrB  & [\state_{[1111]}(k, h)]\\
%      {[ \state_{[1111]}(k, h)]}  & \msrA\mathrm{Insert}(\langle\mbox{\small 'key'},\\
%                                  & h\rangle, k)\msrB & [\state_{[11111]}(k, h)]\\
%     {[ \state_{[11111]}(k, h)]}  & \msrA\mathrm{Insert}(\langle\mbox{\small
%         'att'}, h\rangle, \\
%         & \mbox{\small 'dec'})\msrB  & [\state_{[111111]}(k, h)]\\
%       {[ \state_{[111111]}(k, h)]}  & \msrewrite{~} & [\Out(h),\\
%         && ~ \state_{[1111111]}(k,h)]
%   \end{array}
%   $$
% \end{example}
% \fi %}}}

\begin{remark}
  One may note that, while for all other operators, the translation
  produces well-formed multiset rewriting rules (as long as the
  process is well-formed itself), this is not the case for the
  translation of the \code{lookup} operator, \ie, it violates the
  well-formedness condition from Definition~\ref{def:msr-system}.
  Tamarin's constraint solving algorithm requires all rules, with the
  exception of \textsc{Fresh}, to be well-formed.  We show however
  that, under these specific conditions, the solution procedure is
  still correct. See \iffullversion Appendix~\ref{app:solution} for
  the proof.  \else Appendix~C in the full version for the proof.  \fi
\end{remark}

\subsection{Definition of the translation of trace formulas}

We can now define the translation for formulas. 

\begin{definition}
  \label{def:transformula}
  Given a well-formed trace formula $\varphi$ we define 
  \begin{equation*}
  {\sem \varphi}_\forall := \Ass \Rightarrow \varphi \qquad \mbox{
    and } \qquad {\sem \varphi}_\exists := \Ass \wedge \varphi 
  \end{equation*}
  where $\Ass$ is defined in Figure~\ref{fig:def-ass}.
\end{definition}

\begin{figure*}
  $\Ass := \AssInit \wedge \AssEq \wedge \AssNotEq \wedge \AssSetIn \wedge
  \AssSetNotIn \wedge \AssLock  \wedge \AssIn$ and
  \begin{alignat*}{3}
      \AssInit :=&
      \forall i,   j.\ && \mbox{Init}()@i \wedge \mbox{Init}()
      @ j \Longrightarrow i=j \displaybreak[0]\\
      \AssEq :=& \forall x, y, i.\ && \mbox{Eq}(x,y) @ i \Longrightarrow
      x \approx y  \displaybreak[0]\\
      \AssNotEq :=& \forall x, y, i.\ &&\mbox{NotEq}(x,y) @ i \Longrightarrow
      \neg (x \approx y) \displaybreak[0]\\
      \AssSetIn :=& \forall x,y,t_3.\ &&\mbox{IsIn}(x,y)@t_3 \Longrightarrow
      \exists t_2 . \ \mbox{Insert}(x,y)@t_2 \wedge  t_2\lessdot t_3
      \displaybreak[0]\\
      &&& \phantom{\mbox{IsIn}(x,y)@t_3 \Longrightarrow
        \exists t_2 . \ }\wedge  \forall t_1,y .\ \mbox{Insert}(x,y)@t_1 \Longrightarrow (t_1\lessdot t_2
      \vee t_1 \doteq t_2 \vee t_3 \lessdot t_1 )\\
    &&& \phantom{\mbox{IsIn}(x,y)@t_3 \Longrightarrow
      \exists t_2 . \ } \wedge  \forall t_1.\phantom{,y} \ \mbox{Delete}(x)@t_1 \Longrightarrow (t_1\lessdot t_2
    \vee t_3 \lessdot t_1 ) \\
    \AssSetNotIn :=&\forall x,y,t_3.\ &&\mbox{IsNotSet}(x)@t_3 
    \Longrightarrow 
    (\forall t_1,y.\  \mbox{Insert}(x,y)@t_1 \Longrightarrow t_3
    \lessdot t_1  ) \vee\\
    &&& \phantom{\mbox{IsNotSet}(x)@t_3 
    \Longrightarrow\; } (\exists t_1 .\ \mbox{Delete}(x)@t_1 \wedge t_1 \lessdot t_3\\
    &&& \phantom{\mbox{IsNotSet}(x)@t_3 
    \Longrightarrow  (\exists t_1 .\ }\wedge 
    \forall t_2,y .\ (\mbox{Insert}(x,y)@t_2 \wedge
        t_2\lessdot t_3) \Longrightarrow t_2 \lessdot t_1)\displaybreak[0]\\
        \AssLock :=& \forall x,l,l',i,j .\ &&\mathrm{Lock}(l,x)@i
    \wedge \mathrm{Lock}(l',x)@j \wedge i\lessdot j\\
    &&& \qquad \qquad 
    \Longrightarrow 
    \exists k.\ \mathrm{Unlock}(l,x)@k \wedge i\lessdot k \wedge
    k\lessdot j \\
    &&& \phantom{ \qquad \qquad 
    \Longrightarrow 
    \exists k.\ } \wedge 
    (\forall l', m .\ \mathrm{Lock}(l',x)@m \Longrightarrow \neg(i\lessdot m
    \wedge m\lessdot k)) \\
    &&& \phantom{ \qquad \qquad 
    \Longrightarrow 
    \exists k.\ } \wedge 
    (\forall l', m .\ \mathrm{Unlock}(l',x)@m \Longrightarrow
    \neg(i\lessdot m \wedge m \lessdot k)) \displaybreak[0]\\
    \AssIn :=& \forall t,i .\ &&\ChannelInEvent(t)@i
    \Longrightarrow  \exists j.\ \mathrm{K}(t)@j 
    \wedge  (\forall k .\ \Event()@k \Longrightarrow (k\lessdot j
    \vee i\lessdot k))\\
    &&&\phantom{\ChannelInEvent(t)@i
      \Longrightarrow  \exists j.\ }
    \wedge (\forall k, t'.\ \mathrm{K}(t')@k \Longrightarrow (k\lessdot j
    \vee i\lessdot k \vee k\approx j)) 
  \end{alignat*}
  \caption{Definition of $\Ass$.}
  \label{fig:def-ass}
\end{figure*}

The formula $\alpha$ uses the actions of the generated rules to filter
out executions that we wish to discard:
\begin{itemize}

\item $\AssInit$ ensures that the init rule is only fired once.

\item $\AssEq$ and $\AssNotEq$ ensure that we only consider traces where all (dis)equalities hold.

\item $\AssSetIn$ and $\AssSetNotIn$ ensure that a successful
  lookup was preceded by an insert that was neither revoked nor
  overwritten while an unsuccessful lookup was either never
  inserted, or deleted and never re-inserted.

\item $\AssLock$ checks that between each two matching locks there
  must be an unlock. Furthermore, between the first of these locks and
  the corresponding unlock, there is neither a lock nor an unlock.

\item $\AssIn$ ensures that whenever an instance of \textsc{MDIn} is
  required to generate an \In-fact, it is generated as late as
  possible, \ie, there is no visible event between the action $K(t)$
  produced by \textsc{MDIn}, and a rule that requires $\In(t)$.
\end{itemize}
We also note that $\Tr \vDash^\forall {\sem \varphi}_\forall$ iff  $\Tr
\not\vDash^\exists {\sem {\neg \varphi}}_\exists$.

The axioms in the translation of the formula are designed to work hand
in hand with the translation of the process into rules. They express
the correctness of traces with respect to our calculus' semantics, but
are also meant to guide tamarin's constraint solving algorithm.
$\AssSetIn$ and $\AssSetNotIn$ illustrate what kind of axioms work
well: when a node with the action \textsf{IsIn} is created, by
definition of the translation, this corresponds to a \apip{lookup}
command. The existential translates into a graph constraint that
postulates an insert node for the value fetched by the lookup, and
three formulas assuring that
\begin{inparaenum}[\itshape a\upshape)]
\item this insert node appears before the lookup,
\item is uniquely defined, \ie, it is the last
  insert to the corresponding key, and
\item there is no delete in between.
\end{inparaenum}
Due to these conditions, \AssSetNotIn only adds one \textsf{Insert}
node per \textsf{IsIn} node -- the case where an axiom postulates a
node, which itself allows for postulating yet another node needs to be
avoided, as tamarin runs into loops otherwise.  Similarly, a na\"{i}ve
way of implementing locks using an axiom would postulate that every
lock is preceeded by an unlock and no lock or unlock in between,
unless it is the first lock. This again would cause tamarin to loop,
because an unlock is typically preceeded by yet another lock.
  % , because an unlock node
  % postulated in the first case is typically preceeded by a lock --
  % to
  % which the axiom applies again.  
The axiom \AssLock avoids this caveat because it only applies to pairs
of locks carrying the same annotations.
  % Furthermore, possible pairings are restricted to those
  % carrying the same annotations.

\iffullversion\else The interaction between the \AssLock axiom and
tamarin's constraint solving algorithm is described in more detail in
the full version. \fi \iffullversion We will outline how \AssLock is
applied during the constraint solving procedure:
\begin{enumerate}
\item If there are two locks for the same term and with possibly
  different annotations, an unlock for the first of those locks is
  postulated, more precisely, an unlock with the same term, the same
  annotation and no lock or unlock for the same term in-between.  The
  axiom itself contains only one case, so the only case distinction
  that takes place is over which rule produces the matching
  $\mathsf{Unlock}$-action. However, due to the annotation, all but
  one are refuted immediately in the next step. Note further that
  \AssLock postulates only a single node, namely the node with the
  action $\mathsf{Unlock}$.

\item Due to the annotation, the fact $\state_p(\ldots)$ contains the
  fresh name that instantiates the annotation variable. Let
  $a:\freshsort$ be this fresh name. Every fact $\state_{p'}(\ldots)$
  for some position $p'$ that is a prefix of $p$ and a suffix of the
  position of the corresponding lock contains this fresh name.
  Furthermore, every rule instantiation that is an ancestor of a node
  in the dependency graph corresponds to the execution of a command
  that is an ancestor in the process tree. Therefore, the backward
  search eventually reaches the matching lock, including the
  annotation, which is determined to be $a$, and hence appears in the
  \Fr-premise.

\item Because of the \Fr-premise, any existing subgraph that already
  contains the first of the two original locks would be merged with
  the subgraph resulting from the backwards search that we described
  in the previous step, as otherwise $\Fr(a)$ would be added at two
  different points in the execution.

\item The result is a sequence of nodes from the first lock to the
  corresponding unlock, and graph constraints restricting the second
  lock to not take place between the first lock and the unlock. We
  note that the axiom \AssLock is only instantiated once per pair of
  locks, since it requires that $i\lessdot j$, thereby fixing their
  order.

\end{enumerate}

In summary, the annotation helps distinguishing which unlock is
expected between to locks, vastly improving the speed of the backward
search.  This optimisation, however, required us to put restrictions
on the locks.  \fi

\subsection{Correctness of the translation}

The correctness of our translation is stated by the following theorem.

\begin{theorem}
  \label{thm:main}
  Given a well-formed ground process $P$ and a well-formed trace
  formula $\varphi$ we have that 
  $$
  \tracespi(P) \vDash^\star \varphi \mbox{ iff }
  \tracesmsr(\sem{P}) \vDash^\star \sem{\varphi}_\star
  $$
  where $\star$ is either $\forall$ or $\exists$.
\end{theorem}

We here give an overview of the main propositions and lemmas needed
to  prove Theorem~\ref{thm:main}. % Detailed proofs are given in
% Appendix.
%
To show the result we need two additional definitions. We first define
an operation that allows to restrict a set of traces to those that
satisfy the trace formula $\alpha$ as defined in
Definition~\ref{def:transformula}.
\begin{definition} \label{def:filter}
  Let $\alpha$ be the trace formula as defined in
  Definition~\ref{def:transformula} and $\Tr$ a set of traces. We
  define
  $$
  \filter(\mathit{Tr}) := \{ \tr\in\mathit{Tr} \mid \forall \theta.(\tr,\theta) \vDash
  \alpha \}
  $$
\end{definition}

The following proposition states that if a set of traces satisfies the
translated formula then the filtered traces satisfy the original
formula.
\begin{restatable}{proposition}{filterproposition}
  \label{prop:filter}
    Let $\Tr$ be a set of traces and $\varphi$ a trace formula.
    We have that 
    $$ \mathit{Tr} \vDash^\star \sem{\varphi}_\star \mbox{ iff } \filter(\mathit{Tr})
    \vDash^\star \varphi $$
    where $\star$ is either $\forall$ or $\exists$.
\end{restatable}

The proof (detailed in \iffullversion Appendix\else the full
version\fi) follows directly from the definitions.  Next we define the
\emph{hiding} operation which removes all reserved facts from a trace.
\begin{definition}[\hide]
  Given a trace $\tr$ and a set of facts $F$ we inductively define
  $\hide([]) = []$ and
  \[ 
  \hide(F \cdot \mathit{tr}) :=
  \begin{cases}
    \hide (\tr) & \mbox{if } F \subseteq \calF_\mathit{res}\\
     (F \setminus \calF_\mathit{res}) \cdot \hide(\tr) &  \mbox{otherwise}
  \end{cases}
  \]
  Given a set of traces $\Tr$ we define $\hide(\Tr) = \{ \hide(t) \mid t
  \in \Tr\}$.
\end{definition}
As expected well-formed formulas that do not contain reserved facts
evaluate the same whether reserved facts are hidden or not.

\begin{restatable}{proposition}{hideproposition}
  \label{prop:hide}
  Let $\Tr$ be a set of traces and $\varphi$ a well-formed trace formula.
  We have that 
  $$\Tr \vDash^\star \varphi \mbox{ iff }\hide(\Tr) \vDash^\star
  \varphi$$
  where $\star$ is either $\forall$ or $\exists$.
\end{restatable}

We can now state our main lemma which is relating the set of traces of
a process $P$ and the set of traces of its translation into multiset
rewrite rules (proven in the full version).
\begin{restatable}{lemma}{traceequivalence}
  \label{lem:trace-equivalence}
  Let $P$ be a well-formed ground process. We have that
  \[\tracespi(P) = \hide (\filter (\tracesmsr(\sem{P}))).\]
\end{restatable}

% This lemma is the main part of the proof.
% %NEWSTEVE
% The proof of this lemma in the full version proceeds by first showing each message
% that can be deduced in our calculus can also be derived using the
% rules in \textsc{MD} and \textsc{Fresh} (and vice versa), whenever a
% frame (in our calculus) corresponds to the set of $\K$ facts (in
% a msr state).  Next we prove the inclusion of $\tracespi(R)$ in
% $\hide(\filter(\tracesmsr(\sem{P})))$, by showing an invariant that
% implies the inclusion of traces generated using the semantics of our
% calculus in the fragment of multiset rewriting traces of the
% translation of process that fulfills the axioms. Said invariant
% relates a given process configuration to a multiset of facts derived
% from the translated rules.  To simplify the proof for the opposite
% inclusion, we first define a normal form of multiset rewriting
% executions with respect to our translation.  A ``normal'' execution,
% \eg, removes semi-states as soon as possible, keeps the rules
% corresponding to an internal communication together and produces
% \In-facts as late as possible.  We can show that we can bring every
% execution into this form, which allows us to only reason about
% ``normalized'' executions in the proof of the opposite inclusion.

Our main theorem can now be proven by applying
Lemma~\ref{lem:trace-equivalence}, \autoref{prop:hide} and
\autoref{prop:filter}.
% \begin{proof}[Proof of Theorem~\ref{thm:main}]
% \begin{align*}
%     \tracespi& (P)  \vDash^\star \varphi \\
%      \Leftrightarrow &
%     \hide ( \filter (\tracesmsr(\sem{P}))) \vDash^\star \varphi
%      \tag{by Lemma~\ref{lem:trace-equivalence}}\\
%      \Leftrightarrow &
%     \filter (\tracesmsr(\sem{P})) \vDash^\star \varphi
%      \tag{by Proposition~\ref{prop:hide}}\\
%      \Leftrightarrow &
%     \tracesmsr(\sem{P})\vDash^\star \sem{\varphi}_\star
%      \tag{by Proposition~\ref{prop:filter}}
% \end{align*}
% \end{proof}
\begin{proof}[Proof of Theorem~\ref{thm:main}]
$$
\begin{array}{r@{~}l}
  & \tracespi (P)  \vDash^\star \varphi \\
  \Leftrightarrow &
     \hide ( \filter (\tracesmsr(\sem{P}))) \vDash^\star \varphi
     \hfill \mbox{\qquad by Lemma~\ref{lem:trace-equivalence}}\\
     \Leftrightarrow &
    \filter (\tracesmsr(\sem{P})) \vDash^\star \varphi
     \hfill  \mbox{by Proposition~\ref{prop:hide}}\\
     \Leftrightarrow &
    \tracesmsr(\sem{P})\vDash^\star \sem{\varphi}_\star
     \hfill  \mbox{by Proposition~\ref{prop:filter}}
\end{array}
$$
\end{proof}

% Given the above propositions and lemma we can now easily proof our
% main theorem.
% \begin{proof}[Theorem~\ref{thm:main}]
% \begin{align*}
%     \tracespi& (P)  \vDash^\star \varphi \\
%      \Leftrightarrow &
%     \hide ( \filter (\tracesmsr(\sem{P}))) \vDash^\star \varphi
%      \tag{by Lemma~\ref{lem:trace-equivalence}}\\
%      \Leftrightarrow &
%     \filter (\tracesmsr(\sem{P})) \vDash^\star \varphi
%      \tag{by Proposition~\ref{prop:hide}}\\
%      \Leftrightarrow &
%     \tracesmsr(\sem{P})\vDash^\star \sem{\varphi}_\star
%      \tag{by Proposition~\ref{prop:filter}}\\
% \end{align*}
% \end{proof}

\lstset{morekeywords=[2]{YubiPress,Login,Smaller}}

\section{Case studies}
\label{sec:casestudies}

In this section we briefly overview some case studies we performed.
These case studies include a simple security API similar to
PKCS\#11~\cite{PKCS11}, the Yubikey security token, the optimistic
contract signing protocol by Garay, Jakobsson and MacKenzie
(GJM)~\cite{GJM99} and a few other examples discussed in Arapinis et
al.~\cite{ARR-csf11} and M\"odersheim~\cite{Modersheim-ccs10}.
%\iffullversion %
The results are summarized in \autoref{fig:case-studies}. % \else %
%The results are summarized in the following table.  \fi%
For each case study we provide the number of typing lemmas that were
needed by the tamarin prover and whether manual guidance of the tool
was required. In case no manual guidance is required we also give
execution times. We do not detail all the formal models of the
protocols and properties that we studied, and sometimes present
slightly simplified versions. All files of our prototype
implementation and our case studies are available at
\url{\downloadlink}.  % \iffullversion\else \medskip

\begin{figure}[ht] % (fold)
\centering
\iffullversion
\begin{tabular}{c | c c }
  Example                                                      &
  {Typing Lemmas} & {Automated Run$^{*}$}\\
  \hline
  Security API \`a la PKCS\#11                                 & 1             & yes ($51s$) \\
  % Needham-Schroeder-Lowe (w/ AS, w/ tagging)                 & not working   & - \\
  Yubikey Protocol~\cite{KS-stm12,yubikey}                     & 3             & no\\
  GJM protocol~\cite{ARR-csf11,GJM99}         & 0             & yes ($36s$)\\
  {Mödersheim's example (locks/inserts)~\cite{Modersheim-ccs10}} & 0             & no$^{**}$ \\
  {Mödersheim's example (embedded msr rules)~\cite{Modersheim-ccs10}} & 0             & yes ($1s$) \\
  Security Device~\cite{ARR-csf11}                     & 1             &
  yes ($21s$)\\
  Needham-Schroeder-Lowe~\cite{Lowe1996}                       & 1
  & yes ($5s$) \\
\end{tabular}
\begin{flushright}
  \footnotesize $^{*}$~(Running times on Intel Core2 Duo 2.66Ghz with
  4GB RAM)
  \\
  \footnotesize $^{**}$~(little interaction: 7 manual rule selections)
\end{flushright}

\else

\newcommand{\specialcell}[2][c]{%
  \begin{tabular}[#1]{@{}c@{}}#2\end{tabular}}
\begin{tabular}{c | c c }
  Example                                                      & \specialcell{Typing\\ Lemmas} & \specialcell{Automated\\Run$^{*}$}\\
  \hline
  Security API \`a la PKCS\#11                                 & 1             & yes ($51s$) \\
  % Needham-Schroeder-Lowe (w/ AS, w/ tagging)                 & not working   & - \\
  Yubikey Protocol~\cite{KS-stm12,yubikey}                     & 3             & no\\
  GJM protocol~\cite{ARR-csf11,GJM99}         & 0             & yes ($36s$)\\
  \specialcell{Mödersheim's example\\ (locks/inserts)~\cite{Modersheim-ccs10}} & 0             & no$^{**}$ \\
  \specialcell{Mödersheim's example\\ (embedded
    msr rules)~\cite{Modersheim-ccs10}} & 0             & yes ($1s$) \\
  Security Device~\cite{ARR-csf11}                     & 1             &
  yes ($21s$)\\
  Needham-Schroeder-Lowe~\cite{Lowe1996}                       & 1
  & yes ($5s$) \\
\end{tabular}
\begin{flushright}
  \footnotesize $^{*}$~(Running times on Intel Core2 Duo 2.66Ghz with
  4GB RAM)\\
  \footnotesize $^{**}$~(little interaction: 7 manual rule selections)
\end{flushright}

\fi
\caption{Case studies.}
\label{fig:case-studies}
% figure case-studies (end)
\end{figure}

% \begin{centering}
% \newcommand{\specialcell}[2][c]{%
%   \begin{tabular}[#1]{@{}c@{}}#2\end{tabular}}
% \begin{tabular}{c | c c }
%     Example                                                      & \specialcell{Typing\\ Lemmas} & \specialcell{Automated\\Run$^{**}$}\\
% \hline
% Security API \`a la PKCS\#11                                 & 1             & yes ($51s$) \\
% % Needham-Schroeder-Lowe (w/ AS, w/ tagging)                 & not working   & - \\
% Yubikey Protocol~\cite{KS-stm12,yubikey}                     & 3             & no\\
% GJM protocol~\cite{ARR-csf11,GJM99}         & 0             & yes ($36s$)\\
% \specialcell{Mödersheim's example\\ (locks/inserts)~\cite{Modersheim-ccs10}} & 0             & no$^*$ \\
% \specialcell{Mödersheim's example\\ (embedded
%   MSRs)~\cite{Modersheim-ccs10}} & 0             & yes ($1s$) \\
% Security Device~\cite{ARR-csf11}                     & 1             &
% yes ($21s$)\\
% Needham-Schroeder-Lowe~\cite{Lowe1996}                       & 1
% & yes ($5s$) \\
% \end{tabular}
%   \begin{flushright}
%       \footnotesize $^*$~(little interaction: 7 manual rule selections)\\
%       \footnotesize $^{**}$~(Running times on Intel Core2 Duo 2.66Ghz
%   with 4GB RAM)
      
%   \end{flushright}
% \end{centering}
% \fi

% \iffullversion %
% All files of our prototype implementation and our case studies are
% available at \url{\downloadlink}.
% \else %
% Because of space constraints we do not detail all the
% formal models of the protocols and properties that we studied. All
% files of our prototype implementation and our case studies are
% available at \url{\downloadlink}.
% \fi

\subsection{Security API  \`a la PKCS\#11}

This example illustrates how our modelling might be useful for the
analysis of Security APIs in the style of the PKCS\#11
standard~\cite{PKCS11}.  We expect studying a complete model of
PKCS\#11, such as in~\cite{DKS-jcs09}, to be a straightforward
extension of this example.  In addition to the processes presented in the running
example in \autoref{sec:calculus} the actual case study models the
following two operations:
\begin{inparaenum}[\it (i)]
\item \emph{encryption:} given a handle and a plain-text, the user can
  request an encryption under the key the handle points to.
\item \emph{unwrap} given a ciphertext $\mathit{senc}(k_2,k_1)$, and a
  handle $h_1$, the user can request the ciphertext to be
  \emph{unwrapped}, i.e.  decrypted, under the key pointed to by
  $h_1$. If decryption is successful the result is stored on the
  device, and a handle pointing to $k_2$ is returned.
\end{inparaenum}
Moreover, contrary to the running example, at creation time keys are
assigned the attribute `init', from which they can move to either
`wrap', or `unwrap', see the following snippet:
\begin{lstlisting}[numbers=left,numbersep=3pt,numberstyle=\footnotesize]
in(<`set_dec',h>); lock <`att',h>;
	 lookup <`att',h> as a in
		if a=`init'	then 
			insert <`att',h>,`dec'; unlock <`att',h>
\end{lstlisting}
Note that, in contrast to the running example, it is necessary to
encapsulate the state changes between lock and unlock. Otherwise an
adversary can stop the execution after line 3, set the attribute to
`wrap' in a concurrent process and produce a wrapping.
After resuming operation at line 4, he can set the key's attribute to
`dec', even though the attribute is set to `wrap'. Hence, the
attacker is allowed to decrypt the wrapping he has produced and can 
obtain the key. Such subtleties can produce attacks that our modeling allows to
detect. If locking is handled correctly, we show secrecy of keys
produced on the device, proving the property introduced in
\autoref{ex:property}. If locks are removed the attack described
before is found.

% \subsection{Needham-Schoeder-Lowe}

% We can show secrecy for a session-key established between two honest parties
% running the Needham-Schroeder-Lowe protocol \ref{missing-ref}. The
% modelling does not require tags on the messages.

\subsection{Yubikey}
The Yubikey~\cite{yubikey} is a small hardware device designed to
authenticate a user against network-based services.  Manufactured by
Yubico, a Swedish company, the Yubikey itself is a low cost (\$25),
thumb-sized USB device. In its typical configuration, it
generates one-time passwords based on encryptions of a secret value,
a running counter and some random values using a unique AES-128 key
contained in the device. The Yubikey authentication server accepts
a one-time password only if it decrypts under the correct AES key to
a valid secret value containing a counter larger than the last counter
accepted. The counter is thus used as a means to prevent replay
attacks.  To date, over a million Yubikeys have been shipped to more
than 30,000 customers including governments, universities and
enterprises, e.g.  Google, Microsoft, Agfa and
Symantec~\cite{Yubicos-custome}.

Besides the counter values used in the one-time password, the Yubikey
stores three additional pieces of information: the public id
$\mathit{pid}$ that is used to identify the Yubikey, a secret id
$\mathit{secretid}$ that is transmitted as part of the one-time
password and only known to the server and the Yubikey, as well as the
AES key $k$, which is also shared with the server. The following
process $P_\mathit{Yubikey}$ models a single Yubikey, as well as its
initial configuration, where an entry in the server's database for the
public id $\mathit{pid}$ is created. This entry contains a tuple
consisting of the Yubikey's secret id, AES key, and an initial counter
value.
\begin{lstlisting}
$P_\mathit{Yubikey}=$
 new k; new pid; new secretid;
 insert <`Server',pid>, <secretid, k, `zero'>;
 insert <`Yubikey', pid>,`zero'+`one';
 out(pid);
 $!P_\mathit{Plugin}$ | $!P_\mathit{ButtonPress}$
\end{lstlisting}
Here, the processes $!P_\mathit{Plugin}$ and $!P_\mathit{ButtonPress}$
model the Yubikey being unplugged and plugged in again (possibly on a
different computer), and the emission of the one-time password.  We
will only discuss $P_\mathit{ButtonPress}$ here.  When the user
presses the button on the Yubikey, the device outputs a
one-time password consisting of a counter $tc$, the secret id
$\mathit{secretid}$ and additional randomness $\mathit{npr}$ encrypted
using the AES key $k$.
\begin{lstlisting}
$P_\mathit{ButtonPress}=$
 lock pid;
   lookup <`Yubikey',pid> as tc in
     insert <`Yubikey',pid>, tc + `one';
     new nonce; new npr;
     event YubiPress(pid,secretid,k,tc);
     out(<pid,nonce,senc(<secretid,tc,npr>,k)>);
 unlock pid
\end{lstlisting}
The one-time password $senc(\langle secretid,tc,npr \rangle,k)$ can
be used to authenticate against a server that shares the same secret
key, which we model in the process $P_\mathit{Server}$. The process
receives the encrypted one-time password along with the public id
$\mathit{pid}$ of a Yubikey and a $\mathit{nonce}$ that is part of the
protocol, but is irrelevant for the authentication of the Yubikey on
the server.

The server looks up the secret id and the AES key associated to the
public id, \ie, to the Yubikey sending the request, as well as the
last recorded counter value $\mathit{otc}$. If the key and secret id
used in the request match the values retrieved from the database, then
the event $\mathrm{Smaller}(otc,tc)$ is logged along with the event
$\mathrm{Login}(\mathit{pid},k,\mathit{tc})$, which marks a successful
login of the Yubikey $\mathit{pid}$ with key $k$ for the counter value
$\mathit{tc}$. Afterwards, the old tuple $\langle
\mathit{secretid},k,\mathit{otc} \rangle$ is replaced by $\langle
\mathit{secretid},k,\mathit{tc} \rangle$, to update the latest counter
value received.
\begin{lstlisting}
$P_\mathit{Server}=$
$!$ in(<pid,nonce,senc(<secretid,tc,npr>,k)>);
 lock pid;
 lookup <`Server',pid> as tuple in
   if fst(tuple)=secretid then
     if fst(snd(tuple))=k then
       event Smaller(snd(snd(tuple)), tc)
       event Login(pid,k,tc);
       insert <`Server',pid>, <secretid,k,tc>;
 unlock pid
\end{lstlisting}

Note that, in our modelling, the server keeps one lock per public id,
which means that it is possible to have several active instances of
the server thread in parallel as long as all requests concern 
different Yubikeys.

An important part of the modelling of the protocol is to determine
whether one counter value is smaller than another. To this end, our
modelling employs a feature added to the development version of
tamarin as of October 2012, a union operator $\cup^\#$ for multisets
of message terms. The operator is denoted with a plus sign (``$+$'').
% in tamarin; to avoid confusion we use the same notation in the process
% calculus.
We model the counter as a multiset only consisting of the symbols
``one'' and ``zero''. The multiplicity of `one' in the multiset is the value of the
counter. A counter value is considered smaller than another one, if
the first multiset is included in the second. A test $a < b$ is
included by adding the event $\mathrm{Smaller}(a,b)$ and an axiom that
requires that $a$ is a subset of $b$:
\begin{align*}
    \alpha_\mathit{Smaller} := & \forall i:\tempsort, a, b:\msgsort.~\mathrm{Smaller}(a,b)@i \\
                               & \quad \Rightarrow \exists z:\msgsort.\  a+z=b
\end{align*}
We incorporate this axiom into the security properties just like in
Definition~\ref{def:transformula}. Intuitively, we are only interested
in traces where $a$ is indeed smaller than $b$.%$\mathrm{Smaller}$ has the correct semantics.

The process we analyse models a single authentication server (that may
run arbitrary many threads) and an arbitrary number of Yubikeys, \ie,
$P_\mathit{Server} \mid\, !P_\mathit{Yubikey}$.  Among other
properties, we show by the means of an injective correspondence
property that an attacker that controls the network cannot perform
replay attacks, and that each successful login was preceded by a user
``pressing the button'', formally:
\begin{align*}
\forall \, & \mathit{pid}, k, x, t_{2}. \mathrm{Login}(\mathit{pid,k,x})@t_{2} \Rightarrow  \\
 & \quad \exists \mathit{sid}, t_{1}  .
\mathrm{YubiPress}(\mathit{pid,sid,k,x})@t_{1} \wedge 
t_{1}\lessdot t_{2}  \\
       & \qquad   \wedge \forall \mathit{%otp_2,
 t_{3} }.
\mathrm{Login}(\mathit{pid,k,x})@t_{3} \Rightarrow t_{3}=t_{2}
\end{align*}
Besides injective correspondence, we show the absence of replay
attacks and the property that a successful login invalidates
previously emitted one-time passwords. All three properties follow
more or less directly from a stronger invariant, which itself can be
proven in 295 steps. To find theses steps, tamarin needs some
%TODO verify this number.
additional human guidance, which can be provided using the interactive
mode.  This mode still allows the user to complement his manual
efforts with automated backward search. The example files contain the
modelling in our calculus, the complete proof, and the manual part of
the proof which can be verified by tamarin without interaction.

%\iffullversion
Our analysis makes three simplifications: First, in
$P_\mathit{Server}$, we use pattern matching instead of decryption as
demonstrated in the process $P_\mathit{dec}$ we introduced in
\autoref{sec:calculus}. Second, we omit the CRC checksum and the
time-stamp that are part of the one-time password in the actual
protocol, since they do not add to the security of the protocol in the
symbolic setting. Third, the Yubikey has actually two counters instead
of one, a session counter, and a token counter.  We treat the session
and token counter on the Yubikey as a single value, which we justify
by the fact that the Yubikey either increases the session counter and
resets the token counter, or increases only the token counter, thereby
implementing a complete lexicographical order on the pair
$(\mathit{session\ counter}, \mathit{token\ counter})$.
%\fi

A similar analysis has already been performed by Künnemann and Steel,
using tamarin's multiset rewriting calculus~\cite{KS-stm12}.  However,
the model in our new calculus is more fine-grained and we believe more
readable. Security-relevant operations like locking and tests on state
are written out in detail, resulting in a model that is closer to the
real-life operation of such a device. The modeling of the Yubikey
takes approximately 38 lines in our calculus, which translates to 49
multiset rewrite rules. The model of~\cite{KS-stm12} contains only
four rules, but they are quite complicated, resulting in 23 lines
of code.  More importantly, the gap between their model and the actual
Yubikey protocol is larger -- in our calculus, it becomes clear that
the server can treat multiple authentication requests in parallel, as
long as they do not claim to stem from the same Yubikey.
%NEWSTEVE
An implementation on the basis of the model from K\"{u}nnemann and
Steel would need to implement a global lock accessible to the
authentication server and all Yubikeys.  This is however unrealistic,
since the Yubikeys may be used at different places around the world,
making it unlikely that there exist means of direct communication
between them.  While a server-side global lock might be conceivable
(albeit impractical for performance reasons), a real global lock could
not be implemented for the Yubikey as deployed.

\subsection{Further Case Studies}

We also investigated the case study presented by
Mödersheim~\cite{Modersheim-ccs10}, a key-server example. We encoded
two models of this example, one using the insert construct, the other
manipulating state using the embedded multiset rewrite rules. For
this example the second model turned out to be more natural and more
convenient allowing for a direct automated proof without any
additional typing lemma.

We furthermore modeled the contract signing protocol by Garay et
al.~\cite{GJM99} and a simple security device which both served as
examples in~\cite{ARR-csf11}. In the contract signing protocol a
trusted party needs to maintain a database with the current status of
all contracts (aborted, resolved, or no decision has been taken).  In
our calculus the status information is naturally modelled using our
insert and lookup constructs. The use of locks is indispensable to
avoid the status to be changed between a lookup and an insert.
Arapinis et al.~\cite{ARR-csf11} showed the crucial property that the
same contract can never be both aborted and resolved. However, due to
the fact that StatVerif only allows for a finite number of memory
cells, they have shown this property for a single contract and provide
a manual proof to lift the result to an unbounded number of
contracts. We directly prove this property for an unbounded number of
contracts.
%
% : the device is initialized once to left or
% right. Later on it accepts pairs of encryptions and decrypts either
% the left component of the pair or the right component, but not
% both. As the input language of StatVerif is very similar to ours their
% model could be easily adapted to our tool.
%
Finally we also illustrate the tool's ability to analyze classical
security protocols, by analyzing the Needham Schroeder Lowe
protocol~\cite{Lowe1996}.

\section{Conclusion}

We present a process calculus which extends the applied pi calculus
with constructs for accessing a global, shared memory together with an
encoding of this calculus in labelled msr rules which enables
automated verification using the tamarin prover as a backend. Our
prototype verification tool, automating this translation, has been
successfully used to analyze several case studies. As future work we
plan to increase the degree of automation of the tool by automatically
generating helping lemmas. To achieve this goal we can exploit the
fact that we generate the msr rules, and hence control their form.  We
also plan to use the tool for more complex case studies including a
complete model of PKCS\#11 and a study of the TPM 2.0 standard,
currently in public review. Finally, we wish to investigate how our
constructs for manipulating state can be used to encode loops, needed
to model stream protocols such as TESLA.

% We presented a translation procedure that allows for analyzing a wide
% range of security properties for processes in a high-level
% specification language.  
% Future work will concentrate on developing a methodology for
% deriving the typing lemmas tamarin needs to obtain a proof
% automatically, thus permitting more automation and more expressive
% models.

% We presented a translation procedure that allows for analyzing a wide
% range of security properties for processes in a high-level
% specification language.  For future work, we plan to extend the case
% studies by modeling for instance a larger fragment of PKCS\#11 as to
% be able to possibly identify new secure configurations, embedding the
% fragments of the TPM analysed in previous
% work~\cite{Delaune2011A-Formal-Analys,DKRS-csf11} into one model while
% at the same time strengthening their model,\eg, allowing for an
% unlimited number of reboots.  In the long term, we aim to integrate
% our process calculus into tamarin, as to provide a user experience
% similiar to that of ProVerif. The process calculus provides some
% additional structure that would, for instance, allow tamarin to
% distinguish different principals thus visualise attacks better.
% Finally, our work will concentrate on developing a methodology for
% deriving the typing lemmas tamarin needs to obtain a proof
% automatically, thus permitting more automation and more expressive
% models.

\paragraph{Acknowledgements} The research leading to these results has
received funding from the European Research Council under the European
Union's Seventh Framework Programme (FP7/2007-2013) / ERC grant
agreement no 258865, project ProSecure and was supported by CASED
(\href{http://www.cased.de}{http://www.cased.de}).

%\bibliography{../recherche/references}

\iffullversion
\bibliographystyle{abbrv}
\else
\bibliographystyle{IEEEtran}
\fi
\bibliography{references}

\appendix

%\newpage

\section{Definition of the process annotation}
\label{app:defs}

\begin{definition}[Process annotation] % (fold)
\label{def:process-annotation}
  Given a ground process $P$ we define the annotated ground process $\overline P$ as follows:
 \renewcommand{\arraystretch}{1.3}
\[\arraycolsep=1.4pt
\begin{array}{rcl}
\overline 0 &:= &0 \\
\overline {P | Q} &:= &\overline P | \overline Q \\
\overline {! P} &:= & ! \overline P \\
\begin{array}{r}
  \overline {\text{if $t_1=t_2$ then $P$}}\\[-1mm]
 \overline {\text{else $Q$}}
\end{array}
 &:= &
 {\text{if $t_1=t_2$ then $\overline P$ else $\overline Q$}}\\
\begin{array}{r}
\overline{\text{lookup  $M$ as $x$}}\\[-1mm]
\overline{\text{in $P$ else $Q$}} 
\end{array}
&:= &
\begin{array}{l}
\text{lookup  $M$ as $x$}\\[-1mm]
\text{in $\overline P$ else $\overline Q$}
\end{array}\\
\overline { \alpha; P } &:= & \alpha; \overline P \\[-2mm]
\multicolumn{3}{r}{\text{where }\alpha\notin\set{\mathrm{lock}~t,\mathrm{unlock}~t
\colon t\in\Terms}}\\
\overline {\mathrm{lock}~t; P} &:=& \mathrm{lock}^l~t; \overline {
\mathit{au}(P,t,l)} \\[-1.5mm]
\multicolumn{3}{r}{\text{where $l\in\setN$ is a fresh label}}\\
\overline {\mathrm{unlock}^l~t; P} &:=& \mathrm{unlock}^l~t; \overline {P}\\
\overline {\mathrm{unlock}~t; P} &:=& \bot
\end{array}\]
%
% \begin{align*}
% \overline 0 &:= 0 \\
% \overline {P | Q} &:= \overline P | \overline Q \\
% \overline {! P} &:= ! \overline P \\
% \begin{aligned}
% \overline {\text{if $t_1=t_2$ then $P$}}\\[-1mm]
%  \overline {\text{else $Q$}}
% \end{aligned}
%  &:= 
%  {\text{if $t_1=t_2$ then $\overline P$ else $\overline Q$}}\\
% \begin{aligned}
% \overline{\text{lookup  $M$ as $x$}}\\[-1mm]
% \overline{\text{in $P$ else $Q$}} 
% \end{aligned}
% &:= 
% \text{lookup  $M$ as $x$ in $\overline P$ else $\overline Q$}\\
% \overline { \alpha; P } &:= \alpha; \overline P \\
% & \text{where }\alpha\notin\set{\mathrm{lock}~t,\mathrm{unlock}~t
% \colon t\in\Terms}\\
% \overline {\mathrm{lock}~t; P} &:= \mathrm{lock}^l~t; \overline {
% \mathit{au}(P,t,l)} \\
% \overline {\mathrm{unlock}~t; P} &:= \bot
% \end{align*}
where $\mathit{au}(P,t,l)$ annotates the first unlock that has parameter
$t$ with the label $l$, \ie:
\[\arraycolsep=1.4pt
\begin{array}{rcl}
  \mathit{au}({P | Q}, t,l)&:= & \bot \\
  \mathit{au}({! P}, t,l)&:= & \bot\\
  \renewcommand{\arraystretch}{1}
  \begin{array}{r}
    \mathit{au}({\text{if $t_1=t_2$ then}}\\
    \text{$P$ else $Q$},t,l)
  \end{array}
  &:= & 
  \renewcommand{\arraystretch}{1}
  \begin{array}{l}
    \text{if } t_1=t_2\text{ then $\mathit{au}(P,t,l)$}\\
    \text{else $\mathit{au}(Q,t,l)$}
  \end{array}
  \\
  \renewcommand{\arraystretch}{1}
  \begin{array}{r}
    \mathit{au}(\text{lookup  $M$ as $x$ }\\
    \text{in $P$ else $Q$},t,l)
  \end{array}
  &:= & 
  \renewcommand{\arraystretch}{1}
  \begin{array}{l}
    \text{lookup  $M$ as $x$ in}\\
    \text{$\mathit{au}(P,t,l)$ else
      $\mathit{au}(Q,t,l)$}
  \end{array}
  \\
  \mathit{au}( { \alpha; P }, t,l)&:= & \alpha; \mathit{au}(P,t,l) \\[-2mm]
  \multicolumn{3}{r}{\text{where }\alpha\neq {unlock}~t}\\
  \mathit{au}( {\mathrm{unlock}~t; P}, t,l)&:= & \mathrm{unlock}^l~t; P \\
  \mathit{au}(0,t,l)&:= & 0 \\
\end{array}
\]

% \begin{align*}
% \mathit{au}({P | Q}, t,l)&:= \bot \\
% \mathit{au}({! P}, t,l)&:= \bot\\
% \begin{aligned}
% \mathit{au}({\text{if $t_1=t_2$ then}}\\[-1mm]
% \text{$P$ else $Q$},t,l)
% \end{aligned}
% &:= 
% \begin{aligned}
%  \text{if }&t_1=t_2\text{ then}\\[-1mm]
% &\text{$\mathit{au}(P,t,l)$ else $\mathit{au}(Q,t,l)$}
% \end{aligned}
% \\
% \begin{aligned}
% \mathit{au}(&\text{lookup  $M$ as $x$ }\\[-1mm]
% & \text{in $P$ else $Q$},t,l)
% \end{aligned}
% &:= 
% \begin{aligned}
% 	&\text{lookup  $M$ as $x$}\\[-1mm]
% & \text{ in $\mathit{au}(P,t,l)$ else
%      $\mathit{au}(Q,t,l)$}
% \end{aligned}
% \\
% \mathit{au}( { \alpha; P }, t,l)&:= \alpha; \mathit{au}(P,t,l) 
% \quad \text{where }\alpha\neq {unlock}~t \\
% \mathit{au}( {\mathrm{unlock}~t; P}, t,l)&:= \mathrm{unlock}^l~t; P \\
% \mathit{au}(0,t,l)&:= \bot \\
% \end{align*}
% definition process_annotation (end)
\end{definition}

\iffullversion
% we need to introduce \sem{\cdot}^D for the proof of lemma lem:freshnotded
\section{Correctness of tamarin's solution procedure for translated rules}
\label{app:solution}

The multiset rewrite system produced by our translation for a well-formed
process $P$ could actually contain rewrite rules that are not valid with
respect to Definition~\ref{def:msr-system}, because they violate the third
condition, which is: for each $l'\msrewrite {a'} r' \in R \in_\ET \ginsts(l\msrewrite
    a r)$ we have that $\cap_{r'' =_\ET r'}\names(r'') \cap \FN \subseteq
    \cap_{l'' =_\ET l'}\names(l'')\cap \FN$.

This does not hold for rules in $\sem{P}_{=p}$ where $p$ is the position of
the lookup-operator. The right hand-side of this rule can be instantiated
such that, assuming the variable bound by the lookup is named $v$, this
variable $v$ is substituted by a names that does not appear on the
left-hand side. In the following, we will show that the results
from~\cite{SMCB-csf12} still hold. In practice, this means that the
tamarin-prover can be used for verification, despite the fact that it
outputs well-formedness errors for each rule that is a translation of a
lock.

We will introduce some notation first. We re-define $\sem{P}$ to contain
the \textsc{Init} rule and $\sem{\overline P, \emptyrule, \emptyrule}$, but
not \textsc{MD} (which is different to Definition~\ref{def:transprocess}).
We furthermore define a translation with dummy-facts, denoted $\sem{P}^D$,
that contains \textsc{Init} and $\sem{\overline P, \emptyrule,
\emptyrule}^D$, which is defined as follows:

\begin{definition}
\label{def:dummy-trans}
We define $\sem{P}^D:= \textsc{Init} \cup \sem{\overline P, \emptyrule,
\emptyrule}^D$, where $\sem{\overline P, \emptyrule,
\emptyrule}^D$ is defined just as $\sem{\overline P, \emptyrule,
\emptyrule}$, with the exception of two cases,
$P= \mbox{lookup}~M\mbox{ as }v \mbox{ in } P \mbox{ else } Q$ and
$P=\mbox{insert } s, t ; P$, where it is defined as follows:
\begin{align*}
  \sem{\mbox{lookup}~M\mbox{ as }v \mbox{ in } P \mbox{ else } Q, p,
\tilde x}^D & =
 \{ [\state_p(\tilde x),\dummy(v)]
  \msrewrite{\mathrm{IsIn}(M,v)}
  [\state_{p\cdot 1}(\tilde M,v)], \\
   & \phantom{ = \{ ~} [\state_p(\tilde x)]
\msrewrite{\mathrm{IsNotSet}(M)}
   [\state_{p\cdot 2}(\tilde x)] \}
\\
 & \phantom{ = \,} \cup \sem{P,p\cdot~1,(\tilde x,v)}^D\cup  \sem{Q, p\cdot 2, \tilde
x}^D\\
   \sem{\mbox{insert } s, t ; P, p, \tilde x}^D &=
 \set{ [\state_p(\tilde x)] \msrewrite{\mathrm{Insert}(s,t)}
  [\state_{p\cdot 1}(\tilde x),\dummy(t)] } \\
&\phantom{ = \,} \cup \sem{P, p\cdot 1, \tilde x}^D \\
\end{align*}
\end{definition}

The only difference between $\sem{P}$ and $\sem{P}^D$ is therefore, that
$\sem{P}^D$ produces a permanent fact $\dummy$ for every value $v$ that
appears in an action $\mathrm{insert}(k,v)$, which is
 a premise to every rule instance with an action $\mathrm{IsIn}(k',v)$.
We see that $\sem{P}^D$ contains now only valid multiset rewrite rules.

In the following, we would like to show that the tamarin-prover's solution
algorithm is correct for $\sem{P}$. To this end, we make use of the proof
of correctness of tamarin as presented in Benedikt Schmidt's Ph.D.
thesis~\cite{benschmi-thesis}. We will refer to Lemmas, Theorems and
Corollaries in this work by their numbers. We will use the notation of this
work, to make it easier to the reader to compare our statements against the
statements there. In particular,
$\overline{\mathit{trace(\mathit{execs}(R))}}$ is $\tracesmsr(R)$  in our
notation. We have to show that:

\begin{lemma}
For all well-formed process $P$ and guarded trace properties $\phi$,
\[
\mathit{trace}(\mathit{execs}(\sem{P}\cup\textsc{MD})
\vDash_{\mathcal{DH}_e} \neg \AssSetIn \vee \phi
\]
if and only if
\[
\mathit{trace}(\mathit{ndgraphs}(\sem{P})) \vDash_{ACC}\neg \AssSetIn \vee
\phi.
\]
\end{lemma}
\begin{proof}
The proof proceeds similar to the proof to Theorem~3.27. We refer to
results in~\cite{benschmi-thesis}, whenever their proofs apply despite the
fact that the rules in $\sem{P}$ do not satisfy the third condition of
multiset rewrite rules.
\begin{align*}
& \mathit{trace}(\mathit{execs}(\sem{P}\cup\textsc{MD})
\vDash_{\mathcal{DH}_e} \neg \AssSetIn \vee \psi\\
\Leftrightarrow~&
\overline {\mathit{trace}(\mathit{execs}(\sem{P}\cup\textsc{MD})}
\vDash_{\mathcal{DH}_e} \neg \AssSetIn \vee \phi
\tag{Lemma~3.7 (unaltered)}\\
\Leftrightarrow~&
\overline {\mathit{trace}(\mathit{execs}(\sem{P}\cup\textsc{MD}))}
\downarrow_{\mathcal{RDH}_e}
\vDash_{\mathcal{DH}_e} \neg \AssSetIn \vee \phi
\tag{Definition of $\vDash_{\mathcal{DH}_e}$}\\
\Leftrightarrow~&
\overline
{\mathit{trace}(\mathit{dgraphs}_{\mathcal{DH}_e}(\sem{P}\cup\textsc{MD}))}
\downarrow_{\mathcal{RDH}_e}
\vDash_{\mathcal{DH}_e} \neg \AssSetIn \vee \phi
\tag{Lemma~3.10 (unaltered)}\\
\Leftrightarrow~&
\overline
{\mathit{trace}(\{\mathit{dg}\mid
\mathit{dg}\in\mathit{dgraphs}_\mathcal{ACC}(\lceil \sem{P}\cup \textsc{MD}
\rceil^{\mathcal{RDH}_e}_\mathit{insts}) }\\
&\overline{
\wedge \mathit{dg} \downarrow_{\mathcal{RDH}_e}\text{-normal}\})
}
\vDash_{\mathcal{DH}_e} \neg \AssSetIn \vee \phi
\tag{Lemma~3.11 (unaltered)}\\
\Leftrightarrow~&
\overline
{\mathit{trace}(\mathit{ndgraphs}(\sem{P}))
}
\vDash_{\mathcal{DH}_e} \neg \AssSetIn \vee \phi
\tag{Lemma~A.12 (*)}\\
\Leftrightarrow~&
\mathit{trace}(\mathit{ndgraphs}(\sem{P}))
\vDash_{\mathcal{ACC}} \neg \AssSetIn \vee \phi
\tag{Lemma~3.7 and A.20(both unaltered)}\\
\end{align*}

It is only in Lemma~A.12 where the third condition is used: The proof to
this lemma applies Lemma~A.14, which says that all factors (or their
inverses) are known to the adversary. We will quote Lemma~A.14 here:

\begin{lemma}[Lemma~A.14 in~\cite{benschmi-thesis}]
For all $\mathit{ndg}\in\mathit{ndgraphs}(P)$, conclusions $(i,u)$ in
$\mathit{ndg}$ with conclusion fact $f$ and terms
$t\in\mathit{afactors}(f)$, there is a conclusion $(j,v)$ in $\mathit{ndg}$
with $j<i$ and conclusion fact $\mathsf{K}^d(m)$ such that
$m\in_\mathit{ACC} \set{t,(t^{-1})\downarrow_{\mathcal{RBP}_e}}$.
\end{lemma}

If there is  $\mathit{ndg}\in\mathit{ndgraphs}(\sem{P})$, such that
$\mathit{trace}(\mathit{ndg})\vDash_\mathit{ACC} \AssSetIn$, then
\begin{align*}
& \mathit{trace}(\mathit{ndgraphs}(\sem{P})) \vDash_{\mathcal{ACC}} \neg
\AssSetIn \vee \phi  \\
\Leftrightarrow  & \forall \mathit{ndg}\in\mathit{ndgraphs}(\sem{P}) \text{
~s.\,t.~} \mathit{trace}(\mathit{ndg})\vDash_{\mathit{ACC}} \AssSetIn \\
&\qquad \mathit{trace}(\mathit{ndg}) \vDash_{\mathcal{ACC}} \vDash \phi
\end{align*}

Since for the empty trace, $\emptyrule \vDash_{\mathit{ACC}}\AssSetIn$, we
only have to show that Lemma~A.14 holds for
$\mathit{ndg}\in\mathit{ndgraphs}(\sem{P})$, such that
$\mathit{trace}(\mathit{ndg})\allowbreak \vDash_\mathit{ACC} \AssSetIn$.

For every $\mathit{ndg}\in\mathit{ndgraphs}(\sem{P})$, such that
$\mathit{trace}(\mathit{ndg})\allowbreak \vDash_\mathit{ACC} \AssSetIn$,
there is a trace
equivalent $\mathit{ndg}'\in\mathit{ndgraphs}(\sem{P}^D)$, since the only
difference between $\sem{P}$ and $\sem{P}^D$ lies in the dummy conclusion
and premises, and $\AssSetIn$ requires that any $v$ in an action
$\mathrm{IsIn(u,v)}$ appeared previously in an action
$\mathrm{Insert(u,v)}$ (equivalence modulo $\mathit{ACC}$). Therefore,
$\mathit{ndg'}$ has the same $\mathsf{K}^d$-conclusions $\mathit{ndg}$ has,
and every conclusion in $\mathit{ndg}$ is a conclusion in $\mathit{ndg}'$.

We have that Lemma~A.14 holds for $\sem{P}^D$, since all rules generated in
this translation are valid multiset rewrite rules. Therefore, Lemma~A.14
holds for all $\mathit{ndg}\in\mathit{ndgraphs}(\sem{P})$, such that
$\mathit{trace}(\mathit{ndg})\vDash_\mathit{ACC} \AssSetIn$, too,
concluding the proof by showing the marked (*) step.

\end{proof}

\section{Proofs of Section~\ref{sec:translation}}

%abbreviation for convenience:

\newcommand{\invariantSubst}[1]{\ensuremath{
\set{x \Subst_{#1} \mid x\in\Dom(\Subst_{#1})}^\#=\{\mathsf{Out}(t)\in \cup_{k
\leq f({#1})} S_k\}^\#
 }}
\newcommand{\invariantNonce}[1]{\ensuremath{
 \Names_{#1} =  \set{ a \mid \ProtoNonce(a)\in\bigcup_{1 \leq j \leq f(#1)}F_j}
}}
\newcommand{\invariantProc}[1]{\ensuremath{
\Processes_{#1} \leftrightarrow_P \{\state_p(\tilde x)\in^\# S_{i_{#1}}\}
}}

\filterproposition*

\begin{proof}
We first show the two directions for the case $\star = \forall$. We
start by showing that $\Tr \vDash^\forall \sem{\varphi}$ implies
$\filter(\mathit{Tr})\vDash \varphi$.
\begin{align*}
  \Tr \vDash^\forall \sem{\varphi}_\forall & \Rightarrow
  \filter(\Tr) \vDash^\forall \sem{\varphi}_\forall   \tag{since $\filter(\Tr) \subseteq
    \Tr$}\\
  & \Leftrightarrow  \filter(\Tr) \vDash^\forall \alpha \Rightarrow \varphi
   \tag{by definition of $\sem{\varphi}_\forall$} \\
  & \Leftrightarrow  \filter(\Tr) \vDash^\forall  \varphi
   \tag{since $\filter(\Tr)\vDash^\forall \alpha$}
\end{align*}

\noindent We next show that $\filter(\mathit{Tr})\vDash^\forall \varphi$ implies $ \mathit{Tr} \vDash^\forall
\sem{\varphi}_\forall$.
\begin{align*}
\filter(\Tr) \vDash^\forall \varphi &\Rightarrow
\filter(\Tr) \vDash^\forall \alpha \wedge \varphi  \tag{since $\filter(\Tr) \vDash^\forall
    \alpha$}\\
  &\Leftrightarrow \Tr \vDash^\forall \neg \alpha \lor (\alpha \land \varphi)
  \tag{since $\filter(\Tr)\subseteq \Tr$ and $(\Tr \setminus \filter(\Tr))
    \not \vDash^\forall \alpha$}\\
&\Leftrightarrow \Tr \vDash^\forall \alpha \Rightarrow \varphi \\
&\Leftrightarrow \Tr \vDash^\forall \sem{\varphi}_\forall & \tag{by
  definition of $\sem \varphi_\forall$}
\end{align*}
The case of $\star = \exists$ now easily follows:
\[
\begin{array}{rcl}
   \Tr \vDash^\exists{\sem \varphi}_\exists &
  \mbox{ iff } & \Tr \not  \vDash^\forall{\sem {\neg
      \varphi}}_\forall\\ 
&  \mbox{ iff } & \filter(\mathit{Tr})\not \vDash^\forall \neg \varphi\\
&  \mbox{ iff } &  \filter(\Tr) \vDash^\exists \varphi.
\end{array}
\]
\end{proof}

\hideproposition

\begin{proof}
  We start with the case $\star = \exists$ and show the stronger
  statement that for a trace $\tr$
  $$\forall \theta.  \exists
  \theta'. \mbox{ if } (\tr, \theta) \vDash \varphi \mbox{ then }(\hide(\tr), \theta') \vDash \varphi$$
  and
  $$\forall \theta. \exists
  \theta'. \mbox{ if } (\hide(\tr), \theta) \vDash \varphi \mbox{ then
  } (\tr, \theta') \vDash \varphi$$ We will show both statements by a
  nested induction on $|\tr|$ and the structure of the formula. (The
  underlying well-founded order is the lexicographic ordering of the
  pairs consisting of the length of the trace and the size of the
  formula)

  \smallskip

  \noindent If $|\tr|=0$ then $\tr =[]$ and $\tr = \hide(\tr)$ which allows us to
  directly conclude letting $\theta' := \theta$.

  \smallskip

  \noindent If $|\tr| = n$, we define $\overline \tr$ and $F$ such that $\tr = \overline{\tr} \cdot
  F$. By induction hypothesis we have that
  $$\forall \overline\theta. \exists
  \overline\theta'. \mbox{ if } (\overline {\tr}, \overline\theta) \vDash \varphi \mbox{ then } (\hide(\overline {\tr}), \overline\theta') \vDash \varphi$$
  and
  $$\forall \overline\theta. \exists
  \overline\theta'. \mbox{ if } (\hide(\overline {\tr}),
  \overline\theta) \vDash \varphi \mbox{ then } (\overline {\tr},
  \overline\theta') \vDash \varphi$$
  We proceed by structural
  induction on $\varphi$.
  \begin{itemize}

  \item $\varphi = \bot$, $\varphi= i\lessdot j$, $\varphi= i\doteq j$
    or $t_1 \approx t_2$. In these cases we trivially conclude as the
    truth value of these formulas does not depend on the trace and for
    both statements we simply let $\theta':= \theta$.

  \item $\varphi = f @ i$. We start with the first statement. Suppose
    that $(\tr, \theta) \vDash f@i$. If $\theta(i) < n$
    then we have also that $\overline{\tr}, \theta \vDash f@i$. By
    induction hypothesis, there exists $\overline
    \theta'$ such that $(\overline{\tr}, \overline \theta') \vDash f@i$. Hence we
    also have that $ (\tr, \overline \theta') \vDash f@i$ and letting
    $\theta' := \overline \theta'$ allows us to conclude.  If $\theta(i)
    = n$ we know that $f \in tr_n$. As $\varphi$ is well-formed $f
    \not \in \calF_{res}$ and hence $f \in \hide(\tr)_{n'}$ where $n' =
    |\hide(\tr)|$.  The proof of the other statement is similar.

  \item $\varphi = \neg \varphi'$, $\varphi = \varphi_1 \wedge
    \varphi_2$, or $\varphi = \exists x:s. \varphi'$. We directly
    conclude by induction hypotheses (on the structure of $\varphi$).

  \end{itemize}
  From the above statements we easily have that $\Tr \vDash^\exists
  \varphi$ iff $\hide(\Tr) \vDash^\exists \varphi$.

  The case of $\star =\forall$ now easily follows:
  $$\Tr \vDash^\forall
  \varphi \mbox{ iff } \Tr \not \vDash^\exists
  \neg \varphi \mbox{ iff } \hide(\Tr) \not \vDash^\exists
  \neg \varphi \mbox{ iff } \hide(\Tr) \vDash^\forall \varphi$$

\end{proof}

In order to prove Lemma~\ref{lem:trace-equivalence}, we need a few additional
lemmas.  

We say that a set of traces $\Tr$ is prefix closed if for all $\tr \in
\Tr$ and for all $\tr'$ which is a prefix of $\tr$ we have that $\tr'
\in \Tr$.

\begin{lemma}[\filter is prefix-closed]\label{lem:filter-prefix}
  Let \Tr be a set of traces.  If \Tr is prefix closed then
  $\filter(\Tr)$ is prefix closed as well.
\end{lemma}
\begin{proof}
  It is sufficient to show that for any trace $\tr = \tr' \cdot
  a$ we have that if $\forall \theta.\: (\tr,\theta) \vDash \alpha$ then
  $\forall \theta.\: (\tr',\theta) \vDash \alpha$.
  This can be shown by inspecting each of the conjuncts of
  $\alpha$. 
\end{proof}

We next show that the translation with dummy facts defined in
\autoref{def:dummy-trans} produces the same traces as $\sem{P}$,
excluding traces not consistent with the axioms. For this we define
the function $d$ which removes any dummy fact from an execution, i.e.,
$$d(\emptyset
\stackrel{F_1}{\longrightarrow} S_1 \stackrel{F_2}{\longrightarrow}
\ldots \stackrel{F_{n}}{\longrightarrow} S_n) = \emptyset
\stackrel{F_1}{\longrightarrow} S_1' \stackrel{F_2}{\longrightarrow}
\ldots \stackrel{F_{n}}{\longrightarrow} S_n'$$ where
$S_i'=S_i\msetminus \cup_{t\in\Terms}\, \dummy(t)$.

\begin{lemma}\label{lem:sametracesfordummytrans}
  Given a ground process $P$, we have that
  \[ \filter(\execmsr(\sem{P}))=\filter(d(\execmsr(\sem{P}^D\cup \MD)))\]
\end{lemma}
\begin{proof}
  The only rules in $\sem{P}^D$ that differ from $\sem{P}$ are
  translations of insert and lookup. The first one only adds a
  permanent fact, which by the definition of $d$, is removed when
  applying $d$. The second one requires a fact $\dummy(t)$, whenever
  the rule is instantiated such the actions equals
  $\mathsf{IsIn}(s,t)$ for some $s$. Since the translation is otherwise
  the same, we have that
  \[ \filter(d(\execmsr(\sem{P}^D\cup \MD))) \subseteq
  \filter(\execmsr(\sem{P})) \] For any trace in
  $\filter(d(\execmsr(\sem{p}\cup \MD)))$ and any action
  $\mathsf{IsIn}(s,t)$ in this trace, there is an earlier action
  $\mathsf{Insert}(s',t')$ such that $s=s'$ and $t=t'$, as otherwise
  $\AssSetIn$ would not hold. Therefore the same trace is part of
  $\filter(d(\execmsr(\sem{p}^D\cup \MD)))$, as this means that
  whenever $\dummy(t)$ is in the premise, $\dummy(t')$ for $t=t'$ has
  previously appeared in the conclusion. Since it is a permanent fact,
  it has not disappeared and therefore
  \[\filter(d(\execmsr(\sem{P}^D\cup \MD))) \subseteq
  \filter(\execmsr(\sem{P}))\]
\end{proof}

We slightly abuse notation by defining $\filter$ on executions to filter out all traces contradicting the axioms, see \autoref{def:filter}.

\begin{lemma}\label{lem:freshnotded}
  Let $P$ be a ground process and $\emptyset
  \stackrel{F_1}{\longrightarrow} S_1 \stackrel{F_2}{\longrightarrow}
  \ldots \stackrel{F_{n}}{\longrightarrow} S_n \in
  \filter(\execmsr(\sem{P}))$. For all $1 \leq i \leq n$, if $\Fr(a) \in S_i$
  and $F(t_1,\ldots,t_k)\in S_i$ for any $F\in\FSign\setminus \set{\Fr}$, then $a \not \in
  \cap_{t =_E t'} \names(t')$, for any $t\in\{t_1,\ldots,t_k\}$.
\end{lemma}
\begin{proof}
  The translation with the dummy fact introduced in
  \autoref{app:solution} will make this proof easier as for
  $\sem{P}^D\cup\MD$, we have that the third condition of
  \autoref{def:msr-system} holds, namely,
  \begin{equation}
    \label{eq:thirdcondition}
    \forall l'\msrewrite {a'} r' \in_\ET \ginsts(l\msrewrite
    a r): \cap_{r'' =_\ET r'}\names(r'') \cap \FN \subseteq
    \cap_{l'' =_\ET l'}\names(l'')\cap \FN
  \end{equation}
  We will show that the statement holds for all
  $\emptyset \stackrel{F_1}{\longrightarrow} S_1 \stackrel{F_2}{\longrightarrow}
  \ldots \stackrel{F_{n}}{\longrightarrow} S_n \in
  \filter(\execmsr(\allowbreak \sem{P}^D\allowbreak \cup \MD))$, which
  implies the claim by \autoref{lem:sametracesfordummytrans}. We
  proceed by induction on $n$, the length of the execution. 
  \begin{itemize}
  \item Base case, $n=0$.  We have that $S_0=\emptyset$ and therefore the
    statement holds trivially.
  \item Inductive case, $n \geq 1$.   We distinguish two cases.
    \begin{enumerate}
    \item A rule that is not \Fresh was applied and there is a fact
      $F(t_1,\ldots, t_k)\in S_n$, such that $F(t_1,\ldots,t_k)\notin
      S_{n-1}$, and $\Fr(a)\in S_n$ such that $a \in \cap_{t_i =_E t'}
      \names(t')$ for some $t_i$. (If there are no such
      $F(t_1,\ldots,t_k)$ and $\Fr(a)$ we immediately conclude by
      induction hypothesis.)  By \autoref{eq:thirdcondition}, $a\in
      t'_j$ for some $F'(t'_1,\ldots,t'_l)\in S_{n-1}$. Since \Fresh
      is the only rule that adds a \Fr-fact and $\Fr(a)\in S_{n}$, it
      must be that $\Fr(a)\in S_{n-1}$, contradicting the induction
      hypothesis. Therefore this case is not possible.

    \item The rule \Fresh was applied, \ie, $\Fr(a)\in S_n$ and
      $\Fr(a)\notin S_{n-1}$.  If there is no $a \in \cap_{t_i =_E t'}
      \names(t')$ for some $t_i$, and $F(t_1,\ldots,t_k)\in S_{n}$,
      then we conclude by induction hypothesis. Otherwise, if there is
      such a $F(t_1,\ldots,t_k)\in S_{n}$, then, by
      \autoref{eq:thirdcondition}, $a\in t'_j$ for some
      $F'(t'_1,\ldots,t'_l)\in S_{i}$ for $i < n$. We construct a
      contradiction to the induction hypothesis by taking the prefix
      of the execution up to $i$ and appending the instantiation of
      the \Fresh rule to its end. Since $d(\execmsr(\sem{P}^D\cup
      \MD))$ is prefix closed by \autoref{lem:filter-prefix} we have
      that $\emptyset \stackrel{F_1}{\longrightarrow} S_1'
      \stackrel{F_2}{\longrightarrow} \ldots
      \stackrel{F_{i}}{\longrightarrow} S_i \in
      \filter(d(\execmsr(\sem{P}^D\cup \MD)))$. Moreover as rule
      \Fresh was applied adding $\Fr(a)\in S_n$ it is also possible to
      apply the same instance of \Fresh to the prefix
      (by~\autoref{def:execution}) and therefore
      \[
      \emptyset \stackrel{F_1}{\longrightarrow} S_1'
      \stackrel{F_2}{\longrightarrow} \ldots
      \stackrel{F_{i}}{\longrightarrow} S_i \longrightarrow S_i \cup
      \set{\Fr(a)} \in \filter(d(\execmsr(\sem{P}^D\cup \MD)))
      \]
      contradicting the induction hypothesis.      
    \end{enumerate}
  \end{itemize}
\end{proof}

\begin{lemma}\label{lem:dedrestriction}
  For any frame $\nu \tilde n.\sigma$, $t\in\Mess$ and $a\in\FN$, if $a \not
  \in \st(t)$, $a \not \in \st(\sigma)$ and  $\nu \tilde n.\sigma
  \vdash t$, then $\nu \tilde n,a.\sigma\vdash t$.
\end{lemma}
\begin{proof}
  In~\cite[Proposition~1]{AbadiCortierTCS06} it is shown that $\nu \tilde n. \sigma \vdash t$ if and
  only if $\exists M. \mathit{fn}(M) \cap \tilde n =\emptyset$ and $M\sigma
  =_E t$.  Define $M'$ as $M$ renaming $a$ to some fresh name, i.e., not
  appearing in $\tilde n, \sigma,t$. As $a \not \in st(\sigma, t)$ and
  the fact that equational theories are closed under bijective
  renaming of names we have that
  $M'\sigma =_E t$ and $\mathit{fn}(M') \cap (\tilde n,a) = \emptyset$.
  Hence $\nu \tilde n, a. \sigma \vdash t$.
\end{proof}

\begin{lemma}%[Deductive capabilities]
\label{lem:ded}
Let $P$ be a ground process and $\emptyset
\stackrel{F_1}{\longrightarrow} S_1 \stackrel{F_2}{\longrightarrow}
\ldots \stackrel{F_{n}}{\longrightarrow} S_n \in \filter(\execmsr(\sem{P}))$.
% old condition:
% such that $\{\Out(t)\in S \mid t\in\Mess\}=\emptyset$.
Let
\[
\tilde n = \{ a:\freshsort \mid \ProtoNonce(a)\in \bigcup_{1\leq j
\leq n}F_j\},
\]
% \[
% \tilde r = \{ a:\freshsort \mid \RepNonce(a)\in\bigcup_{1\leq j
% \leq n}F_j\},
% \]
\[
\{t_1,\ldots,t_m\}=\{t \mid \mathsf{Out}(t) \in_{1\leq j\leq n} S_j  \}.
\]
Let $\sigma=\{^{t_1} / _{x_1},\ldots,^{t_m} / _{x_m}\}$. We have that
\begin{enumerate}
\item if $\K(t) \in S_n$ then $\nu \tilde n. \sigma \vdash
  t$;
\item if $\nu \tilde n. \sigma \vdash t$ then there exists
  $S$ such that
  \begin{itemize}
  \item $\emptyset \stackrel{F_1}{\longrightarrow} S_1
    \stackrel{F_2}{\longrightarrow} \ldots
    \stackrel{F_{n}}{\longrightarrow} S_n {\longrightarrow}^* S\in \filter(\execmsr_E(\sem{P}))$,
  \item  $\K(t) \in_E S$ and
  \item $S_n \rightarrow^*_{R} S$ for $R =
\MDOutPubFreshApplFresh$.
  \end{itemize}
\end{enumerate}
\end{lemma}
\begin{proof}
  We prove both items separately.
  \begin{enumerate}
  \item The proof proceeds by induction on $n$, the number of steps of
    the execution.

    \paragraph{Base case: n=0.} This case trivially holds as $S_n =
    \emptyset$.

    \paragraph{Inductive case: n>0.} By induction we suppose that if
    $\K(t) \in S_{n-1}$ then $\nu \tilde n'. \sigma' \vdash
    t$ where $\tilde n', \sigma' $ are defined in a similar
    way as $\tilde n, \sigma$ but for the execution of size
    $n-1$. We proceed by case analysis on the rule used to extend the
    execution.
    \begin{itemize}
    \item \textsc{MDOut}. Suppose that $\mathsf{Out}(u) \msrewrite~
      \K(u)\allowbreak \in \ginsts(\textsc{MDOut})$ is the rule used to extend the
      execution. Hence $\mathsf{Out}(u) \in S_{n-1}$ and by definition
      of $\sigma$ there exists $x$ such that $x\sigma = u$. We can
      apply deduction rule {\sc DFrame} and conclude that $\nu \tilde
      n. \sigma \vdash u$.  If $\K(t) \in S_n$ and $t \not =
      u$ we conclude by induction hypothesis as $\tilde n = \tilde n',
     \sigma = \sigma'$.

    \item \textsc{MDPub}.  Suppose that $\msrewrite{~} \mathsf{K}(a:
      \mathit{pub}) \in \ginsts(\textsc{MDPub})$ is the rule used to
      extend the execution. As names of sort $pub$ are never added to
      $\tilde n$  we can apply deduction rule {\sc DName}
      and conclude that $\nu \tilde n. \sigma \vdash a$.  If
      $K(t) \in S_n$ and $t \not = a$ we conclude by induction
      hypothesis as $\tilde n = \tilde n',
      \sigma = \sigma'$.

    \item \textsc{MDFresh}. Suppose that $\Fr(a:\mathit{fresh})
      \msrewrite{~} \mathsf{K}(a:\mathit{fresh}) \in
      \ginsts(\textsc{MDFresh})$ is the rule used to extend the
      execution. By definition of an execution we have that
      $\Fr(a:\mathit{fresh}) \not = (S_{j+1} \setminus S_{j})$ for any
      $j \not = n-1$. Hence $n \not \in \tilde n$. We can
      apply deduction rule {\sc DName} and conclude that $\nu \tilde
      n. \sigma \vdash a$.  If $\K(t) \in S_n$ and $t \not =
      a$ we conclude by induction hypothesis as $\tilde n = \tilde n',
      \sigma = \sigma'$.

    \item \textsc{MDAppl.} Suppose that $\K(t_1),\ldots,\K(t_k)
      \msrewrite{~} \K(u) \in \ginsts(\textsc{MDAppl})$ is the rule
      used to extend the execution. We have that $K(t_1),\ldots,K(t_k)
      \in S_{n-1}$ and $u =_\ET f(t_1,\ldots,t_k)$. By induction
      hypothesis, $\nu \tilde n'.\sigma' \vdash t_i$ for $1 \leq i
      \leq k$. As $\tilde n = \tilde n', \sigma = \sigma'$ we have
      that $\nu \tilde n.\sigma \vdash t_i$ for $1 \leq i \leq k$. We
      can apply deduction rule {\sc DAppl} and conclude that $\nu
      \tilde n. \sigma \vdash f(t_1, \ldots , t_k)$. Hence, $\nu
      \tilde n. \sigma \vdash u$ by rule {\sc DEq}.  If $K(t) \in S_n$
      and $t \not = f(t_1, \ldots , t_k)$ we conclude by induction
      hypothesis as $\tilde n = \tilde n', \sigma = \sigma'$.

    \item If $S_{n-1} {\trans{\ProtoNonce(a) } } S_n$ we have that
      $\Fr(a) \in S_{n-1}$. By \autoref{lem:freshnotded}, we obtain
      that if $\K(t) \in S_{n-1}$ then there exist $t'$ and $\sigma''$
      such that $t' =_\ET t$, $\sigma '' =_\ET \sigma'$ and $a$ $\not
      \in \st(t')$ and $a \not \in \st(\sigma'')$.  For each $\K(u)\in
      S_n$ there is $\K(u)\in S_{n-1}$, and by induction hypothesis,
      $\nu \tilde n'.\sigma' \vdash u$.  By
      \autoref{lem:dedrestriction} and the fact that $\sigma'' =_\ET
      \sigma'$ we have that $\nu \tilde n', a.\sigma' \vdash u$. As
      $\tilde n', a =\tilde n$ and $\sigma' = \sigma$ we conclude.

      % By \autoref{lem:freshnotded}, there is no term
      % $t$ with $\K(t)\in S_n$ such that $\mathit{rep}_i \in \st(t)$
      % as a subterm.  For each $\K(u)\in S_n$ there is $\K(u)\in
      % S_{n-1}$, and by induction hypothesis, $\tilde n',\tilde
      % r'.\sigma \vdash u$. Since, by nonce uniqueness,
      % $\mathit{rep}_i\not\in \tilde r'$ (or $a\not\in\tilde n'$), and
      % none of them can appear in $\sigma'$, we can apply
      % \autoref{lem:dedrestriction} and yield $\tilde n,\tilde r,\sigma
      % \vdash u$.

    \item All other rules do neither add $\K(~)$-facts nor do they
      change $\tilde n$ and may only extend $\sigma$. Therefore we
      conclude by the induction hypothesis.
    \end{itemize}

  \item Suppose that $\nu \tilde n. \sigma \vdash t$. We proceed by
    induction on the proof tree witnessing $\nu \tilde n. \sigma
    \vdash t$.

    \paragraph{Base case.} The proof tree consists of a single
    node. In this case one of the deduction rules \textsc{DName} or
    \textsc{DFrame} has been applied.

    \begin{itemize}
    \item \textsc{DName}. We have that $t \not \in \tilde n$. If $t \in \PN$ we use rule \textsc{MDPub} and we have that
      $S_n \rightarrow S= S_n \cup \{ \K(t) \}$. In case
      $t \in \FN$ we need to consider 3 different cases:
      \begin{inparaenum}[\itshape (i)]
      \item $\K(t) \in S_n$ and we immediately conclude (by letting $S=S_n$),
      \item $\Fr(t) \in S_n$ and applying rule \textsc{MDFresh} we have
        that $S_n \rightarrow S=S_n \cup \{ \K(t) \}$,
      \item $\Fr(t) \not \in S_n$. By inspection of the rules we see
        that $\Fr(t) \not \in S_i$ for any $1 \leq i \leq n$: the only
        rules that could remove $\Fr(t)$ are \textsc{MDFresh} which
        would have created the persistent fact $\K(t)$, or the
        $\ProtoNonce$ rules which would however have added $t$ to
        $\tilde n$. Hence, applying successively rules
        $\textsc{Fresh}$ and \textsc{MDFresh} yields a valid extension
        of the execution $S_n \to S_n \cup \{ \Fr(t) \}\to S = S_n
        \cup \{ \K(t) \}$.
      \end{inparaenum}

    \item \textsc{DFrame}. We have that $x\sigma = t$ for some
      $x\in\Dom(\sigma)$, that is, $t\in\{t_1,\ldots,t_m\}$. By
      definition of $\{t_1,\ldots,t_m\}$, $\Out(t)\in S_i$ for some
      $i\leq n$. If $\Out(t)\in S_n$ we have that $S_n \to
      S = ( S_n \setminus \{ \Out(t) \} ) \cup \{ \K(t) \}$ applying rule
      $\textsc{MDOut}$.  If $\Out(u) \not \in S_n$, the fact that the
      only rule in $\sem P$ that allows to remove an $\Out$-fact is
      $\textsc{MDOut}$, suggests that it was applied before, and thus
      $\K(u)\in S$.
    \end{itemize}

    \paragraph{Inductive case.} We proceed by case distinction on the
    last deduction rule which was applied.

    \begin{itemize}
    \item \textsc{DAppl}. In this case $t=f(t_1,\ldots,t_k)$, such
      that $f\in\Sigma^k$ and $\nu \tilde n \tilde r.\sigma \vdash
      t_i$ for every $i\in\{1,\ldots, k\}$. Applying the induction
      hypothesis we obtain that there are $k$ transition sequences
      $S_n \rightarrow_R^* S^i$ for $1 \leq i \leq k$ which extend
      the execution such that $t_i \in S^i$. All of them only add \K
      facts which are persistent facts. If any two of these extensions
      remove the same $\Out(t)$-fact or the same $\Fr(a)$-fact it also
      adds the persistent fact $\K(t)$, respectively $\K(a)$, and we
      simply remove the second occurrence of the transition. Therefore,
      applying the same rules as for the transitions $S_n \rightarrow^*
      S^i$ (and removing duplicate rules) we have that $S_n
      \to^*S'$ and $\K(t_1),\ldots,\K(t_k) \in S'$. Applying
      rule $\textsc{MDAppl}$ we conclude.

    \item \textsc{DEq}. By induction hypothesis there exists $S$ as
      required with $\K(t') \in_E S$ and $t=_Et'$ which allows us to
      immediately conclude that $\K(t) \in_E S$.
    \end{itemize}

  \end{enumerate}
\end{proof}

\begin{lemma}
\label{lem:ded-eq}
	If $\nu \tilde n. \sigma \vdash t$, $\tilde n=_\ET\tilde n'$,
$\sigma=_\ET \sigma'$ and $t=_\ET t'$, then $\nu \tilde n'.\sigma' \vdash
t'$.
\end{lemma}
\begin{proof}
Assume  $\nu \tilde n. \sigma \vdash t$.
	Since an application of \textsc{DEq} can be appended to the leafs of
its proof tree, we have $\nu \tilde n. \sigma' \vdash t$. Since
\textsc{DEq} can be applied to its root, we have $\nu \tilde n. \sigma'
\vdash t'$. Since $\tilde n, \tilde n'$ consist only of names, $\tilde
n=\tilde n'$ and thus $\nu \tilde n'.\sigma' \vdash t'$.
\end{proof}

To state our next lemma we need two additional definitions.
\begin{definition}
\label{def:trans-at-pos}
  Let $P$ be a well-formed ground process and $p_t$ a position in
  $P$. We define the set of multiset rewrite rules generated for position
$p_t$ of $P$, denoted $\sem P_{=p_t}$ as follows:
\[ \sem P_{=p_t} := \sem{P,[],[]}_{=p_t} \]
where $\sem{\cdot,\cdot,\cdot}_{=p_t}$ is defined in Figure~\ref{fig:transatp}.
\end{definition}

\begin{figure*}
\renewcommand{\arraystretch}{1.5}
 \renewcommand{\theactualrule}[1]{\ensuremath{\{\,#1\,\}_{p \stackrel{?}{=}p_t}}}
 \renewcommand{\underscorethingy}{\ensuremath{_{=p_t}}}

\input{def-trans.tex}
\caption{Definition of $\sem{P,p,\tilde x}_{=p_t}$ where $\{\cdot\}_{a\stackrel{?}{=}b}=\set{\cdot}$ if $a=b$ and
$\emptyset$ otherwise.}\label{fig:transatp}
\end{figure*}

The next definition will be useful to state that for a process $P$ every fact of the
form $\state_p(\tilde t)$ in a multiset rewrite execution of $\sem P$ corresponds
to an active process in the execution of $P$ which is an instance of
the subprocess $P|_p$.

\begin{definition}
\label{def:process-bijection}
Let $P$ be a ground process, \calP be a multiset of processes and $S$
a multiset of multiset rewrite rules. We write ${\cal P}
\leftrightarrow_P S$ if there exists a bijection between \calP and
the multiset $\{ \state_p(\tilde t) \mid \exists p, \tilde
t.\ \state_p(\tilde t) \melem S \}^\#$ such that whenever $Q \melem
\calP$ is mapped to $\state_p(\tilde t) \melem S$ we have that
\begin{enumerate}
\item $P|_p\tau = Q\rho$, for some substitution $\tau$ and some
  bijective renaming $\rho$ of fresh, but not bound names in $Q$, and
\item $\exists ri \in_\ET \ginsts(\sem{P}_{=p}).\ \state_p(\tilde t)
  \in \prems(ri)$.
\end{enumerate}
\end{definition}
When $\calP \leftrightarrow_P S$, $Q \melem \calP$ and
$\state_p(\tilde t) \melem S$ we also write $Q \leftrightarrow_P
\state_p(\tilde t)$ if this bijection maps $Q$ to $\state_p(\tilde
t)$.

\begin{remark}
  \label{remark:properties-bijection}
  Note that $\leftrightarrow_P$ has the following properties (by the
  fact that it defines a bijection between multisets).
  \begin{itemize}
  \item  If $\calP_1 \leftrightarrow_P S_1$ and $\calP_2 \leftrightarrow_P S_2$ then
    $\calP_1 \mcup \calP_2 \leftrightarrow_P S_1 \mcup S_2$.
  \item If $\calP_1 \leftrightarrow_P S_1$ and $Q \leftrightarrow_P
    \state_p(\tilde t)$ for $Q \in \calP_1$ and $\state_p(\tilde t) \in
    S_1$ (i.e. $Q$ and $\state_p(\tilde t)$ are related by the bijection
    defined by $\calP_1 \leftrightarrow_P S_1$) then $\calP_1 \msetminus
    \{Q\} \leftrightarrow_P S_1 \msetminus \{ \state_p(\tilde t) \}$.
  \end{itemize}
\end{remark}

\begin{proposition}
\label{prop:total-order}
Let $A$ be a finite set, $<$  a strict total order on $A$ and $p$ a predicate
on elements of $A$. We have that
$$
\begin{array}{rrl}
  \forall i\in A.p(i) &\Leftrightarrow& \forall i\in A.\ p(i)\vee \exists j\in A.\ i<j
  \wedge \neg p(j) \\
  &( \Leftrightarrow & \forall i\in A.\ \neg p(i)\rightarrow \exists
  j\in A.\ i<j
  \wedge \neg p(j))
\end{array}
$$
and
$$
\begin{array}{rrl}
  \exists i\in A.p(i)  &\Leftrightarrow & \exists i\in A .p(i) \wedge
  \forall j \in A.\ i< j \to\neg p(j)
\end{array}
$$
\end{proposition}

\begin{lemma}
\label{lem:inclusion-apip-in-msr}
  Le $P$ be a well-formed ground process. If
  \[
  \apipstate{0}
  \stackrel{E_1}{\longrightarrow}
  \apipstate{1}
  \stackrel{E_2}{\longrightarrow}
  \ldots \stackrel{E_n}{\longrightarrow}
  \apipstate{n}
  \]
  where  $\apipstate{0} = (\emptyset,\emptyset,\emptyset,\set{P},\emptyset,\emptyset)$
  then  there are $(F_1,S_1),\ldots,(F_{n'},S_{n'})$ such that
  \[
  \emptyset \stackrel{F_1}{\longrightarrow}_\sem P S_1
  \stackrel{F_2}{\longrightarrow}_\sem P \ldots
  \stackrel{F_{n'}}{\longrightarrow}_\sem P S_{n'} \in \execmsr(\sem
  P)\]
  and there exists a monotonic, strictly increasing function $f\colon \setN_n
  \to \setN_{n'}$ such that $f(n)=n'$ and for all $i \in \setN_n$
  \begin{enumerate}
  \item \label{e:nonce}
    \invariantNonce i

  \item \label{e:storea}
    $\forall t\in\Mess.\:\StoreA_i(t) =
    \begin{cases}
      u &
      \mbox{if } \exists j \leq f(i). \mathrm{Insert}(t,u) \in
      F_j \\
      & \qquad \wedge \forall j', u'. j < j' \leq f(i) \to
      \mathrm{Insert}(t,u') \not \in_\ET F_{j'} \wedge
      \mathrm{Delete}(t) \not \in_\ET F_{j'}\\
      \bot & \text{otherwise}
    \end{cases}$

  \item \label{e:storeb}
    $\StoreB_i = S_{f(i)} \msetminus \ReservedFacts$

  \item \label{e:proc}
    $\Processes_i \leftrightarrow_P S_{f(i)}$

  \item \label{e:subst} \invariantSubst{i}

  \item \label{e:lock}
    $\calL_i =_\ET \set{ t \mid \exists j\leq f(i),u .\
    \mathrm{Lock}(u,t) \in_\ET F_j
    \wedge \forall j < k \leq f(i). \mathrm{Unlock}(u,t) \not \in_\ET F_k}$

  \item \label{e:axioms} $[F_1,\ldots,F_{n'}]\vDash \alpha$ where
    $\Ass$ is defined as in Definition~\ref{def:transformula}.

  % \item \label{e:actions}
  %   $E_i=F_{f(i)}$ and $\cup_{f(i-1) < j < f(i)-1} F_j  \subseteq
  %   \ReservedFacts$

  \item \label{e:actions}
    $\exists k.\ f(i-1) < k \leq f(i) \mbox{ and } E_i=F_k \mbox{ and
    }\cup_{f(i-1) < j\not = k \leq f(i)} F_j  \subseteq
    \ReservedFacts$

  \end{enumerate}

\end{lemma}
\begin{proof}
  \renewcommand\itemautorefname{Condition}

We proceed by induction over the number of transitions $n$.
\begin{component}[Base Case]
  For $n=0$, we let $f(n)=1$ and $S_1$ be the multiset obtained by using
  the Rule \textsc{Init}:
  \[
  \emptyset
  \stackrel{\mathsf{Init}}{\longrightarrow}
  \mset{\state_{[]}()}
  \]
\autoref{e:nonce}, \autoref{e:storea}, \autoref{e:storeb},
\autoref{e:subst}, \autoref{e:lock},
\autoref{e:axioms}
and
\autoref{e:actions}
 hold trivially.
 To show that \autoref{e:proc} holds, we have to show that
 $\Processes_0\leftrightarrow_P
 \mset{\state_{[]}()}$.  Note that
 $\Processes_0=\mset{P}$. We choose the bijection such that $P
 \leftrightarrow_P \state_{[]}()$.  For
 $\tau=\emptyset$ and $\rho = \emptyset$ we have that
 $P|_{[]}\tau=P\tau=P\rho$.  By \autoref{def:trans-at-pos},
 $\sem{P}_{=[]}=\sem{P,[],[]}_{=[]}$.  We see
 from Figure~\ref{fig:transatp} that for every $P$ we have that
 $\state_{[]}() \in
 \prems(\sem{P,[],[]}_{=[]}$). Hence, we conclude
 that there is a ground instance $\ri\in_E \ginsts(\sem{P}_{=[]})$
 with $\state_{[]}()\in\prems(\ri)$.
\end{component}

\begin{component}[Inductive step]
Assume the invariant holds for $n-1\geq0$. We have to show that the lemma
holds for $n$ transitions
\[
  \apipstate{0}
  \stackrel{E_1}{\longrightarrow}
  \apipstate{1}
  \stackrel{E_2}{\longrightarrow}
  \ldots \stackrel{E_n}{\longrightarrow}
  \apipstate{n}
  \]
By induction hypothesis, we have that
there exists a
monotonically increasing function from $\setN_{n-1} \to \setN_{n'}$  and an execution
  \[
  \emptyset \stackrel{F_1}{\longrightarrow}_\sem P S_1
  \stackrel{F_2}{\longrightarrow}_\sem P \ldots
  \stackrel{F_{n'}}{\longrightarrow}_\sem P S_{n'} \in \execmsr(\sem
  P)\]
such that Conditions~\ref{e:nonce} to \ref{e:actions} hold. Let $f_p$ be
this function and note that $n'=f_p(n-1)$.
Fix a bijection such that
$\Processes_{n-1} \leftrightarrow_P S_{f_p(n-1)}$. We
will abuse notation by writing $P \leftrightarrow_P \state_p(\tilde t)$, if
this bijection goes from $P$ to $\state_p(\tilde t)$.

We now proceed by case distinction over the type of transition from
$\apipstate{n-1}$ to $\apipstate n$. We will (unless stated otherwise)
extend the previous execution by a number of steps, say $s$, from
$S_{n'}$ to some $S_{n'+s}$, and prove that Conditions~\ref{e:nonce}
to \ref{e:actions} hold for $n$ (since by induction hypothesis, they
hold for all $i<n$) and a function $f \colon \setN_n \to \setN_{n'+s}$
that is defined as follows:
\[
	f(i) := \begin{cases}
		f_p(i) & \text{if $i\in\setN_{n-1}$} \\
		n'+s & \text{if $i=n$}
	\end{cases}
\]

\newcommand{\inclusionAcaseintro}[4]
% #1 - Process name
% #2 - \ri
% #3 - the new S_n
% #4 - the label F_n
{ By induction hypothesis ${\cal P}_{n-1} \leftrightarrow_P S_{n'}$. Let
  $p$ and $\tilde t$ be such that $#1 \leftrightarrow_P
  \state_p(\tilde t)$.
%We have that $P|_p = #1 \tau$ for some $\tau$.
%
By \autoref{def:process-bijection},
there is
 a $\ri\in\ginsts(\sem{P}_{=p})$ such that
$\state_p(\tilde t)$ is part of its premise.
By definition of $\sem{P}_{=p}$, we can choose $\ri=#2$.
We can extend the previous execution by one step using $\ri$, therefore:
  \[
  \emptyset \stackrel{F_1}{\longrightarrow}_\sem P S_1
  \stackrel{F_2}{\longrightarrow}_\sem P \ldots
  \stackrel{F_{n'}}{\longrightarrow}_\sem P S_{n'}
  \trans{#4}_\sem P S_{n'+1}
\in \execmsr(\sem
  P)\]
 with $S_{n'+1}=#3$. %  We define
% $f(i):=f_p(i)$ for $i<n$ and $f(n):=f(n-1)+1$.
%
It is left to show that Conditions~\ref{e:nonce} to \ref{e:actions} hold
for $n$.
}

\begin{mycase} [%{{{ 0
$
(\Names_{n-1}, \StoreA_{n-1}, \StoreB_{n-1}, \Processes_{n-1}=\Processes'
\cup \{0\}, \Subst_{n-1}, \ActiveLocks_{n-1})
 ~\rightarrow
(\Names_{n-1}, \StoreA_{n-1},\StoreB_{n-1}, \Processes', \Subst_{n-1},
\ActiveLocks_{n-1})$]

\inclusionAcaseintro
{0}
{[\state_p(\tilde t)]\msrewrite~ []}
{ \set{S_{f(n-1)}\setminus\{\state_p(\tilde t) }}
{}

\autoref{e:nonce},
\autoref{e:storea},
\autoref{e:storeb},
\autoref{e:subst},
\autoref{e:lock},
\autoref{e:axioms},
and
\autoref{e:actions}
hold trivially.

\autoref{e:proc} holds because
$\Processes'=\Processes_{n-1} \setminus^\# \{0\}$,
$S_{f(n)}= S_{f(n-1)} \msetminus \mset{\state_p(\tilde t)} $, and  $0
\leftrightarrow_P \state_p(\tilde t)$ (see
\autoref{remark:properties-bijection}).

\end{mycase}%}}}

\begin{mycase}[%{{{ P|Q
$(\Names_{n-1}, \StoreA_{n-1}, \StoreB_{n-1}, \Processes_{n-1}=\Processes'
\cup \{Q | R\}, \Subst_{{n-1}}, \ActiveLocks_{n-1})
~\rightarrow
(\Names_{n-1}, \StoreA_{n-1}, \StoreB_{n-1}, \Processes'\cup \{Q,R\},\\ \Subst_{n-1},
\ActiveLocks_{n-1})$]

\inclusionAcaseintro
{Q|R}
{[\state_p(\tilde t)]\msrewrite~ [\state_{p\cdot 1}(\tilde
t),\state_{p\cdot 2}(\tilde t)]}
{S_{f(n-1)}\setminus \mset{ \state_p(\tilde
t)} \cup \mset{ \state_{p\cdot 1}(\tilde t),\state_{p\cdot 2}(\tilde t)
}}
{}

\autoref{e:nonce},
\autoref{e:storea},
\autoref{e:storeb},
\autoref{e:subst},
\autoref{e:lock},
\autoref{e:axioms}
and
\autoref{e:actions}
hold trivially. We now show that \autoref{e:proc} holds.

\autoref{e:proc} holds because
$\Processes_n=\Processes_{n-1} \setminus^\# \{Q|R\} \cup^\# \{Q,R\}$,
$\{ Q \}
\leftrightarrow_P \{ \state_{p\cdot 1}(\tilde x) \}$ and $\{ R \}
\leftrightarrow_P \{ \state_{p\cdot 2}(\tilde x)\}$ (by definition of the
translation) (see \autoref{remark:properties-bijection}).

\end{mycase}%}}}

\begin{mycase}[%{{{ !P
$(\Names_{n-1}, \StoreA_{n-1}, \StoreB_{n-1}, \Processes_{n-1}=\Processes'
\cup \{! Q \}, \Subst_{{n-1}},\ActiveLocks_{n-1})
~\rightarrow
(\Names_{n-1}, \StoreA_{n-1},\StoreB_{n-1}, \Processes'\cup \{! Q , Q\},
\\
\Subst_{n-1}, \ActiveLocks_{n-1})$]

Let $p$ and $\tilde t$ such that
$!^i Q \leftrightarrow_P  \state_p(\tilde t)$.
By \autoref{def:process-bijection}, there is
a $\ri\in\ginsts(\sem{P}_{=p})$ such that
$\state_p(\tilde t)$ is part of its premise.
By definition of $\sem{P}_{=p}$, we can choose $\ri=
[\state_p(\tilde t) ]\msrewrite~
[\state_{p}(\tilde t), \state_{p\cdot 1}(\tilde t)] $.
We can extend the previous execution by 1 step using
\ri, therefore:
  \[
  \emptyset \stackrel{F_1}{\longrightarrow}_\sem P S_1
  \stackrel{F_2}{\longrightarrow}_\sem P \ldots
  \stackrel{F_{n'}}{\longrightarrow}_\sem P S_{n'}
  \stackrel{(\ri)}{\rightarrow_\sem P} S_{n'+1}
\in \execmsr(\sem
  P)\]
 with
	$S_{n'+1}=S_{f(n)}\mcup \mset{ \state_{p\cdot 1}(\tilde t)}$.
% We define
% $f(i):=f_p(i)$ for $i<n$ and $f(n):=f(n-1)+2$.
%
\autoref{e:proc} holds because $\Processes_n=\Processes_{n-1} \cup^\#
\{Q\}$ and $\{ Q \} \leftrightarrow_P \{ \state_{p\cdot 1}(\tilde t) \}$ (by definition of $\sem{P}_{=p}$).
\autoref{e:nonce},
\autoref{e:storea},
\autoref{e:storeb},
\autoref{e:subst},
\autoref{e:lock},
\autoref{e:axioms}
and
\autoref{e:actions}
hold trivially.
\end{mycase}%}}}

\begin{mycase}[%{{{ new a
$(\Names_{n-1}, \StoreA_{i_n-1},\StoreB_{i_n-1},
\Processes_{n-1}=\Processes' \cup \set{\nu a; Q }, \Subst_{{n-1}},
\ActiveLocks_{n-1})
~\rightarrow
(\Names_{n-1}\cup\{a'\}, \StoreA_{i_n-1},\StoreB_{i_n-1},\\
\Processes'\cup \{Q\{^{a'} / _a\}\},
\Subst_{n-1}, \ActiveLocks_{n-1})$ for a fresh $a'$]
Let $p$ and $\tilde t$ be such that
$\{\nu a; Q\} \leftrightarrow_P  \state_p(\tilde t)$.
There is a $\ri\in\ginsts(\sem{P}_{=p})$ such that
$\state_p(\tilde t)$ is part of its premise.
By definition of $\sem{P}_{=p}$, there is a $\ri\in\ginsts(\sem{P}_{=p})$, $\ri=
 [\state_p(\tilde t), \Fr(a':\freshsort) ]
\msrewrite{\ProtoNonce(a':\freshsort)} [\state_{p\cdot
1}(\tilde t,a':\freshsort)] $.
Assume there is an $i<n'$ such that $\Fr(a')\in S_i$. If $\Fr(a')\in
S_n$, then we can remove the application of the instance of
\textsc{Fresh} that added $\Fr(a')$ while still preserving
Conditions~\ref{e:nonce} to \ref{e:actions}. If $\Fr(a')$ is consumned
at some point, by the definition of $\sem{P}$, the transition where it
is consumned is annotated either $\ProtoNonce(a')$ or
$\mathit{Lock(a',t)}$ for some $t$. In the last case, we can apply a
substitution to the execution that substitutes $a$ by a different
fresh name that never appears in $\cup_i\le n' S_i$. The conditions we
have by induction hypothesis hold on this execution, too, since
$\mathit{Lock}\in\ReservedFacts$, and therefore \autoref{e:actions} is
not affected. The first case implies that $a'\in\Names_{n-1}$,
contradicting the assumption that $a'$ is fresh with respect to the
process execution.
Therefore, without loss of generality, the previous execution does not
contain an $i<n'$ such that $\Fr(a')\in S_i$, and
we can extend the previous execution by two steps using
the \textsc{Fresh} rule and \ri, therefore:
  \[
  \emptyset \stackrel{F_1}{\longrightarrow}_\sem P S_1
  \stackrel{F_2}{\longrightarrow}_\sem P \ldots
  \stackrel{F_{n'}}{\longrightarrow}_\sem P S_{n'}
  \stackrel{(\textsc{Fresh})}{\rightarrow_\sem P} S_{n'+1}
  \stackrel{(\ri)}{\rightarrow_\sem P} S_{n'+2}
\in \execmsr(\sem
  P)\]
 with
	$S_{n'+1}=S_{n'}\mcup \mset{ \Fr(a':\freshsort)}$
and
	$S_{n'+2}=S_{f(n)} = S_{n'} \mcup \mset{ \state_{p\cdot 1}(\tilde t,
a':\freshsort)}$.
 We define
$f(i):=f_p(i)$ for $i<n$ and $f(n):=f(n-1)+2$.
We now show that \autoref{e:proc} holds. As by induction hypothesis
$\nu a; Q \leftrightarrow_P \state_{p\cdot 1}(\tilde t)$ we also have
that $P|_p\sigma = \nu a;Q\rho$ for some $\sigma$ and
$\rho$. Extending $\rho$ with $\{a' \mapsto a\}$ it is easy to see
from definition of $\sem{P}_{=p}$ that $\{ Q\{^{a'} / _a\} \}
\leftrightarrow_P \{ \state_{p\cdot 1}(\tilde t, a')\}$. As
$\Processes_n=\Processes_{n-1} \msetminus \{\nu a; Q\} \mcup
\mset{Q\{^{a'} / _a\}}$, we also immediately obtain that $\Processes_n
\leftrightarrow_P  S_{f(n)}$.
Since $a'$ is fresh, and therefore $\set{a'}=\Names_n \setminus \Names_{n-1}$, and
$F_n=\ProtoNonce(a')$, \autoref{e:nonce} holds.
\autoref{e:storea},
\autoref{e:storeb},
\autoref{e:subst},
\autoref{e:lock},
\autoref{e:axioms}
and
\autoref{e:actions}
hold trivially.

\end{mycase}%}}}

\begin{mycase}[%{{{ silent deduction
$(\Names_{n-1}, \StoreA_{n-1}, \StoreB_{n-1}, \Processes_{n-1}, \Subst_{{n-1}}, \ActiveLocks_{n-1})
\stackrel{K(t)}{\longrightarrow}
(\Names_{n-1}, \StoreA_{n-1}, \StoreB_{n-1}, \Processes_{n-1}, \Subst_{{n-1}}, \ActiveLocks_{n-1})$]

This step requires that $\nu \Names_{n-1}. \Subst_{n-1} \vdash t$.
% Since all rules in $\sem{P}$ that are annotated with $\ProtoNonce$
% have a $\Fr$-fact as a premise, we know that $\Names_{n-1}=\set{a \mid
%   \ProtoNonce(a) \in \cup_{j\leq f(n-1} F_j }$ and $\tilde r:=\set{a
%   \mid \RepNonce(a) \in \cup_{j\leq f(n-1)} F_j }$ are distinct.
% Since $P$ is well-formed, the translated rules $\sem{P}$ assure that
% no $\Out(u)$ is produced, where $u$ contains a fresh name $a\in\tilde
% r$, as those names are reserved.  Therefore, no $a \in \tilde r$
% appears in $\Subst$.  If any $a\in\tilde r$ appears in $t$, we
% substitute $a$ by a new fresh $a'\in\FN$ that does not appear in $t$
% throughout the execution $\emptyset
% \stackrel{F_1}{\longrightarrow}_\sem \ldots
% \stackrel{F_{n'}}{\longrightarrow}_\sem P S_{n'} $ . This altered
% execution is still valid with respect to $\sem{P}$ and fulfills the
% conditions of the induction hypothesis, as they are invariant with
% respect to such a renaming (see previous case).
% Therefore, we
% assume without loss of generality that $\tilde r$ and the fresh names in
% $t$ are distinct.
% We can apply \autoref{lem:dedrestriction} for each
% $a'\in\tilde r$, hence $\nu \Names_{n-1}, \tilde r . \Subst \vdash t$.
From \autoref{lem:ded} follows that there is an execution $\emptyset
\trans{F_1} S_1 \trans {F_2} \ldots \trans {F_{n'}} S_{n'} \rightarrow^* S \in
\execmsr_E(\sem{P})$ such that $\K(t)\in_E S$ and $S_{n'} \rightarrow_R^* S$ for
$R=\MDOutPubFreshAppl$.

From $S$, we can go one further step using \textsc{MDIn}, since $\K(t)\in S$:
  \[
  \emptyset \trans{F_1}_\sem P S_1
  \trans{F_2}_\sem P \ldots
  \trans{F_{n'}}_\sem P S_{n'}
  \trans{}_{R\subset \sem{P}}^* S=S_{n'+s-1}
  \trans{K(t)}_\sem P S_{n'+s}
\in \execmsr(\sem
  P)\]
where $S_{n'+s}=S\cup \{\mathsf{In}(t)\}$.

From the fact that $S_{f(n-1)}
\rightarrow_R^* S_{f(n)} = S$, and the induction hypothesis, we can conclude that
\autoref{e:actions} holds.
\autoref{e:proc} holds since $\Processes_n=\Processes_{n-1}$ and no
$\mathtt{state}$-facts where neither removed nor added.
\autoref{e:nonce},
\autoref{e:storea},
\autoref{e:storeb},
\autoref{e:subst},
\autoref{e:lock}
and
\autoref{e:axioms}
hold trivially.
\end{mycase}%}}}

\begin{mycase}[%{{{ out
  $(\Names_{n-1}, \StoreA_{n-1}, \StoreB_{n-1},
  \Processes_{n-1}=\Processes'\cup\set{\piout(t,t'); Q }, \Subst_{{n-1}}, \ActiveLocks_{n-1})
  \trans{K(t)}
  (\Names_{n-1}, \StoreA_{n-1}, \StoreB_{n-1}, \\
  \Processes'\cup^\# \set{Q}, \Subst_{n-1}\cup
  \{^{t'}/_{x}\}, \ActiveLocks_{n-1})$]
  This step requires that $x$ is fresh and $\nu \Names_{n-1}. \Subst \vdash t$.
%
  % As in the previous case, we can assume without loss of generality
  % that $\tilde r:=\set{a \mid \RepNonce(a) \in \cup_{j\leq f( n-1 )}
  %   F_j}$ is distinct from $\Names_{n-1}$ and the variables in $t$,
  % and thus $\nu \Names_{n-1}, \tilde r. \Subst \vdash t$.
%
Using \autoref{lem:ded}, we have that there is an execution $\emptyset
\trans{F_1} S_1 \trans {F_2} \ldots \trans {F_{f(n)}} S_{f(n-1)} \rightarrow^* S \in
\execmsr_E(\sem{P})$ such that $\K(t)\in_E S$ and $S_{f(n-1)} \rightarrow_R^* S$ for
$R=\MDOutPubFreshAppl$.
Let $p$ and $\tilde t$ such that
$\{\piout(t,t'); Q\} \leftrightarrow_P  \state_p(\tilde t)$.
By \autoref{def:process-bijection},
 there is
 a $\ri\in\ginsts(\sem{P}_{=p})$ such that
$\state_p(\tilde t)$ is part of its premise.
From the definition of $\sem{P}_{=p}$, we see that we can choose
 $\ri =[\state_p(\tilde t),\In(t)]  \msrewrite{\ChannelInEvent(t)}
[\Out(t'),\state_{p\cdot1}(\tilde t)]$.
To apply this rule, we need the fact $\In(t)$.
Since $\nu \Names_{n-1} . \Subst \vdash t$, as mentioned before,
we can apply \autoref{lem:ded}. It follows that there is an execution $\emptyset
\trans{F_1} S_1 \trans {F_2} \ldots \trans {F_{n'}} S_{n'} \rightarrow^* S \in
\execmsr_E(\sem{P})$ such that $\K(t)\in_E S$ and $S_{n'} \rightarrow_R^* S$ for
$R=\MDOutPubFreshAppl$.
From $S$, we can now go two steps further, using \textsc{MDIn} and \ri:
  \[
  \emptyset \stackrel{F_1}{\longrightarrow}_\sem P S_1
   \ldots
  \stackrel{F_{n'}}{\longrightarrow}_\sem P S_{n'}
  \rightarrow_{R\subset \sem{P}}^* S=S_{n'+s-2}
  \trans{K(t)}_\sem P S_{n'+s-1}
  \trans{\ChannelInEvent(t)}_\sem P S_{n'+s}
\in \execmsr(\sem
  P)\]
where
$S_{n'+s-1}=
S\mcup \mset{\In(t)}$ and
 $S_{f(n)}=S\setminus^\# \set{\state_p(\tilde t)} \cup^\#
\set{\Out(t'),\state_{p\cdot 1}(\tilde t)}$.

Taking $k = n'+s-1$ we immediately obtain that
\autoref{e:actions} holds.
Note first that, since $S_{n'}\rightarrow_R S$,
$\mathit{set}(S_{n'})
\setminus \set{\Fr(t),\Out(t) | t\in \Mess} \subset \mathit{set}(S)$
and
$\mathit{set}(S)
\setminus \set{\K(t) | t\in \Mess} \subset \mathit{set}(S_{n'})$.
Since $\Processes_n=\Processes_{n-1} \setminus \{\piout(t,t'); Q\}
\cup \{Q\}$ and $\{ Q \}
\leftrightarrow_P \{ \state_{p\cdot 1}(\tilde t) \}$ (by definition of
$\sem{P}_{=p}$), we have that $\Processes_n \leftrightarrow_P S_{f(n)}$,
\ie, \autoref{e:proc} holds.
\autoref{e:subst} holds since $t'$ was added to $\Subst_{n-1}$ and
$\Out(t)$ added to $S_{f(n-1)}$.
\autoref{e:axioms} holds since $\mathrm{K}(t)$ appears right before
$\ChannelInEvent(t)$.
\autoref{e:nonce},
\autoref{e:storea},
\autoref{e:storeb} and
\autoref{e:lock}
hold trivially.
\end{mycase}%}}}

\begin{mycase}[%{{{ in
$(\Names_{n-1}, \StoreA_{n-1}, \StoreB_{n-1},
\Processes_{n-1}=\Processes'\cup\{\piin(t,N); Q \}, \Subst_{{n-1}}, \ActiveLocks_{n-1})
\rightarrow
(\Names_{n-1}, \StoreA_{n-1}, \StoreB_{n-1}, \\
\Processes'\cup^\# \set{Q\theta}, \Subst_{n-1}, \ActiveLocks_{n-1})$]
This step requires that  $\theta$ is grounding for
$N$ and that $\nu \Names_{n-1}.\Subst_{n-1}\vdash \pair{t}{N\theta}$.
%
% As in the previous cases, we can assume without loss of generality that
% $\tilde r:=\set{a \mid \allowbreak \RepNonce(a) \in \cup_{j\leq f( n-1 )} F_j}$ is distinct from
% $\Names_{n-1}$ and the names in $t$ and $N$, and thus $\nu
% \Names_{n-1}, \tilde r. \Subst \vdash \pair{t}{N\theta}$.
%
Using \autoref{lem:ded}, we have that there is an execution $\emptyset
\trans{F_1} S_1 \trans {F_2} \ldots \trans {F_{f(n-1)}} \allowbreak
S_{f(n-1)} \rightarrow^* S \in \execmsr_E(\sem{P})$ such that
$\K(t)\in_E S$ and $S_{f(n-1)} \rightarrow_R^* S$ for
$R=\MDOutPubFreshAppl$. The same holds for $N\theta$. We can
combine those executions, by removing duplicate instantiations of
\textsc{Fresh}, \textsc{MDFresh} and \textsc{MDOut}. (This is possible
since \K is persistent.) Let $\emptyset
\trans{F_1} S_1 \trans {F_2} \ldots \trans {F_{f(n-1)}} \allowbreak
S_{f(n-1)} \rightarrow_R^* \overline S \in \execmsr_E(\sem{P})$ this combined
execution, and $\K(t),\K(N\theta)\in_E \overline S$.
Let $p$ and $\tilde t$ be such that,
$\piin(t,N); Q \leftrightarrow_P  \state_p(\tilde t)$.
By \autoref{def:process-bijection} there is
 a $\ri\in\ginsts(\sem{P}_{=p})$ such that
$\state_p(\tilde t)$ is part of its premise.
From the definition of $\sem{P}_{=p}$ and the fact that $\theta$ is grounding for
$N\theta$, we have
$\state_p(\tilde t)$ in their premise, namely,
\[ \ri=[\state_p(\tilde t),\In(\pair t {N\theta})]
\msrewrite{\ChannelInEvent(\pair t {N\theta})}
[\state_{p\cdot 1}(\tilde t\cup (\vars(N)\theta)]. \]

From $S_{n'}$, we can first apply the above transition
$S_{n'}\rightarrow_R^* \overline S$, and then, (since $\K(t),\allowbreak \K(N\theta),\allowbreak \state_p(\tilde
x)\in \overline S$), \textsc{MDAppl} for the pair constructer, \textsc{MDIn} and $\ri$:
\begin{align*}
  \emptyset \stackrel{F_1}{\longrightarrow}_\sem P S_1 &
   \ldots
  \stackrel{F_{n'}}{\longrightarrow}_\sem P S_{n'}
  \rightarrow_{R\subset \sem{P}}^* \overline S=S_{n'+s-3} \\
  &
  \trans{(\textsc{MdAppl})}_{\sem P} S_{n'+s-2}
  \trans{K(\pair{t}{N\theta})}_{\sem P} S_{n'+s-1}
  \trans{\ChannelInEvent(\pair{t}{N\theta})}_{\sem P} S_{n'+s}
\in \execmsr(\sem
  P)\text{, where}
\end{align*}
\begin{itemize}
\item since $S_{n'}\rightarrow_R S$, $S$ is such that
  $\mathit{set}(S_{n'})
  \setminus \set{\Fr(t),\Out(t) | t\in \Mess} \subseteq \mathit{set}(S)$,
  $\mathit{set}(S)
  \setminus \set{\K(t) | t\in \Mess} \subseteq \mathit{set}(S_{n'})$, and
  $\K(t),\K(N\theta)\in S$
\item $S_{n'+s-2}= S\mcup \mset{\K(\pair{t}{N\theta})}$,
\item $S_{n'+s-1}= S\mcup \mset{\In(\pair{t}{N\theta})}$,
\item $S_{n'+s}= S\setminus^\# \set{\state_p(\tilde t)} \cup^\#
  \set{\state_{p\cdot 1}(\tilde t\cup (\vars(N)\theta))}$.
\end{itemize}
Letting $k=n'+s-1$ we immediately have that \autoref{e:actions}
holds.

We now show that \autoref{e:proc} holds. Since by induction
hypothesis, $\piin(t,N); Q \leftrightarrow_P \state_p(\tilde t)$, we
have that $P|_p \tau = \piin(t,N); Q\rho$ for some $\tau$ and
$\rho$. Therefore we also have that $P|_{p\cdot 1} \tau \cup
(\theta\rho) = Q\rho (\theta\rho)$ and it is easy to see from
definition of $\sem{P}_{=p}$ that $\{ Q\theta \} \leftrightarrow_P \{
\state_{p\cdot 1}(\tilde t, (\vars(N)\theta))\}$. Since
$\Processes_n=\Processes_{n-1} \msetminus \{\piin(t,N); Q\} \mcup
\{Q\}$, we have that $\Processes_n \leftrightarrow_P S_{f(n)}$, \ie,
\autoref{e:proc} holds.
\autoref{e:axioms} holds since $\mathrm{K}\pair t {N\theta})$ appears right before
$\ChannelInEvent\pair t {N\theta})$.
\autoref{e:nonce},
\autoref{e:storea},
\autoref{e:storeb},
\autoref{e:subst} and
\autoref{e:axioms}
hold trivially.
\end{mycase}%}}}

\begin{mycase}[%{{{ silent input/output
$(\Names, \StoreA, \StoreB,
\Processes \cup \{\piout(c,m);Q\} \cup \{\piin(c',N);R\},\Subst, \ActiveLocks)
~\rightarrow
(\Names, \StoreA, \StoreB, \Processes \cup \set{Q, R \theta },\Subst,
\ActiveLocks)$]
This step requires that $\theta$ grounding
for $N$,
$t =_E N\theta$ and $c=_E c'$.
Let $p,p'$ and $\tilde t, \tilde N$ such that
$\{\piout(c,m); P\} \leftrightarrow_P  \state_p(\tilde t)$,
$\{\piin(c',N); Q\} \leftrightarrow_P  \state_{p'}(\tilde t')$,
and there are
$\ri\in\ginsts(\sem{P}_{=p})$
and
$\ri'\in\ginsts(\sem{P}_{=p'})$
 such that
$\state_p(\tilde t)$ and
$\state_{p'}(\tilde t')$
are part of their respective premise.
From the definition of $\sem{P}_{=p}$ and the fact that $\theta$ is
grounding for $N$, we have:
\begin{align*}
\ri_1 &=&[\state_p(\tilde t)]  &\rightarrow
[\Msg(t,N\theta),\semistate_{p\cdot1}(\tilde t)]\\
 \ri_2 &=&[\state_{p'}(\tilde t'),\Msg(t,N\theta)] &\rightarrow
[\state_{p'\cdot1}(\tilde t'\cup(\mathit{vars}(N)\theta)),\Ack(t,N\theta)]\\
 \ri_3 &=&[\semistate_{p}(\tilde t),\Ack(t,N\theta)] &\rightarrow
[\state_{p\cdot1}(\tilde t)].
\end{align*}
This allows to extend the previous execution by 3 steps:
\[
  \emptyset \stackrel{F_1}{\longrightarrow}_\sem P S_1
   \ldots
  \stackrel{F_{n'}}{\longrightarrow}_\sem P S_{n'}
  \stackrel{(\ri_1)}{\rightarrow_\sem P} S_{n'+s-2}
  \stackrel{(\ri_2)}{\rightarrow_\sem P} S_{n'+s-1}
  \stackrel{(\ri_2)}{\rightarrow_\sem P} S_{n'+s}
\in \execmsr(\sem P)
\]
where:
\begin{itemize}
	\item $S_{n'+s-2}= S_{n'}\msetminus \set{\state_p(\tilde t)}\mcup \mset{\Msg(t,N\theta),\semistate_{p\cdot1}(\tilde t)}$,

	\item $S_{n'+s-1}= S_{n'}\msetminus \set{\state_p(\tilde
t),\state_{p'}(\tilde t')}\mcup
\mset{\semistate_{p\cdot1}(\tilde t),\state_{p'\cdot1}(\tilde
t'\cup(\mathit{vars}(N)\theta)),\allowbreak\Ack(t,N\theta)}$,

	\item $S_{n'+s}= S_{n'}\msetminus \set{\state_p(\tilde
t),\state_{p'}(\tilde t')} \cup^\#
\set{
\state_{p\cdot 1}(\tilde t),
\state_{p'\cdot 1}(\tilde t'\cup (\vars(N)\theta))}$.
\end{itemize}
We have that $\Processes_n=\Processes_{n-1} \msetminus
\set{\piout(c,m);Q,~\piin(c',t'); R} \mcup \set{Q,
  R\theta}^\#$. Exactly as in the two previous cases we have that
$Q \leftrightarrow \state_{p\cdot 1}(\tilde t)$, as well as
$R\theta \leftrightarrow \state_{p'\cdot 1}(\tilde t')$.  Hence we have that, \autoref{e:proc} holds.
\autoref{e:nonce},
\autoref{e:storea},
\autoref{e:storeb},
\autoref{e:subst},
\autoref{e:lock},
\autoref{e:actions}
and
\autoref{e:axioms}
hold trivially.
\end{mycase}%}}}

\begin{mycase}[%{{{ if then - positive
$(\Names_{n-1}, \StoreA_{n-1},\StoreB_{n-1},
\Processes_{n-1}=\Processes' \cup \set{\code{if $t=t'$ then $Q$ else $Q'$}}, \Subst_{{n-1}},
\ActiveLocks_{n-1})
~\rightarrow
(\Names_{n-1}, \StoreA_{n-1},\\
\StoreB_{n-1}, \Processes'\cup \set{Q},
\Subst_{n-1}, \ActiveLocks_{n-1})$]
This step requires that $t=_\ET t'$.
\inclusionAcaseintro
{\code{if $t=t'$ then $Q$ else $Q'$}}
{[\state_p(\tilde t)] \msrewrite{\mbox{Eq}(t,t')} [\state_{p\cdot 1}(\tilde t)]}
{ \set{S_{n'}\msetminus\mset{\state_p(\tilde t)}\mcup \mset{\state_{p\cdot
1}(\tilde t)} }}
{Eq(t,t')}
The last step is labelled $F_{f(n)}=\mset{Eq(t,t')}$. As $t=_E t'$,
\autoref{e:axioms} holds, in particular, \AssEq\ is not violated. Since
$\mathrm{Eq}$ is reserved, \autoref{e:actions} holds as well.

As before, since $\Processes_n=\Processes_{n-1}
\msetminus \set{\text{if $t=t'$ then $Q$ else $Q'$}}
\mcup \{Q\}$ and
$\{ Q \}
\leftrightarrow \{ \state_{p\cdot 1}(\tilde t, a) \}$ (by definition of the
translation), we have that $\Processes_n \leftrightarrow_P
S_{f(n)}$, and therefore \autoref{e:proc} holds.
\autoref{e:nonce},
\autoref{e:storea},
\autoref{e:storeb},
\autoref{e:subst}
and
\autoref{e:lock}
hold trivially.
\end{mycase}%}}}

\begin{mycase}[%{{{ if then - negative
$(\Names_{n-1}, \StoreA_{n-1},\StoreB_{n-1},
\Processes_{n-1}=\Processes' \cup \set{\code{if $t=t'$ then $Q'$ else $Q$}}, \Subst_{{n-1}},
\ActiveLocks_{n-1})
~\rightarrow
(\Names_{n-1}, \StoreA_{n-1},\\
\StoreB_{n-1}, \Processes'\cup \set{Q'},
\Subst_{n-1}, \ActiveLocks_{n-1})$]
This step requires that $t\neq_E t'$.
This proof step is similar to the previous case, except $\ri$ is chosen to be
\[
[\state_p(\tilde t)] \msrewrite{\mbox{NotEq}(t,t')} [\state_{p\cdot
2}(\tilde t)].
\]
The condition in \AssNotEq\ holds since $t \neq_E t'$.
\end{mycase}%}}}

\begin{mycase}[%{{{ event
$(\Names_{n-1}, \StoreA_{n-1},\StoreB_{n-1},
\Processes_{n-1}=\Processes' \cup \set{\code{event($F$);$Q$}}, \Subst_{{n-1}},
\ActiveLocks_{n-1})
\trans{F}
(\Names_{n-1}, \StoreA_{n-1},\StoreB_{n-1},\\
 \Processes'\cup \set{Q},
\Subst_{n-1}, \ActiveLocks_{n-1})$
]
\inclusionAcaseintro
{\code{event($F$);$Q$}}
{[\state_p(\tilde t)] \msrewrite{F,\Event()} [\state_{p\cdot 1}(\tilde t)]}
{S_{n')}
\setminus^\# \{\state_p(\tilde t)\} \mcup \{ \state_{p\cdot 1}(\tilde
t)\}}
{F,\Event()}
\autoref{e:proc} holds because
$\Processes_n=\Processes_{n-1} \setminus^\#
\set{\text{event($F$);$Q$}}
\cup^\# \set{Q} $
and
$\{ Q \} \leftrightarrow \{ \state_{p\cdot 1}(\tilde t) \}$
 (by definition of $\sem{P}_{=p}$). Taking $k=f(n)$ \autoref{e:actions} holds.
\autoref{e:nonce},
\autoref{e:storea},
\autoref{e:storeb},
\autoref{e:subst},
\autoref{e:lock}
and
\autoref{e:axioms}
hold trivially.
\end{mycase}%}}}

\begin{mycase}[%{{{ insert
$(\Names_{n-1}, \StoreA_{n-1},\StoreB_{n-1},
\Processes_{n-1}=\Processes' \cup \set{\code{insert $t,t'$; $Q$}}, \Subst_{{n-1}},
\ActiveLocks_{n-1})
\rightarrow
(\Names_{n-1}, \StoreA_n = \StoreA_{n-1}\lbrack t \mapsto t'\rbrack,\\
 \StoreB_{n-1}, \Processes'\cup \set{Q},
\Subst_{n-1}, \ActiveLocks_{n-1})$]
\inclusionAcaseintro
{\text{insert $t,t'$; $Q$}}
{[\state_p(\tilde t)] \msrewrite{\mathrm{Insert}(t,t')} [\state_{p\cdot
1}(\tilde t)]}
{S_{f(n-1)} \msetminus \mset{ \state_p(\tilde t)} \mcup
\mset{\state_{p\cdot 1}(\tilde t)}}
{\mathrm{Insert}(t,t')}

This step is labelled $F_{f(n)}=\mathrm{Insert}(t,t')$, hence
\autoref{e:actions} holds. To see that \autoref{e:storea} holds we let $j =
f(n)$ for which both conjuncts trivially hold.  Since, by induction
hypothesis, \autoref{e:axioms} holds, i.e., $[ F_1,\ldots F_{n'}]
\vDash\alpha$, it holds for this step too.  In particular, if $[ F_1,\ldots
F_{n'}] \vDash\AssSetIn$ and $[ F_1,\ldots F_{n'}] \vDash\AssSetNotIn$, we
also have that $[ F_1,\ldots F_{n'},F_{f(n)}] \vDash\AssSetIn$ and $[
F_1,\ldots F_{n'},F_{f(n)}] \vDash \AssSetNotIn$: as the $\mathrm{Insert}$-action was added at the last position of the trace, it appears after any
$\mathrm{InIn}$ or $\mathrm{IsNotSet}$-action and by the semantics of the
logic the formula holds.

Since
$\Processes_n=\Processes_{n-1} \setminus^\#
\set{\text{insert $t,t'$; $Q$}}
\cup^\# \set{Q} $
and
$\{ Q \} \leftrightarrow \{ \state_{p\cdot 1}(\tilde t) \}$
 (by definition of $\sem{P}_{=p}$), we have that \autoref{e:proc} holds.
\autoref{e:nonce},
\autoref{e:storeb},
\autoref{e:subst}
and
\autoref{e:lock}
hold trivially.
\end{mycase}%}}}

\begin{mycase}[%{{{ delete
$(\Names_{n-1}, \StoreA_{n-1},\StoreB_{n-1},
\Processes_{n-1}=\Processes' \cup \set{\code{delete $t$; $Q$}}, \Subst_{{n-1}},
\ActiveLocks_{n-1})
\rightarrow
(\Names_{n-1}, \StoreA_n = \StoreA_{n-1}\lbrack t \mapsto \bot\rbrack,\\
 \StoreB_{n-1}, \Processes'\cup \set{Q},
\Subst_{n-1}, \ActiveLocks_{n-1})$]

\inclusionAcaseintro
{\code{delete $t$; $Q$}}
{[\state_p(\tilde t)] \msrewrite{\mathrm{Delete}(t)}
  [\state_{p\cdot 1}(\tilde t)]
}
{ S_{f(n-1)}
\setminus^\# \{ \state_p(\tilde t) \} \mcup \{ \state_{p\cdot 1}(\tilde t)\}}
{\mathrm{Delete}(t)}

This step is labelled $F_{f(n)}=\mathrm{Delete}(t)$, hence
\autoref{e:actions} holds.  Since, by induction hypothesis,
\autoref{e:axioms} holds, i.e., $[ F_1,\ldots F_{n'}] \vDash\alpha$,
it holds for this step too.  In particular, if $[ F_1,\ldots F_{n'}]
\vDash\AssSetIn$ and $[ F_1,\ldots F_{n'}] \vDash\AssSetNotIn$, we
also have that $[ F_1,\ldots F_{n'},F_{f(n)}] \vDash\AssSetIn$ and $[
F_1,\ldots F_{n'},F_{f(n)}] \vDash \AssSetNotIn$: as the
$\mathrm{Insert}$-action was added at the last position of the trace, it
appears after any $\mathrm{InIn}$ or $\mathrm{IsNotSet}$-actions and by
the semantics of the logic the formula holds.

We now show that \autoref{e:storea} holds. We have that $\StoreA_{n}=
\StoreA_{n-1}\lbrack t \mapsto \bot\rbrack$ and therefore, for all
$t'\neq_ET t$, $\StoreA_{n}(x)=\StoreA_{n-1}(x)$. Hence for all such
$t'$ we have by induction hypothesis that  for
some $u$,
$$\exists j \leq n'. \mathrm{Insert}(t',u) \in
      F_j
      \wedge \forall j', u'. j < j' \leq n' \to
      \mathrm{Insert}(t',u') \not \in_\ET F_{j'} \wedge
      \mathrm{Delete}(t') \not \in_\ET F_{j'}$$

As,  $F_{n'+1}\neq_E  \mathrm{Delete}(x,u)$
and, for all $u'\in\Mess$,
$F_{n'+1}\neq_E  \mathrm{Insert}(x,u')$ we also have that
$$ \exists j \leq n'+1. \mathrm{Insert}(t',u) \in
F_j \wedge \forall j', u'. j < j' \leq n'+1 \to \mathrm{Insert}(t',u')
\not \in_\ET F_{j'} \wedge \mathrm{Delete}(t') \not \in_\ET F_{j'}.$$
For $t' =_\ET t$, the above condition can never be true, because
$F_{n'+1}=\mathrm{Delete}(t)$ which allows us to conclude that
\autoref{e:storea} holds.

Since $\Processes_n=\Processes_{n-1} \setminus^\#
\set{\text{delete $t$; $Q$}}
\cup^\# \set{Q} $
and
$\{ P \} \leftrightarrow \{ \state_{p\cdot 1}(\tilde t) \}$
 (by definition of $\sem{P}_{=p}$), we have that \autoref{e:proc} holds.
\autoref{e:nonce},
\autoref{e:storeb},
\autoref{e:subst}
and
\autoref{e:lock}
hold trivially.
\end{mycase}%}}}

\begin{mycase}[%{{{ lookup-positive
$(\Names_{n-1}, \StoreA_{n-1},\StoreB_{n-1},
\Processes_{n-1}=\Processes' \cup \set{
\code{lookup $t$ as $x$ in $Q$ else $Q'$}
}, \Subst_{n-1},
\ActiveLocks_{n-1})
\rightarrow
(\Names_{n-1}, \StoreA_{n-1},\\
 \StoreB_{n-1}, \Processes'\cup \set{Q\{v/x\}},
\Subst_{n-1}, \ActiveLocks_{n-1})$]
This step requires that
$\StoreA_{n-1}(t')=_E v$ for some $t'=_E t$.
\inclusionAcaseintro
{\code{lookup $t$ as $v$ in $Q$ else $Q'$}}
 {[\state_p(\tilde t)]
	\msrewrite{\mathrm{IsIn}(t,v)}
 [\state_{p\cdot 1}(\tilde t,v)]
}
{S_{f(n-1)}
\setminus^\# \{ \state_p(\tilde t) \} \mcup \{ \state_{p\cdot
  1}(\tilde t)\} }
{\mathrm{IsIn}(t,v)}

This step is labelled $F_{f(n)}=\mathrm{IsIn}(t,v)$, hence
\autoref{e:actions} holds.

From the induction hypothesis, \autoref{e:storea}, we
have that there is a $j$ such that $\mathrm{Insert}(t,t')\in_E F_j$,
$j\le n'$ and
\[ \forall j', u'.\ j < j' \leq n' \to
      \mathrm{Insert}(t,u') \not \in_\ET F_{j'} \wedge
      \mathrm{Delete}(t) \not \in_\ET F_{j'} \]
  This can be strengthened, since
$F_{f(n)}=\set{\mathrm{IsIn}(t,v)}$:
\[ \forall j', u'.\ j < j' \leq f(n) \to \mathrm{Insert}(t,u') \not
\in_\ET F_{j'} \wedge \mathrm{Delete}(t) \not \in_\ET F_{j'} \] This
allows to conclude that \autoref{e:storea} holds. From
\autoref{e:storea} it also follows that \autoref{e:axioms}, in
particular \AssSetIn, holds.

We now show that \autoref{e:proc} holds. By induction
hypothesis we have that $\code{lookup } t \code{ as } x \allowbreak \code{ in } Q
\code{ else } Q'
\leftrightarrow_P \state_p(\tilde t)$, and hence
$P|_p \tau = (\code{lookup $t$ as $x$ in $Q$ else $Q'$})\rho$ for some $\tau$ and
$\rho$. Therefore we also have that $P|_{p\cdot 1} \tau \cup
(\{^{v \rho} / _x \}) = Q\rho \{^{v \rho} / _x \})$ and it is easy to see from
definition of $\sem{P}_{=p}$ that $\{ Q\{^v/_x\} \} \leftrightarrow_P \{
\state_{p\cdot 1}(\tilde t, v)\}$. Since $\Processes_n=\Processes_{n-1} \setminus^\#
\set{\code{lookup $t$ as $x$ in $Q$ else $Q'$}}
\cup^\# \set{Q \{^v / _x\}} $
 we have that $\Processes_n \leftrightarrow_P S_{f(n)}$, \ie,
\autoref{e:proc} holds.

\autoref{e:nonce},
\autoref{e:storeb},
\autoref{e:subst}
and
\autoref{e:lock}
hold trivially.
\end{mycase}%}}}

\begin{mycase}[%{{{ lookup-negative
$(\Names_{n-1}, \StoreA_{n-1},\StoreB_{n-1},
\Processes_{n-1}=\Processes' \cup \set{
\code{lookup $t$ as $x$ in $Q$ else $Q'$}
}, \Subst_{n-1},
\ActiveLocks_{n-1})
\rightarrow
(\Names_{n-1},\\
 \StoreA_{n-1}, \StoreB_{n-1}, \Processes'\cup \set{Q'},
\Subst_{n-1}, \ActiveLocks_{n-1})$]
This step requires that $S(t')$ is undefined for all $t'=_E t$.
\inclusionAcaseintro
{\code{lookup $t$ as $x$ in $Q$ else $Q'$}}
{
[\state_p(\tilde t)] \msrewrite{\mathrm{IsNotSet}(t)}
  [\state_{p\cdot 1}(\tilde t)]
}
{S_{f(n-1)}
\setminus^\# \state_p(\tilde t) \mcup \state_{p\cdot 1}(\tilde t)}
{\mathrm{IsNotSet}(t)}

This step is labelled $F_{f(n)}=\mathrm{IsNotSet}(t)$, hence
\autoref{e:actions} holds. \autoref{e:storea} also holds trivially and
will be used to show \autoref{e:axioms}.  Since this step requires
that $S(t')$ is undefined for all $t'=_E t$,
we have by \autoref{e:storea} that
\begin{align*}
  \forall j\le f(n), u  .\ & \mathrm{Insert}(t,u)\in_E F_j \\
 & \rightarrow
\exists j',u'. j<j'\le f(n) \wedge (
\mathrm{Insert}(t,u') \in_E F_{j'}  \vee
\mathrm{Delete}(t) \in_E F_{j'} )
\end{align*}
Now suppose that
$$\exists i\le f(n), y . \mathrm{Insert}(t,y)\in_E
F_i )$$
As there exists an insert, there is a last
insert and hence we also have
$$
\exists i\le f(n), y . \mathrm{Insert}(t,y)\in_E F_i
\quad \wedge \quad
\forall i', y'. i<i'\le f(n) \to \mathrm{Insert}(t,y')
\notin_\ET F_{i'}
$$
Applying \autoref{e:storea} (cf above) we obtain that
\begin{align*}
  \exists i\le f(n), y .& \mathrm{Insert}(t,y)\in_E F_i
\quad \wedge \quad
	 \forall i', y'.\ i<i'\le f(n) \to \mathrm{Insert}(t,y')
\notin_\ET F_{i'}
\\
 \wedge \quad
& \exists j',u'.\ i<j'\le f(n) \wedge (
\mathrm{Insert}(t,u') \in_E F_{j'}  \vee
\mathrm{Delete}(t) \in_E F_{j'} )
\end{align*}
which simplifies to
\begin{align*}
  \exists i\le f(n), y .& \mathrm{Insert}(t,y)\in_E F_i
  \quad \wedge \quad
	 \forall i',y'.\ i<i'\le f(n) \to \mathrm{Insert}(t',y')
\notin F_{i'}
\\
 \wedge \quad
& \exists j'.\  i<j'\le f(n) \wedge
\mathrm{Delete}(t) \in_E F_{j'}
\end{align*}
Now we weaken the statement by dropping the first conjunct and
restricting the quantification $\forall i'. i<i'\le f(n)$ to $\forall
i'. j'<i'\le f(n)$, since $i<j'$.
$$
\exists i\le f(n) .\ \exists j'.\ i<j'\le f(n) \wedge
	 \forall i'.\ j'<i'\le f(n) \to \mathrm{Insert}(t',y')
\notin F_{i'}
 \wedge
\mathrm{Delete}(t) \in_E F_{j'}
$$
We further weaken the statement by weakening the scope of the
existential quantification $\exists j'.\  i<j'\le f(n)$ to $\exists
j'.\  j'\le f(n)$.  Afterwards, $i$ is not needed anymore.
$$
\exists j'.\  j'\le f(n) \wedge
	 \forall i'.\  j'<i'\le f(n) \to \mathrm{Insert}(t',y')
\notin F_{i'}
 \wedge
\mathrm{Delete}(t) \in_E F_{j'}
$$
This statement was obtained under the hypothesis that
$\exists i\le f(n), y . \mathrm{Insert}(t,y)\in_E
F_i )$. Hence we have that
\begin{align*}
\forall i\le f(n), y . &\mathrm{Insert}(t,y)\not\in_E F_i \\
\vee
 \exists j'\le f(n)  .\ &
\mathrm{Delete}(t) \in_E F_{j'}
 \wedge \forall i'.\  j'<i'\le f(n) \to \mathrm{Insert}(t',y')
\notin F_{i'}
\end{align*}
This shows that \autoref{e:axioms}, in particular \AssSetNotIn, holds.

Since
$\Processes_n=\Processes_{n-1} \setminus^\#
\set{\code{lookup $t$ as $x$ in $Q$ else $Q'$}}
\cup^\# \set{Q'} $
and
$\{ Q' \} \leftrightarrow \{ \state_{p\cdot 1}(\tilde t) \}$
 (by definition of $\sem{P}_{=p}$), we have that \autoref{e:proc} holds.
\autoref{e:nonce},
\autoref{e:storeb},
\autoref{e:subst}
and
\autoref{e:lock}
hold trivially.
\end{mycase}%}}}

\begin{mycase}[%{{{ lock
$(\Names_{n-1}, \StoreA_{n-1},\StoreB_{n-1},
 \Processes_{n-1}=\Processes' \cup \set{
\code{lock $t$; $Q$}
}, \Subst_{n-1},
\ActiveLocks_{n-1})
\rightarrow
(\Names_{n-1}, \StoreA_{n-1}, \StoreB_{n-1}, \Processes'\mcup \set{Q'},\\
\Subst_{n-1}, \ActiveLocks_{n-1}\cup \set{t})$]
This step requires that for all $t'=_E t$, $t'\notin \ActiveLocks_{n-1}$.
Let $p$ and $\tilde t$ such that
$\code{lock $t$; $Q$} \leftrightarrow_P  \state_p(\tilde t)$.
By \autoref{def:process-bijection},
there is
 a $\ri\in\ginsts(\sem{P}_{=p})$ such that
$\state_p(\tilde t)$ is part of its premise.
By definition of $\sem{P}_{=p}$, we can choose $\ri=
[\Fr(l),\allowbreak\state_p(\tilde t)]\allowbreak
\msrewrite{\mathrm{Lock}(l, t)}
  [\state_{p\cdot 1}(\tilde t, l)] $
for a fresh name $l$, that never appeared in a $\Fr$-fact in $\cup_{j\leq
f(n-1)} S_j$.
We can extend the previous execution by $s=2$ steps using an instance of
\textsc{Fresh} for $l$ and $\ri$:
  \[
  \emptyset \stackrel{F_1}{\longrightarrow}_\sem P S_1
  \stackrel{F_2}{\longrightarrow}_\sem P \ldots
  \stackrel{F_{n'}}{\longrightarrow}_\sem P S_{n'}
  \trans{~}_{\set{\textsc{fresh} }} S_{n'+s-1}
  \trans{\mathrm{Lock}(l, t)}_\sem P S_{n'+s}
\in \execmsr(\sem
  P)\]
 with
$S_{n'+s-1}=S_{f(n-1)}
\setminus^\# \mset{\state_p(\tilde t)} \mcup \set{\Fr(l)}$
and
$S_{n'+s}=S_{f(n-1)}
\setminus^\# \mset{\state_p(\tilde t)} \mcup \mset{\state_{p\cdot 1}(\tilde t)}$.
It is left to show that Conditions~\ref{e:nonce} to \ref{e:actions} hold
for $n$.

The step from $S_{f(n)-1}$ to $S_{f(n)}$ is labelled
$F_{f(n)}=\mathrm{Lock}(l, t)$, hence
\autoref{e:actions} and \autoref{e:storea} hold.

$F_{f(n)}$ also preserves \autoref{e:lock} for the new set of active
locks $\ActiveLocks_{f(n)}=\ActiveLocks_{f(n-1)} \cup \set{t}$.

In the following we show by contradiction that \AssLock, and therefore
\autoref{e:axioms} holds. \AssLock held in the previous step, and $F_{f(n-1)+1}$ is
empty, so we assume (by contradiction), that
$F_{f(n)}=\mathrm{Lock}(l, t)$ violates \AssLock. If this was
the case, then:
\begin{align}
 \exists i<f(n),l_1 & .\ \mathrm{Lock}(l_1, t)\in_E F_i  \notag \wedge\\
&\wedge \forall j.\ i<j<f(n) \to
\mathrm{Unlock}(l_1,t)\not\in_E F_j
\notag\\
\label{eq:contradiction}
& \quad \vee
\exists l_2, k.\ i<k<j \wedge
(\mathrm{Lock(l_2,t)}\in_E F_k \vee
\mathrm{Unlock(l_2,t)}\in_E F_k)
\end{align}
Since the semantics of the calculus requires that for all $t'=_E t$, $t'\notin \ActiveLocks_{n-1}$, by
induction hypothesis, \autoref{e:lock}, we have that
\begin{align*}
  \forall i<f(n-1),l_1 & .\ \mathrm{Lock}(l_1,t)\in_E  F_i \to \\
  & \exists j.\ i<j<f(n-1) \wedge \mathrm{Unlock}(l_1,t)\in_E F_j
  \intertext{Since $F_{f(n-1)+1}=\emptyset$ and
    $f(n)=f(n-1)+2$, we have:}
  \forall i<f(n),l_1 &. \mathrm{Lock}(l_1,t)\in_E  F_i \to \\
  & \exists j.\ i<j<f(n) \wedge \mathrm{Unlock}(l_1,t)\in_E F_j
\end{align*}
We apply Proposition~\ref{prop:total-order} for the total order $>$
on the integer interval $i+1..f(n)-1$:
\begin{align*}
  \forall i<f(n),l_1.\  & \mathrm{Lock}(l_1,t)\in_E  F_i \to \\
 & \exists j.\ i<j<f(n) \wedge \mathrm{Unlock}(l_1,t)\in_E F_j \\
 & \wedge \quad \forall k.\ i<k<j \to
 \mathrm{Unlock}(l_1,t)\not\in_E F_k
\end{align*}
Combining this with (\ref{eq:contradiction}) we obtain that
\begin{align*}
  \exists i<f(n),l_1 .\ & \mathrm{Lock}(l_1, t)\in_E F_i  \wedge\\
  & \exists j.\ i<j<f(n) \wedge \mathrm{Unlock}(l_1,t)\in_E F_j \\
  & \quad \wedge
  \exists l_2, k.\ i<k<j \wedge
  (\mathrm{Lock(l_2,t)}\in_E F_k \vee
  ( \mathrm{Unlock(l_2,t)}\in_E F_k \wedge l_2\neq_E l_1))
\end{align*}

Fix $i<f(n)$, $j$ such that $i<j<f(n)$, and $l_1$ such that
$\mathrm{Lock}(l_1, t)\in_E F_i$ and $\mathrm{Unlock}(l_1,t)\in_E
F_j$.  Then, there are $l_2$ and $k$ such that $i<k<j$ and either
$\mathrm{Lock(l_2,t)}\in_E F_k$ or $\mathrm{Unlock(l_2,t)}\in_E F_k$,
but $l_2\neq_E l_1$. We proceed by case distinction.

\noindent \underline{Case 1:} there is no unlock in between $i$ and
$j$, \ie, for all $m$, $i<m<j$, $\mathrm{Unlock}(l',t)\not\in
F_m$. Then there is a $k$ and $l_2$ such that
$\mathrm{Lock(l_2,t)}\in_E F_k$. In this case, \AssLock\ is already
invalid at the trace produced by the $k$-prefix of the execution,
contradicting the induction hypothesis.

\noindent \underline{Case 2:} there are $l'$ and $m$, $i<m<j$ such that
$\mathrm{Unlock}(l',t)\in F_m$ (see
\autoref{fig:visualisation-of-case-2-}).

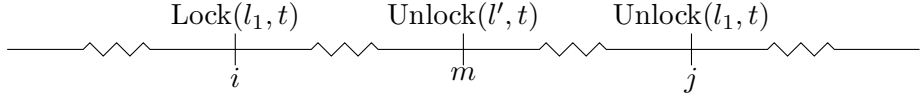
\begin{figure}[h] % (fold)
\centering
\begin{tikzpicture}[scale=3
]

%draw horizontal line
\draw[snake] (0,0) -- (1,0);
\draw[snake] (1,0) -- (2,0);
\draw[snake] (2,0) -- (3,0);
\draw[snake] (3,0) -- (4,0);

%draw vertical lines
\foreach \x in {1,2,3}
   \draw (\x cm,2pt) -- (\x cm,-2pt);

%draw nodes
\draw (1,0) node[below=3pt] {$ i $} node[above=3pt] {$\mathrm{Lock}(l_1,t)$};
\draw (2,0) node[below=3pt] {$ m $} node[above=3pt] {$\mathrm{Unlock}(l',t)$};
\draw (3,0) node[below=3pt] {$ j $} node[above=3pt] {$\mathrm{Unlock}(l_1,t)$};
\end{tikzpicture}
\caption{ Visualisation of Case 2. }
\label{fig:visualisation-of-case-2-}
% figure visualisation-of-case-2- (end)
\end{figure}

We first observe that for any $l, u, i_1, i_2$ , if
$\mathrm{Unlock(l,u)}\in_E F_{i_1}$ and $\mathrm{Unlock(l,u)}\in_E
F_{i_2}$, then $i_1=i_2$. We proceed by contradiction.  By definition of
$\sem{P}$ and well-formedness of $P$, the steps from $i_1-1$ to $i_1$ and
from $i_2-1$ to $i_2$ must be ground instances of rules $\sem{P}_{=q}$ and
$\sem{P}_{=q'}$ such that $P|_{q}$ and $P|_{q'}$ start with unlock commands
that are labelled the same and have the same parameter, since every
variable $\mathit{lock_l}$ in $\sem{P}$ appears in a $\Fr$-fact in the
translation for the corresponding lock command.  By definition of
$\overline P$, this means $q$ and $q'$ have a common prefix $q_l$ that
starts with a lock with this label.

Let $q_l \le q$ denote that $q_l$ is a prefix of $q$.  Since
$\overline P$ gives $\bot$ if there is a replication or a parallel
between $q_l$ and $q$ or $q'$, and since $P$ is well-formed (does not
contain $\bot$), we have that every state fact $\state_r$ for $q_l \le
r \le q$ or $q_l \le r \le q'$ appearing in $\sem{P}$ is a linear
fact, since no replication is allowed between $q_l$ and $q$ or
$q'$. This implies that $q'\neq q$. Furthermore, every rule in
$\cup_{q_l \le r \le q \vee q_l \le r \le q' } \sem{P}_{=r}$ adds at
most one fact $\state_r$ and if it adds one fact, it either removes a
fact $\state_{r'}$ where $r=r'\cdot 1$ or $r'\cdot 2$, or removes a
fact $\semistate_{r'}$ where $r=r'\cdot 1$, which in turn requires
removing $\state_{r'}$ (see translation of \code{out}). Therefore,
either $q\le q'$ or $q' \le q$. But this implies that both have
different labels, and since $\sem{P}_{=q_l}$ requires $\Fr(l)$, and $\ET$
distinguishes fresh names, we
have a contradiction.  (A similiar observation is possible for locks:
For any $l,u,i_1,i_2$ , if $\mathrm{Lock(l,u)}\in_E F_{i_1}$ and
$\mathrm{Lock(l,u)}\in_E F_{i_2}$, then $i_1=i_2$, since by definition
of the translation, the transition from $i_1-1$ to $i_1$ or $i-2-1$ to
$i_2$ removes fact $\Fr(l)$.)

From the first observation we learn that , $l'\neq_E l_1$ for any $l'$ and $m$, $i<m<j$ such that
$\mathrm{Unlock}(l',t)\in F_m$.  We now choose the smallest such $m$. By
definition of $\sem P$, the step from $S_{m-1}$ to $S_m$ must be ground
instance of a rule from $\sem{P}_{=q}$ for $P|_{q}$ starting with
\code{unlock}. Since $P$ is well-formed, there is a $q_l$ such that
$P|_{q_l}$ starts with \code{lock}, with the same label and parameter as
the \code{unlock}. As before, since $P$ is well-formed, and therefore there
are no replications and parallels between $q_l$ and $q$, there must be $n$
such that $\mathrm{Lock}(l',t)\in F_n$ and $n<m$. We proceed again by
case distinction.

\noindent \underline{Case 2a:} $n<i$ (see
\autoref{fig:visualisation-of-case-2a-}). By the fact that $m>i$ we
have that there is no $o$ such that $n<o<i$ and
$\mathrm{Unlock}(l',t)\in_E F_o$ (see first observation). Therefore,
the trace produced by the $i$-prefix of this execution does already
not satisfy \AssLock, \ie, $[F_1,\dots,F_i] \not\vDash \AssLock$.
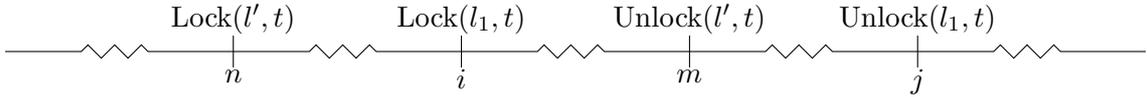
\begin{figure}[h] % (fold)
\centering
\begin{tikzpicture}[scale=3 ]

%draw horizontal line
\draw[snake] (0,0) -- (1,0);
\draw[snake] (1,0) -- (2,0);
\draw[snake] (2,0) -- (3,0);
\draw[snake] (3,0) -- (4,0);
\draw[snake] (4,0) -- (5,0);

%draw vertical lines
\foreach \x in {1,2,3,4}
   \draw (\x cm,2pt) -- (\x cm,-2pt);

%draw nodes
\draw (1,0) node[below=3pt] {$ n $} node[above=3pt] {$\mathrm{Lock}(l',t)$};
\draw (2,0) node[below=3pt] {$ i $} node[above=3pt] {$\mathrm{Lock}(l_1,t)$};
\draw (3,0) node[below=3pt] {$ m $} node[above=3pt] {$\mathrm{Unlock}(l',t)$};
\draw (4,0) node[below=3pt] {$ j $} node[above=3pt] {$\mathrm{Unlock}(l_1,t)$};
\end{tikzpicture}
\caption{ Visualisation of Case 2a. }
\label{fig:visualisation-of-case-2a-}
% figure visualisation-of-case-2a- (end)
\end{figure}

\noindent \underline{Case 2b:} $i<n$ (see
\autoref{fig:visualisation-of-case-2b-}).
Again, \AssLock\ is not satisfied, i.e., $[F_1,\dots,F_n] \not\vDash
\AssLock$, since there is no $o$ such that $i<o<n$ and
$\mathrm{Unlock}(l_1,t)\in_E F_o$.

\begin{figure}[h] % (fold)
\centering
\begin{tikzpicture}[scale=3
]

%draw horizontal line
\draw[snake] (0,0) -- (1,0);
\draw[snake] (1,0) -- (2,0);
\draw[snake] (2,0) -- (3,0);
\draw[snake] (3,0) -- (4,0);
\draw[snake] (4,0) -- (5,0);

%draw vertical lines
\foreach \x in {1,2,3,4}
   \draw (\x cm,2pt) -- (\x cm,-2pt);

%draw nodes
\draw (1,0) node[below=3pt] {$ i $} node[above=3pt] {$\mathrm{Lock}(l_1,t)$};
\draw (2,0) node[below=3pt] {$ n $} node[above=3pt] {$\mathrm{Lock}(l',t)$};
\draw (3,0) node[below=3pt] {$ m $} node[above=3pt] {$\mathrm{Unlock}(l',t)$};
\draw (4,0) node[below=3pt] {$ j $} node[above=3pt] {$\mathrm{Unlock}(l_1,t)$};
\end{tikzpicture}
\caption{ Visualisation of Case 2b. }
\label{fig:visualisation-of-case-2b-}
% figure visualisation-of-case-2b- (end)
\end{figure}

Since we could, under the assumption that Condition~\ref{e:nonce} to
Condition~\ref{e:actions} hold for $i\le n'$, reduce every case in
which $[F_1,\ldots,F_{n'+1}]\not \vDash \AssLock$ to a contradiction,
we can conclude that \autoref{e:axioms} holds for $n'+1$.

Since
$\Processes_n=\Processes_{n-1} \setminus^\#
\set{\text{lock $t$; $Q$}}
\cup^\# \set{Q} $
and
$\{ Q \} \leftrightarrow \{ \state_{p\cdot 1}(\tilde t) \}$
 (by definition of the
translation), we have that \autoref{e:proc} holds.
\autoref{e:nonce},
\autoref{e:storeb}
and
\autoref{e:subst}
hold trivially.

\end{mycase} %}}}

\begin{mycase}[%{{{ unlock
$(\Names_{n-1}, \StoreA_{n-1},\StoreB_{n-1},
\Processes_{n-1}=\Processes' \cup \set{
\code{unlock $t$; $Q$}
}, \Subst_{n-1},
\ActiveLocks_{n-1})
\rightarrow
(\Names_{n-1}, \StoreA_{n-1}, \StoreB_{n-1},\\
\Processes'\mcup \set{Q'},
\Subst_{n-1}, \ActiveLocks_{n-1}\setminus \set{t' \colon t'=_E t})$]
\inclusionAcaseintro
{ \code{unlock $t$; $Q$}}
{
[\state_p(\tilde t)]
\msrewrite{\mathrm{Unlock}(l, t)}
  [\state_{p\cdot 1}(\tilde t)]
}
{ S_{f(n-1)} \setminus^\# \{ \state_p(\tilde t) \} \mcup \{ \state_{p\cdot 1}(\tilde
t)\}}
{\mathrm{Unlock}(l, t)}

The step from $S_{f(n-1)}$ to $S_{f(n)}$ is labelled
$F_{f(n)}=\mathrm{Unlock}(l, t)$, hence
\autoref{e:actions} and \autoref{e:storea} hold.

In order to show that \autoref{e:lock} holds, we perform a case
distinction.  Assume $t\not\in_E \in\ActiveLocks_{n-1}$. Then,
$\ActiveLocks_{f(n-1)}=\ActiveLocks_{f(n)}$. In this case,
\autoref{e:lock} holds by induction hypothesis.  In the following, we
assume $t\in_E\ActiveLocks_{n-1}$.  Thus, there is $j\in n',l'$ such
that $\mathrm{Lock}(l',t)\in_E F_j$ and for all $k$ such that $j<k\le
n'$, $\mathrm{Unlock}(l',t)\not\in_E F_k$.

Since $P|_{p}$ is an unlock node and $P$ is well-formed, there is a prefix
$q$ of $p$, such that $P|_{q}$ is a lock with the same parameter and
annotation. By definition of $\overline P$, there is no parallel and
no replication between $q$ and $p$. Note that
any rule in $\sem{P}$ that produces a state named
$\state_p$ for a non-empty $p$ is such that it requires a fact with name
$\state_{p'}$ for $p=p'\cdot 1$ or $p=p'\cdot 2$ (in case of the
translation of \code{out}, it might require $\semistate_{p'}$, which in
turn requires $\state_{p'}$).
This means that, since $\state_p(\tilde t)\in S_{n'}$, there
is an $i$ such that $\state_q(\tilde t') \in S_i$ and  $\state_q(\tilde t')
\not\in S_{i-1}$ for $\tilde t'$ a prefix to $t$. This rule is an instance
of $\sem{P}_{=q}$ and thus labelled $F_i=\mathrm{Lock}(l,t)$.
We proceed by case distinction.

\begin{figure}[h] % (fold)
\centering
\begin{tikzpicture}[scale=3
]

%draw horizontal line
\draw[snake] (0,0) -- (1,0);
\draw[snake] (1,0) -- (2,0);
\draw[snake] (2,0) -- (3,0);
\draw[snake] (3,0) -- (4,0);

%draw vertical lines
\foreach \x in {1,2,3}
   \draw (\x cm,2pt) -- (\x cm,-2pt);

%draw nodes
\draw (1,0) node[below=3pt] {$ j $} node[above=3pt] {$\mathrm{Lock}(l',t)$};
\draw (2,0) node[below=3pt] {$ i $} node[above=3pt] {$\mathrm{Lock}(l,t)$};
\draw (3,0) node[below=3pt] {$ n'+1 $} node[above=3pt] {$\mathrm{Unlock}(l,t)$};
\end{tikzpicture}
\caption{ Visualisation of Case 1. }
\label{fig:visualisation-of-case-1-}
% figure visualisation-of-case-1- (end)
\end{figure}

\noindent \underline{Case 1:} $j<i$ (see
\autoref{fig:visualisation-of-case-1-}). By induction hypothesis,
\autoref{e:axioms} holds for the trace up to $n'$.  But,
$[F_1,\ldots,F_{i}]\not\vDash \AssLock$, since we assumed that for all
$k$ such that $j<k\le n'$, $\mathrm{Unlock}(l',t)\not\in_E F_k$.

\begin{figure}[h] % (fold)
\centering
\begin{tikzpicture}[
scale=3
]

%draw horizontal line
\draw[snake] (0,0) -- (1,0);
\draw[snake] (1,0) -- (2,0);
\draw[snake] (2,0) -- (3,0);
\draw[snake] (3,0) -- (4,0);

%draw vertical lines
\foreach \x in {1,2,3}
   \draw (\x cm,2pt) -- (\x cm,-2pt);

%draw nodes
\draw (1,0) node[below=3pt] {$ i $} node[above=3pt] {$\mathrm{Lock}(l,t)$};
\draw (2,0) node[below=3pt] {$ j $} node[above=3pt] {$\mathrm{Lock}(l',t)$};
\draw (3,0) node[below=3pt] {$ n'+1 $} node[above=3pt] {$\mathrm{Unlock}(l,t)$};
\end{tikzpicture}
\caption{ Visualisation of Case 2. }
\label{fig:visualisation-of-case-2-}
% figure visualisation-of-case-2- (end)
\end{figure}

\noindent \underline{Case 2:} $i<j$ (see
\autoref{fig:visualisation-of-case-2-}).  As shown in the \code{lock}
case, any $k$ such that $\mathrm{Unlock}(l,t)\in_E F_k$ is
$k=n'+1$. This contradicts \autoref{e:axioms} for the trace up to $j$,
since $[F_1,\ldots,F_{j}]\not\vDash \AssLock$, because there is not
$k$ such that $i<k<j$ such that $\mathrm{Unlock}(l,t)\in_E F_k$.  This
concludes the proof that \autoref{e:lock} holds for $n+1$.

\autoref{e:axioms} holds, since none of the axioms, in particular not
\AssLock, become unsatisfied if they were satisfied for the trace up
to $f(n-1)$ and an $\mathrm{Unlock}$ is added.

Since
$\Processes_n=\Processes_{n-1} \setminus^\#
\set{\text{unlock $t$; $Q$}}
\cup^\# \set{Q} $
and
$\{ Q \} \leftrightarrow \{ \state_{p\cdot 1}(\tilde t) \}$
 (by definition of the
translation), we have that \autoref{e:proc} holds.
\autoref{e:nonce},
\autoref{e:storeb}
and
\autoref{e:subst}
hold trivially.
\end{mycase} %}}}

\begin{mycase}[%{{{ embedded msrs
$(\Names_{n-1}, \StoreA_{n-1},\StoreB_{n-1},
\Processes_{n-1}=\Processes' \cup \set{
l \msrewrite a r;~Q
}, \Subst_{n-1},
\ActiveLocks_{n-1})
\trans a \newline
(\Names_{n-1}, \StoreA_{n-1}, \StoreB_{n-1}\setminus \lfacts(l') \cup^\# \mathit{mset}(r),
 \Processes'\mcup \set{Q},
\Subst_{n-1}, \ActiveLocks_{n-1})$]
This step requires that $l' \msrewrite{a'} r'
\in_E \ginsts(l \msrewrite a r)$ and
$\lfacts(l') \subset^\# \StoreB_{n-1}, \pfacts(l')\subset
\mathit{mset}(\StoreB_{n-1})$. Let $\theta$ be a substitution such
that $(l \msrewrite{a} r)\theta = (l' \msrewrite{a'} r')$.
Since, by induction hypothesis,
$\StoreB_{n-1}=S_{n'}\msetminus \ReservedFacts$,
we therefore have $\lfacts(l') \subset^\# S_{n'}, \pfacts(l')\subset
\mathit{mset}(S_{n'})$.
\inclusionAcaseintro
{\code{$l \msrewrite a r$; Q}}
{ [\state_p(\tilde t), l'] \msrewrite {a',\Event()} [r',\state_{p\cdot 1}(\tilde t\cup
(\vars(l)\theta))]
}
{ S_{f(n-1)} \setminus^\#  \mset{ \state_p(\tilde t)}\msetminus \lfacts(l')
\mcup \mset{\state_{p\cdot 1}(\tilde t \cup
(\vars(l)\theta))} \mcup \mathit{mset}(r') }
{a',\Event()}

\autoref{e:storeb} holds since
\begin{align*}
	\StoreB_n & = \StoreB_{n-1}\msetminus \lfacts(l')
\cup^\# \mathit{mset}(r)  \\
	& =  ( S_{n'} \msetminus  \ReservedFacts )\setminus \lfacts(l')
\cup^\# \mathit{mset}(r)  \tag{induction hypothesis} \\
	& =  ( S_{n'} \msetminus \lfacts(l')
\cup^\# \mathit{mset}(r)  \msetminus \mset{ \state_p(\tilde t)} \mcup
\mset{\state_{p\cdot 1}(\tilde t \cup
(\vars(l)\theta))} )
 \msetminus  \ReservedFacts  \tag{since $\state_p(\tilde t), \state_{p\cdot 1}(\tilde t \cup
(\vars(l)\theta))\in
\ReservedFacts$ }\\
	& = S_{f(n)} \msetminus \ReservedFacts
\end{align*}

The step from $S_{f(n-1)}$ to $S_{f(n)}$ is labelled
$F_{f(n)}=a$, and does not contain actions in
$\ReservedFacts$, since $P$ is well-formed. Hence
\autoref{e:storea},
 \autoref{e:lock},
\autoref{e:axioms}
and
\autoref{e:actions}
hold.

Since
$\Processes_n=\Processes_{n-1} \setminus^\#
\set{\text{$l \msrewrite{a} r$; $Q$}}
\cup^\# \set{Q} $
and
$\{ Q \} \leftrightarrow \{ \state_{p\cdot 1}(\tilde t \cup
(\vars(l)\theta)) \}$
 (by definition of $\sem{P}_{=p}$), we have that \autoref{e:proc} holds.
\autoref{e:nonce},
and
\autoref{e:subst}
hold trivially.
\end{mycase}%}}}

\end{component}

\end{proof}

\begin{definition}[normal msr execution] % (fold)%{{{
\label{def:normal-msr-execution}
A msr execution $\emptyset \trans{E_1}_{\sem{P}} \cdots
\trans{E_n}_{\sem{P}} S_n \in \execmsr(\sem{P})$ for the multiset rewrite
system $\sem{P}$ defined by a ground process $P$ is \emph{normal} if:
\begin{enumerate}
	\item \label{item:init}The first transition is an instance of the \textsc{Init} rule,
\ie, $S_1=\state_{[]}()$ and there is at least this transition.
	\item \label{item:no-semi-no-ack} $S_n$ neither contains any fact with the symbol $\semistate_p$
for any $p$, nor any fact with symbol $\Ack$.
	\item \label{item:silent-ded}if for some $i$ and $t_1,t_2\in\Mess$, $\Ack(t_1,t_2) \in (
S_{i-1} \msetminus S_i)$, then there are $p$ and $q$ such that:
 \[ S_{i-3} \trans{~}_{R_1}  S_{i-2}
\trans{~}_{R_2}  S_{i-1}
\trans{~}_{R_3}  S_i \quad\text{, where:}\]
		\begin{itemize}
			\item  $R_1 = [\state_p(\tilde x)] \to [\Msg(t_1,t_2),
\semistate_{p}(\tilde x)]$
		\item $R_2=[\state_q(\tilde y), \Msg(t_1,t_2)] \to
  [\state_{q\cdot 1}(\tilde y \cup \tilde y'), \Ack(t_1,t_2)]$
		\item $R_3=[\semistate_p(\tilde x),\Ack(t_1,t_2) ] \to [\state_{p\cdot
1}(\tilde x)]$.
		\end{itemize}
	\item \label{item:real-step} $S_{n-1}\trans{E_n}_{\sem{P,[],[]},\textsc{MDIn,Init}}S_n$
	\item\label{item:mdin-immeadately} if $\In(t) \in ( S_{i-1} \msetminus S_i)$ for some $i$ and $t\in\Mess$, then
$S_{i-2} \trans{K(t)}_\textsc{MDIn} S_{i-1}$
	\item \label{item:m-a} if $n\ge 2$ and no \Ack-fact in $(
S_{i-1} \msetminus S_i)$, then
 there exists $m<n$ such that $S_m \rightarrow^*_R
S_{n-1}$ for $R=\MDOutPubFreshAppl \cup  \textsc{Fresh}$ and
	$\emptyset \trans{E_1}_{\sem{P}} \cdots
\trans{E_m}_{\sem{P}} S_m \in \execmsr(\sem{P})$ is normal.
	\item \label{item:m-b} if for some $t_1,t_2\in\Mess$, $\Ack(t_1,t_2) \in (
S_{n-1} \msetminus S_n)$, then
 there exists $m\le n-3$ such that $S_m \rightarrow^*_R
S_{n-3}$ for $R=\MDOutPubFreshAppl \cup  \textsc{Fresh}$ and
	$\emptyset \trans{E_1}_{\sem{P}} \cdots
\trans{E_m}_{\sem{P}} S_m \in \execmsr(\sem{P})$ is normal.
\end{enumerate}
% definition normal-msr-execution (end)
\end{definition}%}}}

\begin{lemma}[Normalisation]
\label{lem:msr-normalisation}
  Le $P$ be a well-formed ground process. If
  \[
	S_0=\emptyset \stackrel{E_1}{\longrightarrow}_\sem P S_1
	\stackrel{E_2}{\longrightarrow}_\sem P \ldots
	\stackrel{E_{n}}{\longrightarrow}_\sem P S_{n} \in \execmsr(\sem
  P)
  \]
 and $[E_1,\ldots,E_n]\vDash \alpha$,
 then there exists a normal msr execution
  \[
	T_0=\emptyset \stackrel{F_1}{\longrightarrow}_\sem P T_1
	\stackrel{F_2}{\longrightarrow}_\sem P \ldots
	\stackrel{F_{n'}}{\longrightarrow}_\sem P T_{n'} \in \execmsr(\sem
  P)
  \]
	such that $\hide([E_1,\ldots,E_n])=\hide(F_1,\ldots,F_{n'})$ and
$[F_1,\ldots,F_{n'}]\vDash \alpha$.
\end{lemma}
\begin{proof}
%for convenience
\newcommand{\tracename}[1]{\ensuremath{
S_0^{(#1)}  \stackrel{E_1^{(#1)}}{\longrightarrow}_\sem P
	\ldots \stackrel{E_{n}^{(#1)}}{\longrightarrow}_\sem P S_{n^{(#1)}}^{(#1)} }}

\newcommand{\conditiona}[1]{\ensuremath{
	\hide([E_1,\ldots,E_n])=\allowbreak\hide([E_1^{(#1)},\allowbreak\ldots,E_{n^{(#1)}}^{(#1)}])
 }}
\newcommand{\conditionb}[1]{\ensuremath{
	[E_1^{(#1)},\ldots,E_{n^{(#1)}}^{(#1)}] \vDash \alpha
 }}

We will modify $S_0 \stackrel{E_1}{\longrightarrow}_\sem P
	\ldots \stackrel{E_{n}}{\longrightarrow}_\sem P S_{n}$ by applying one
transformation after the other, each resulting in an msr execution that
still fulfills the conditions on its trace.
\begin{enumerate}
	\item If an application of the \textsc{Init} rule appears in
$S_0 \stackrel{E_1}{\longrightarrow}_\sem P
	\ldots \stackrel{E_{n}}{\longrightarrow}_\sem P S_{n}$, we move it to
the front. Therefore,  $S_1=\state_{[]}()$.
This is possible since the left-hand side of the \textsc{Init} rule is
empty. If the rule is never instantiated, we prepend it to the trace. Since
$\mathrm{Init}()\in\ReservedFacts$, the resulting msr execution
\[ \tracename{1} \]
is such that \conditiona{1}. Since $\mathrm{Init}()$ is only added if it
was not present before, \conditionb{1}, especially \AssInit.

	\item For each fact $\Ack(t_1,t_2)$ contained in $S^{(1)}_{n^{(1}}$, it
also contains a fact $\semistate_p(\tilde t)$ for some $p$ and $\tilde t$
such that there exists a rule of type $R_3$ that consumes both of them,
since $\Ack(t_1,t_2)$ can only be produced by a rule of type $R_2$ which
consumes $\Msg(t_1,t_2)$ which in turn can only be produced along with a
fact $\semistate_p(\tilde t)$, and by definition of $\sem{P}$, there
exists a rule in $\sem{P}_{=p}$ of form $R_3$ that consumes
$\Ack(t_1,t_2)$ and $\semistate_p(\tilde t)$.
We append as many applications of rules of type
$R_3$ as there are facts $\Ack(t_1,t_2)\in S^{(1)}_{n^{(1)}}$, and repeat
this for all $t_1,t_2$ such that $\Ack(t_1,t_2)\in S^{(1)}_{n^{(1)}}$.
Then,  $S^{(1)}_{n^{(1)}} \trans{~}_{\sem{P}} S^{(1)}_{n'}$ and
$S^{(1)}_{n'}$ does not contain \Ack-facts anymore.

 If $S^{(1)}_{n'}$ contains a fact $\semistate_p(\tilde t)$, we remove
the last transition that produced this fact, \ie, for $i$ such that $S_i =
S_{i-1} \msetminus \mset{\state_p(\tilde t)} \mcup
\mset{\Msg(t_1,t_2),\semistate_p(\tilde t)}$, we define
\[
	S^{(1)'}_j := \begin{cases}
	S^{(1)}_j & \text{if $j\le i-1$} \\
	S^{(1)}_{j+1} \msetminus
\mset{\Msg(t_1,t_2),\semistate_p(\tilde t)} \mcup \mset{\state_p(\tilde t)}
	& \text{if $i-1< j < n'$}
	\end{cases}
\]
The resulting execution is valid, since $\semistate_p(\tilde t)\in
S^{(1)}_{n'}$ and since $\Msg(t_1,t_2)\in S^{(1)}_{n'}$. The latter is the
case because if $\Msg(t_1,t_2)$ would be consumned at a later point, say
$j$, $j+1$ would contain $\Ack(t_1,t_2)$, but since $S^{(1)'}_{n'-1}$ does
not contain \Ack-facts, they can only be consumned by a rule of type $R_3$,
which would have consumned $\semistate_p(\tilde t)$. We repeat this
procedure for every remaining $\semistate_p(\tilde t)\in S^{(1)}_{n'}$, and
call the resulting trace
\[ \tracename{2} \]

Since no rule added or removed or removed has an action, \conditiona{2} and
\conditionb{2}.

	\item We transform \tracename{1} as follows (all equalities are modulo
\ET):
Let us call instances of $R_1$, $R_2$ or $R_3$ that appear outside
a chain
\[ S_{i-3} \trans{~}_{R_1}  S_{i-2} \trans{~}_{R_2}  S_{i-1} \trans{~}_{R_3}  S_i \]
 for some $i$ $t_1,t_2\in\Mess$ ``unmarked''.
Do the following for the smallest $i$ that is an unmarked instance of
$R_3$ ( we will call the instance
of $R_3$ $\ri_3$ and suppose it is applied from $S_{i-1}$ to $S_i$):
Apply $\ri_3$ after $j<i$ such that $S_{j-1}$ to $S_j$ is the first
unmarked instance of $R_2$, for some $q$ and $\tilde y$,\ie, this instance
produces a fact $\state_{q\cdot 1}(\tilde y,\tilde y')$ and
a fact $\Ack(t_1,t_2)$. Since there is no rule between $j$ and $i$ that
might consume $\Ack(t_1,t_2)$ (only rules of form $R_3$ do, and $\ri_3$ is
the first unmarked instance of such a rule) and since $\ri_3$ does not consume
 $\state_{q\cdot 1}(\tilde y,\tilde y')$, we can move $ri_3$ between $j$
and $j+1$, adding the conclusions of $\ri_3$ and removing the premises of
$\ri_3$ from every $S_{j+1},\ldots,S_{i}$. Note that unmarked instances of
$R_2$ and $R_3$ are guaranteed to be preceeded by a marked $R_1$, and
therefore  only remove facts of form $\Ack(\ldots)$ or $\Msg(\ldots)$ that
have been added in that preceeding step.
Since the transition at step $j$ requires a fact $\Msg(t_1,t_2)$, there is
an instance of $R_1$ prior to $j$, say at $k<j$, since only rules of form
$R_1$ produces facts labelled $\Msg(t_1,t_2)$. Since $\ri_3$ is now
applied from $Sj$ to $S_{j+1}$, we have that an instance $\ri_1$ of a rule of form
$R_1$ that produces $\semistate_p(\tilde t)$ must appear before $j$, \ie,
$\ri_1 \in\ginsts(\sem{P}_{=p})$. Therefore, it produces a fact
$\Msg(t_1,t_2)$ indeed. We choose the largest $k$ that has an unmarked
$R_1$ that produces $\Msg(t_1,t_2)$ and $\semistate_p(\tilde t)$ and move
it right before $j$, resulting in the following msr execution:
\[
	S^{(1)'}_t := \begin{cases}
	S^{(1)}_t & \text{if $t<k$} \\
	S^{(1)}_{t+1}
	 \mcup \mset{\Msg(t_1,t_2),\semistate_p(\tilde t)} \msetminus \mset{\state_p(\tilde t)}
	& \text{if $k\le t < j-1 $} \\
	S^{(1)}_{(t)}
	& \text{if $j-1 \le t < j+1 $} \\
	S^{(1)}_{(t-1)} \msetminus \mset{\semistate_p(\tilde t),\Ack(t_1,t_2)} \mcup \mset{\state_{p\cdot 1}(\tilde t)}
	& \text{if $j+1\le t < i+1 $} \\
	S^{(1)}_t & \text{if $i+1 \le t$}
	\end{cases}
\]

We apply this procedure until it reaches a fixpoint and call the resulting trace
\[ \tracename{3} \]

Since no rule moved during the procedure has  an action, \conditiona{3} and
\conditionb{3}.

	\item If the last transition is in \MDOutPubFreshApplFresh, we remove
it. Repeat until fixpoint is reached and call the resulting trace
\[ \tracename{4} \]
Since no rule removed during the procedure has an action, \conditiona{4} and
\conditionb{4}.

	\item If there is  $\In(t)\in S^{(4)}_{n^{(4)}-1}$, then there is a
transition where $\In(t)$ is produced and never consumned until
$n^{(4)}-1$. The only rule producing $\In(t)$ is \textsc{MDIn}.
We can move this transition to just before $n^{(4)}-1$ and call the resulting trace
\[ \tracename{5} \]
Since $\conditionb{4}$, especially $\AssIn$, there is no action that is not
in $\ReservedFacts$ between the abovementioned instance of \textsc{MDIn},
therefore, \conditiona{5} holds. Since $\AssIn$ is the only part of
$\alpha$ that mentions $\mathrm{K}$, and since the tranformation preserved
$\AssIn$, we have that \conditionb{5}.

	\item We will show that \ref{item:m-a} and \ref{item:m-b} hold for
	\[ \tracename{5} \] in one step.

If $n^{(5)}\ge 2$ and there is no \Ack-fact in
$S^{(5)}_{n^{(5)-1}}
\setminus S^{(5)}_{n^{(5)}}$,
then we chose the largest $m<n$ such that
$S^{(5)}_{m-1}\trans{E^{(5)}_m}_{\sem{P,[],[]},\textsc{Init,MDIn}}
S^{(5)}_m$, or, if there is an \Ack-fact in
$S^{(5)}_{n^{(5)-1}}
\setminus S^{(5)}_{n^{(5)}}$,
 we will chose the largest $m'<n-2$ such that
$S^{(5)}_{m'-1}\trans{E^{(5)}_{m'}}_{\sem{P,[],[]},\textsc{Init,MDIn}}
S^{(5)}_{m'}
$.

This trivially fulfills \ref{item:real-step}.
$S^{(5)}_m \rightarrow^*_R
S^{(5)}_{n^{(5)}}$
and
$S^{(5)}_{m'} \rightarrow^*_R
S^{(5)}_{n^{(5)-3}}$,
 since otherwise there would be a larger $m$ or $m'$. This also
implies \ref{item:no-semi-no-ack}, as none of the rules in
$R=\MDOutPubFreshApplFresh$ remove \Ack- or
$\semistate$-facts, and the chain of rules $R_1,R_2,R_3$ consumes as many as
it produces.  Thus, if they where in $S^{(5)}_m$, they would be in
$S^{(5)}_{n^{(5)}}$, too.
 Since $n>2$, $m>1$, and therefore \ref{item:init}. \ref{item:silent-ded}  holds
for all parts of the trace, and therefore also for the $m$ prefix. Similar for
\ref{item:mdin-immeadately}.

Since we can literally apply the same argument
for the largest $\tilde m<m$ such that
$$S^{(5)}_{m-1}\trans{E^{(5)}_m}_{\sem{P,[],[]},\textsc{Init,MDIn}}
S^{(5)}_m$$ or, in case that there is an \Ack-fact in
$S^{(5)}_{m-1}
\setminus S^{(5)}_{m}$,
for the largest  $\tilde m<m-2$, can show that \ref{item:m-a} and
\ref{item:m-b} hold for the trace up to $m$ or $m'$, concluding it is
normal.

\end{enumerate}

\end{proof}

\begin{definition}
\label{def:second-process-bijection}
Let $P$ be a ground process, \calP be a multiset of processes and $S$
a multiset of multiset rewrite rules. We write ${\cal P}
\bijp S$ if there exists a bijection between \calP and
the multiset $\{ \state_p(\tilde t) \mid \exists p, \tilde
t.\ \state_p(\tilde t) \melem S \}^\#$ such that whenever $Q \melem
\calP$ is mapped to $\state_p(\tilde t) \melem S$, then:
\begin{enumerate}
\item $\state_p(\tilde t)\in_E\prems(R)$ for $R\in\ginsts(\sem{P}_{=p})$.
\item Let
$\theta$ be a grounding substitution for
$state(\tilde x)\in\prems(\sem{P}_{=p})$ such that
$\tilde t=\tilde{x}\theta$. Then
 \[ (P|_p\tau)\rho =_E Q \]
for a substitution $\tau$,
and a bijective renaming $\rho$ of fresh, but not bound names in $Q$,
defined as follows:
\begin{align*}
	\tau(x):=& \theta(x) && \text{if $x$ not a reserved variable}\\
	\rho(a):=& a' && \text{if $\theta(n_a)=a'$}
\end{align*}

\end{enumerate}
\end{definition}
When $\calP \bijp S$, $Q \melem \calP$ and
$\state_p(\tilde t) \melem S$ we also write $Q \bijp
\state_p(\tilde t)$ if this bijection maps $Q$ to $\state_p(\tilde
t)$.

\begin{remark}
  \label{remark:second-properties-bijection}
  Note that $\bijp$ has the following properties (by the
  fact that it defines a bijection between multisets).
  \begin{itemize}
  \item  If $\calP_1 \bijp S_1$ and $\calP_2 \bijp S_2$ then
    $\calP_1 \mcup \calP_2 \bijp S_1 \mcup S_2$.
  \item If $\calP_1 \bijp S_1$ and $Q \bijp
    \state_p(\tilde t)$ for $Q \in \calP_1$ and $\state_p(\tilde t) \in
    S_1$ (i.e. $Q$ and $\state_p(\tilde t)$ are related by the bijection
    defined by $\calP_1 \bijp S_1$) then $\calP_1 \msetminus
    \{Q\} \bijp S_1 \msetminus \{ \state_p(\tilde t) \}$.
  \end{itemize}
\end{remark}

\begin{lemma}
\label{lem:inclusion-msr-in-apip}
  Le $P$ be a well-formed ground process. If
  \[
	S_0=\emptyset \stackrel{E_1}{\longrightarrow}_\sem P S_1
	\stackrel{E_2}{\longrightarrow}_\sem P \ldots
	\stackrel{E_{n}}{\longrightarrow}_\sem P S_{n} \in \execmsr(\sem
  P)
  \]
	is normal (see \autoref{def:normal-msr-execution}) and
  $[E_1,\ldots,E_{n}]\vDash \alpha$ (see
Definition~\ref{def:transformula}),
 then there are
	$\apipstate{0},\ldots, \apipstate{n'}$ and $F_1,\ldots,F_{n'}$ such
that:
  \[
  \apipstate{0}
  \stackrel{F_1}{\longrightarrow}
  \apipstate{1}
  \stackrel{F_2}{\longrightarrow}
  \ldots \stackrel{F_{n'}}{\longrightarrow}
  \apipstate{n'}
  \]
  where  $\apipstate{0} = (\emptyset,\emptyset,\emptyset,\set{P},\emptyset,\emptyset)$
  and there exists a monotonically increasing, surjective
function $f\colon \setN_n
\setminus \set{0}
  \to \setN_{n'}$ such that $f(n)=n'$ and for all $i \in \setN_n$
  \begin{enumerate}
  \item\label{f:nonce}
	$ \Names_{f(i)} =  \set{ a\in\FN \mid \ProtoNonce(a)\in_E\bigcup_{1 \leq j \leq
i}E_j}$

  \item\label{f:storea}
    $\forall~t\in\Mess.\:\StoreA_{f(i)} (t) =
    \begin{cases}
      u &
      \mbox{if } \exists j \leq i. \mathrm{Insert}(t,u) \in_\ET
      E_j \\
      & \qquad \wedge \forall j', u'. j < j' \leq i \to
      \mathrm{Insert}(t,u') \not \in_\ET E_{j'} \wedge
      \mathrm{Delete}(t) \not \in_\ET E_{j'}\\
      \bot & \text{otherwise}
    \end{cases}$

  \item\label{f:storeb}
    $\StoreB_{f(i)} =_\ET S_{i} \msetminus \ReservedFacts$

  \item\label{f:proc}
    $\Processes_{f(i)} \bijp S_{i}$

  \item\label{f:subst}
$\mset{x \Subst_{f(i)} \mid x\in\Dom(\Subst_{f(i)})}
=_\ET \mset{\mathsf{Out}(t)\in \cup_{k \leq i} S_k}$

  \item\label{f:lock}
    $\calL_{f(i)} =_\ET \set{ t \mid \exists j\leq i,u .\
    \mathrm{Lock}(u,t) \in_\ET E_j
    \wedge \forall j < k \leq i. \mathrm{Unlock}(u,t) \not \in_\ET E_k}.$
\end{enumerate}
Furthermore,
\begin{enumerate}[resume]
  \item\label{f:actions}
		% if $E_i \not\subset \ReservedFacts$, then $E_i=F_{f(i)}$ and
% $E_{i'}\subset \ReservedFacts$ for all $i'\neq i, f(i')=f(i)$.
		$\mathit{hide}([E_1,\ldots,E_n])=_\ET[F_1,\ldots,F_n'].$
  \end{enumerate}

\end{lemma}
\begin{proof}
  \renewcommand\itemautorefname{Condition}

We proceed by induction over the number of transitions $n$.
\begin{component}[Base Case]
  A normal msr execution contains at least an application of the init rule,
thereby the shortest normal msr execution is
	\[ \emptyset \trans{~}_{\sem{P}} S_1=\mset{
\state_{[]}()} \]
	We chose $n'=0$ and thus
\[ \apipstate{0}=
(\emptyset,\emptyset,\emptyset,\mset{P},\emptyset,\emptyset). \]
	We define $f: \set{1} \to \set{0}$ such that $f(1)=0$.

 To show that \autoref{f:proc} holds, we have to show that
 $\Processes_0 \bijp
 \mset{\state_{[]}(\mathit{s}:\freshsort)}$.  Note that
 $\Processes_0=\mset{P}$. We choose the bijection such that $P
 \bijp \state_{[]}(\mathit{s}:\freshsort)$.

 By \autoref{def:trans-at-pos},
 $\sem{P}_{=[]}=\sem{P,[],[]}_{=[]}$.  We see
 from Figure~\ref{fig:transatp} that for every $P$ we have that
 $\state_{[]}(\mathit{s}:\freshsort) \in
 \prems(R\theta)$, for $R\in \sem{P,[],[]}_{=[]}$ and
 $\theta=\emptyset$. This induces
 $\tau=\emptyset$ and $\rho=\emptyset$. Since $P|_{[]}\tau\rho=P$, we have
 $P \bijp \state_{[]}()$, and therefore $\Processes_0
\bijp S_1$.

\autoref{f:nonce}, \autoref{e:storea}, \autoref{e:storeb},
\autoref{f:subst}, \autoref{e:lock},
and
\autoref{f:actions}
 hold trivially.
\end{component}
\begin{component}[Inductive step]

Assume the invariant holds for $n-1\geq 1$. We have to show that the lemma
holds for $n$ transitions, \ie, we assume that
\[
  \emptyset \stackrel{E_1}{\longrightarrow}_\sem P S_1
  \stackrel{E_2}{\longrightarrow}_\sem P \ldots
  \stackrel{E_{n}}{\longrightarrow}_\sem P S_{n} \in \execmsr(\sem
  P)
\]
is normal and $[E_1,\ldots,E_{n}]\vDash \alpha$.  Then it is to show that there is
\[
  \apipstate{0}
  \stackrel{F_1}{\longrightarrow}
  \apipstate{1}
  \stackrel{F_2}{\longrightarrow}
  \ldots \stackrel{F_n'}{\longrightarrow}
  \apipstate{n'}
  \]
fulfilling Conditions~\ref{f:nonce} to~\ref{e:actions}.

Assume now for the following argument, that there is not fact with the
symbol \Ack in $S_{n-1} \msetminus S_n$. This is the case for all cases
except for the case where rule instance applied from $S_{n-1}$ to $S_n$ has
the form $\ri= \msr {\semistate_p(\tilde s),\Ack(t_1,t_2)} {} {
\state_{p\cdot 1}(\tilde s) }$. This case will require a similiar, but
different argument, which we will present when we come to this case.

Since
$\emptyset \trans{E_1}_{\sem{P}} \cdots
\trans{E_n}_{\sem{P}} S_n \in \execmsr(\sem{P})$
is normal and $n\ge 2$, there exists $m<n$ such that $S_m \rightarrow^*_R
S_{n}$ for $R=\MDOutPubFreshApplFresh$ and
	$\emptyset \trans{E_1}_{\sem{P}} \cdots
\trans{E_m}_{\sem{P}} S_m \in \execmsr(\sem{P})$ is normal, too.
This allows us to apply the induction hypothesis on
$\emptyset \trans{E_1}_{\sem{P}} \cdots
\trans{E_m}_{\sem{P}} S_m \in \execmsr(\sem{P})$.
Hence there is a monotonically
increasing function from $\setN_{m} \to \setN_{n'}$  and an execution
such that Conditions~\ref{f:nonce} to~\ref{e:actions} hold. Let $f_p$ be
this function and note that $n'=f_p(m)$.

In the following case distinction, we will (unless stated otherwise) extend
the previous execution by one step from $\apipstate{n'}$ to
$\apipstate{n'+1}$, and prove that Conditions~\ref{f:nonce}
to~\ref{f:actions} hold for $n'+1$.  By induction hypothesis, they hold for
all $i\le n'$. We define a function $f \colon \setN_n \to \setN_{n'+1}$ as
follows:
\begin{equation*}
\label{eq:definition-of-f}
 f(i) :=
\begin{cases}
	 f_p(i) & \text{if $i\in\setN_{m}$} \\
	 n' & \text{if $m<i<n$} \\
	 n'+1 & \text{if $i=n$}
\end{cases}
\end{equation*}

Since, $S_m \rightarrow^*_R S_{n}$ for $R=\MDOutPubFreshApplFresh$, only
$S_n \msetminus S_m$ contains only \Fr-facts and \K-facts, and
$S_m \msetminus S_n$ contains only \Fr-facts and \Out-facts. Therefore,
\autoref{f:storeb},\ref{f:proc} and \ref{f:subst} hold for all
$i\le n-1$. Since $E_{m+1},\ldots,E_{n-1}=\emptyset$, \autoref{f:nonce},
\ref{f:storea},\ref{f:lock} and \ref{f:actions} hold for all $i\le n-1$.

Fix a bijection such that
$\Processes_{n'} \bijp S_{m}$. We will abuse notation
by writing $P \bijp \state_p(\tilde t)$, if this bijection maps
$P$ to $\state_p(\tilde t)$.

We now proceed by case distinction over the last type of transition from
$S_{n-1}$ to $S_n$. Let $l_\mathit{linear}=_\ET S_{n-1} \setminus S_n$ and
$r=_\ET S_{n} \setminus S_{n-1}$. $l_\mathit{linear}$ can only contain linear
facts, while $r$ can contain linear as well as persistent facts. The rule
instance $\ri$ used to go from $S_{n-1}$ to $S_n$ has the following form:
\[ [l_\mathit{linear},l_\mathit{persistent}] \msrewrite{E_n} r \]
for some $l_\mathit{persistent} \msubset_\ET S_{n-1}$.

Note that $l_\mathit{linear}$, $E_n$ and $r$ uniquely identify which rule
in $\sem{P,[],[]}$ $\ri$ is an instance of -- with exactly one
exception: $\sem{ [] \msrewrite{a} [];~P, p, \tilde x} = \sem{\code{event
$a$;~$P$},p,\tilde x}$. Luckily, we can treat the last as a special case of the
first.

If $R$ is uniquely determined, we fix some $\ri\in\ginsts(R)$.

\begin{mycase}[$R=\textsc{Init}$ or $R\in \textsc{MD}\setminus%{{{
\set{\textsc{MDIn}}$]
	In this case, $\emptyset \trans{E_1} \ldots \trans{E_n} S_n$ is not a
well-formed msr execution.
\end{mycase}%}}}

\begin{mycase}[$R=\textsc{MDIn}$]%{{{
	Let $t\in\Mess$ such that $\ri=R\tau=\K(t) \msrewrite{K(t)} \In(t)$.

From the induction hypothesis, and since
$E_{m+1},\ldots,E_n = \emptyset$,
we have that
 \[ \Names_{n'} =  \set{ a\in \FN \mid \ProtoNonce(a)\in_\ET\bigcup_{1 \leq j \leq
n} E_j}. \]

From the induction hypothesis, and since no rule producing $\Out$-facts is
applied between step $m$ and step $n$, we have that
\[ \set{x \Subst_{n'} \mid x\in\Dom(\Subst_{n'})}^\#=_\ET\mset{\mathsf{Out}(t)\in \cup_{k
\leq n} S_k}. \]

Let $\tilde r= \set{ a\in\FN\mid \RepNonce(a)\in_\ET\bigcup_{1\leq j \leq
n} F_j}$. Then, by \autoref{lem:ded} and \autoref{lem:ded-eq}, we have that
$\nu \Names_{n'},\tilde r.\Subst_{n'} \vdash t$. Therefore, $\nu
\Names_{n'}.\Subst_{n'} \vdash t$.
This allows us to chose the following transition:
\[
 \cdots \trans{F_{n'}}
  \apipstate{n'}
  \trans{K(t)}
  \apipstate{n'+1}
\]
with \apipstate{n'+1}=\apipstate{n'}.

Conditions~\ref{f:nonce} to~\ref{e:actions} hold trivially.
\end{mycase}%}}}

\newcommand{\inclusionBcaseintroHelper}[5]{%{{{
% #1 - Process #2 - for some bla bla.  #3 - Why is it possible to execution
% this step #4 - the new \Names,\StoreA etc.  #5 - the label F_{n'+1}
By induction hypothesis, we have ${\cal P}_{n'} \bijp S_{m}$, and thus, as
previously established, ${\cal P}_{n'} \bijp S_{n-1}$.
Let $Q\melem \calP_{n'}$ such that $Q \bijp \state_p(\tilde t)$. Let
$\theta$ be a grounding substitution for $\state_p(\tilde
x)\in\prems(\sem{P}_{=p})$ such that $\tilde t=\tilde{x}\theta$. Then
$\theta$ induces a substitution $\tau$ and a bijective renaming $\rho$ for
fresh, but not bound names (in $Q$) such that $P|_{p}\tau\rho=Q$ (see
\autoref{def:second-process-bijection}).

From the form of the rule $R$, and since $Q=P|_{p}\tau\rho$, we can deduce
that $Q=#1$#2.

% Why is it possible to execute this step
#3 We therefore chose the following transition:
\[
  \cdots \trans{F_n'} \apipstate{n'} \trans{#5} \apipstate{n'+1}
\]
with
#4.

We define $f$  as on page~\pageref{eq:definition-of-f}. Therefore,
Conditions~\ref{f:nonce} to \ref{e:actions} hold for $i<n-1$.
It is left to show that Conditions~\ref{f:nonce} to \ref{e:actions} hold for $n$.
}%}}}

\newcommand{\inclusionBcaseintro}[6]{%{{{
% #1 - Process
% #2 - for some bla bla.
% #3 - Why is it possible to execution this step
% #4 - the new \Names,\StoreA etc.
% #5 - the label F_{n'+1}
% #6 - S_n=S_{n-1} \mcup bla etc..
\inclusionBcaseintroHelper{#1}{#2}{#3}{#4}{#5}

By definition of $\sem{P}$ and $\sem{P}_{=p}$, we have that
$Q' \leftrightarrow \state_{p\cdot 1}(\tilde t)$.
Therefore, and since
$#1 \leftrightarrow \state_{p}(\tilde t)$,
$\calP_{n'+1}=\calP_{n'} \msetminus \mset{#1} \mcup \mset{Q'}$,
and #6,
\autoref{f:proc} holds.
}%}}}

\begin{mycase}[ %{{{ zero
$\ri= \msr{\state_p(\tilde t)}{}{}$ (for some $p$ and
$\tilde t$)]
\inclusionBcaseintroHelper
{\code{0}}
{}
{}
{$\Names_{n'+1}=\Names_{n'}$,
$\StoreA_{n'+1}=\StoreA_{n'}$,
$\StoreB_{n'+1}=\StoreB_{n'}$,
$\calP_{n'+1}=\calP_{n'}\msetminus \mset{\code{0}}$,
$\Subst_{n'+1}=\Subst_{n'}$ and
$\ActiveLocks_{n'+1}=\ActiveLocks_{n'}$}
{K(t)}

\autoref{f:proc} holds since $Q \leftrightarrow \state_p(\tilde t)$,
$\calP_{n'+1}=\calP_{n'} \msetminus \mset{\code{0}}$ and $S_{n}=S_{n-1}
\msetminus \mset{ \state_p(\tilde t)}$.
Conditions~\ref{f:nonce},
\ref{f:storea},
\ref{f:storeb},
\ref{f:subst} and
\ref{f:actions}
hold trivially.
\end{mycase}%}}}

\begin{mycase}[$\ri= \msr{\state_p(\tilde t)}{}{\state_{p\cdot%{{{
1}(\tilde t),\state_{p\cdot 2}(\tilde t)}$ (for some $p$ and
$\tilde t$)]
\inclusionBcaseintroHelper
{\code{$Q_1$|$Q_2$}}
{, for some processes $Q_1=P|_{p\cdot 1}\tau \rho$ and $Q_2=P|_{p\cdot
2}\tau \rho$}
{}
{$\Names_{n'+1}=\Names_{n'}$,
$\StoreA_{n'+1}=\StoreA_{n'}$,
$\StoreB_{n'+1}=\StoreB_{n'}$,
$\calP_{n'+1}=\calP_{n'}\msetminus \mset{Q_1\mid Q_2} \mcup
\mset{Q_1,Q_2}$,
$\Subst_{n'+1}=\Subst_{n'}$ and
$\ActiveLocks_{n'+1}=\ActiveLocks_{n'}$}
{}

By definition of $\sem{P}$ and $\sem{P}_{=p}$, we have that
$Q_1 \leftrightarrow \state_{p\cdot 1}(\tilde t)$ and
$Q_2 \leftrightarrow \state_{p\cdot 2}(\tilde t)$.
Therefore, and since $Q \leftrightarrow \state_p(\tilde t)$,
$\calP_{n'+1}=\calP_{n'}\msetminus \mset{Q_1\mid Q_2} \mcup
\mset{Q_1,Q_2}$,
and $S_{n}=S_{n-1} \msetminus \mset{ \state_p{\tilde t} }
\mcup \mset{\state_{p\cdot 1}(\tilde t),\state_{p\cdot 2}(\tilde t)}$,
\autoref{f:proc} holds.

Conditions~\ref{f:nonce},
\ref{f:storea},
\ref{f:storeb},
\ref{f:subst} and
\ref{f:actions}
hold trivially.
\end{mycase}%}}}

\begin{mycase}%{{{ replication
[$\ri=
\msr
{\mathsf{!state}_p(\tilde t)}
{}
{\state_{p\cdot
1}(\tilde t)}$
 (for some $p,\tilde t$)]
\inclusionBcaseintroHelper
{\code{$!Q'$}}
{for a process $Q'=P|_{p\cdot 1}\tau \rho$.}
{}
{$\Names_{n'+1}=\Names_{n'}$,
$\StoreA_{n'+1}=\StoreA_{n'}$,
$\StoreB_{n'+1}=\StoreB_{n'}$,
$\calP_{n'+1}=\calP_{n'}\mcup \mset{Q'}$,
$\Subst_{n'+1}=\Subst_{n'}$ and
$\ActiveLocks_{n'+1}=\ActiveLocks_{n'}$}
{}

By definition of $\sem{P}$ and $\sem{P}_{=p}$, we have that
$Q' \bijp \state_{p\cdot 1}(\tilde t)$.
Therefore, and since
$\calP_{n'+1}=\calP_{n'} \mcup \mset{Q'}$,
while $S_{n}=S_{n-1} \mcup \mset{\state_{p\cdot 1}(\tilde t)}$,
\autoref{f:proc} holds.

Conditions~\ref{f:nonce},
\ref{f:storea},
\ref{f:storeb},
\ref{f:subst} and
\ref{f:actions}
hold trivially.
\end{mycase}%}}}

\begin{mycase}%{{{ new
[$\ri=
\msr
{\state_p(\tilde t), \Fr(a'\colon \freshsort)}
{\ProtoNonce(a'\colon \freshsort)}
{\state_{p\cdot
1}(\tilde t,a' \colon \freshsort)}$
 (for some $p,\tilde t$ and $a'\in FN$)]
\inclusionBcaseintroHelper
{\code{$\nu\,a$; $Q'$}}
{ for a name $a\in FN$ and a process $Q'=P|_{p\cdot 1}\tau \rho$}
{By definition of $\execmsr$, the fact $\Fr(a')$ can only be produced
once. Since this fact is linear it can only be consumed once. Every rule in
$\sem{P}$ that produces a label $\ProtoNonce(x)$ for some $x$ consumes a
fact $\Fr(x)$. Therefore,
	\[ a' \notin  \set{ a\in\FN \mid \ProtoNonce(a)\in_\ET\bigcup_{1 \leq j \leq
n-1}E_j}.\]
The induction hypothesis allows us to conclude that $a'\notin
\Names_{n'}$,\ie, $a'$ is fresh.
}
{$\Names_{n'+1}=\Names_{n'}\cup {a'}$,
$\StoreA_{n'+1}=\StoreA_{n'}$,
$\StoreB_{n'+1}=\StoreB_{n'}$,
$\calP_{n'+1}=\calP_{n'}\msetminus \mset {\nu~a; Q'} \mcup
\mset{Q'\set{^{a}/_{a'}} }$,
$\Subst_{n'+1}=\Subst_{n'}$ and
$\ActiveLocks_{n'+1}=\ActiveLocks_{n'}$}
{~}

By definition of $\sem{P}$,  $state_{p\cdot 1}(\tilde
x,a)\in \prems(R')$ for an $R'\in \sem{P}_{=p\cdot 1}$.
We can choose
$\theta':=\theta[n_a \mapsto a']$ and have $\state_{p\cdot1}(\tilde
t,a')=\state_{p\cdot 1}(\tilde x,a)\theta'$.
Since $Q=P|_{p}\tau\rho$ for $\tau$ and $\rho$ induced by $\theta$,
$Q'\set{^{a'}/_{a}}=P|_{p} \tau' \rho'$ for $\tau'$ and $\rho'$ induced by $\theta'$, \ie,
$\tau'=\tau$ and $\rho'=\rho[a \mapsto a']$.
Therefore, $Q'\set{^{a'}/_{a}} \bijp \state_{p\cdot 1}(\tilde t,a')$.

\autoref{f:proc} holds,
since furthermore
$\code{$\nu~a'$; Q'} \leftrightarrow \state_{p}(\tilde t)$,
$\calP_{n'+1}=\calP_{n'} \msetminus \mset{\code{$\nu~a'$; Q'}} \mcup
\mset{Q'\set{^{a'}/_a}}$,
and
$S_n=S_{n-1} \msetminus \mset{\Fr(a), \state_p(\tilde t)} \mcup {\state_{p\cdot
1}(\tilde t,a \colon \freshsort)}$.

\autoref{f:nonce}, holds since
$\Names_{n'+1}=\Names_{n'}\cup {a'}$, and $E_n=\ProtoNonce(a')$.
\autoref{f:actions} holds since $\ProtoNonce(a)\in\ReservedFacts$.

Conditions
\ref{f:storea},
\ref{f:storeb} and
\ref{f:subst}
hold trivially.
\end{mycase}%}}}

\begin{mycase}%{{{ out
[$\ri=
\msr
{\state_p(\tilde t), \In(t_1)}
{\ChannelInEvent(t_1)}
{\state_{p\cdot
1}(\tilde t),\Out(t_2)}$
 (for some $p,\tilde t$ and $t_1,t_2\in \Mess$)]
Since the msr execution is normal, we have that
$S_{n-2}
\trans{K(t_1)}_\textsc{MDIn} S_{n-1}$. Since
$S_0 \trans{E_1}_{\sem{P}} \ldots \trans{E_n}_\sem{P} S_n$ is normal, so is
$S_0 \trans{E_1}_{\sem{P}} \ldots \trans{E_{n-1}}_\sem{P} S_{n-1}$, and
therefore
$S_0 \trans{E_1}_{\sem{P}} \ldots \trans{E_{n-2}}_\sem{P} S_{n-2}$. Hence
there is an $m< n-2$ such
$S_0 \trans{E_1}_{\sem{P}} \ldots \trans{E_{m}}_\sem{P} S_{m}$ is a normal
trace and $S_m \rightarrow^*_R
S_{n-1}$ for $R=\MDOutPubFreshApplFresh$.

By induction hypothesis, we have ${\cal P}_{n'} \bijp S_{m}$, and thus,
since \MDOutPubFreshAppl and \textsc{Fresh} do not add or remove
\state-facts, ${\cal P}_{n'} \bijp S_{n-2}$.
Let $Q\melem \calP_{n'}$ such that $Q \bijp \state_p(\tilde t)$.  Let
$\theta$ be a grounding substitution for $state(\tilde
x)\in\prems(\sem{P}_{=p})$ such that $\tilde t=\tilde{x}\theta$.  Then
$\theta$ induces a substitution $\tau$ and a bijective renaming $\rho$ for
fresh, but not bound names (in $Q$) such that $P|_{p}\tau\rho=Q$ (see
\autoref{def:second-process-bijection}).

From the form of the rule $R$, and since $Q=P|_{p}\tau\rho$, we can deduce
that $Q=\code{\piout$(t_1,t_2);Q'$}$ for a process $Q'=P|_{p\cdot 1}\tau \rho$.

% Why is it possible to execute this step
From the induction hypothesis, and since
$E_{m+1},\ldots,E_{n-2} = \emptyset$,
we have that
 \[ \Names_{n'} =  \set{ a\in\FN \mid \ProtoNonce(a)\in_\ET\bigcup_{1 \leq j \le
n-2} E_j}. \]

From the induction hypothesis, and since no rule producing $\Out$-facts is
applied between step $m$ and step $n-2$, we have that
\begin{equation}
\label{eq:subst}
 \mset{x \Subst_{n'} \mid x\in\Dom(\Subst_{n'})}=_\ET\mset{\mathsf{Out}(t)\in \cup_{k
\le n-2} S_k}.
\end{equation}

Let $\tilde r= \set{ a:\freshsort \mid \RepNonce(a)\in\bigcup_{1\leq j
\le n-2} F_j}$.
Since $\K(t_1)\in \prems(\textsc{MDIn}\sigma)$ for $\sigma(x)=t_1$, we have
 $\K(t)\in_\ET S_{n-2}$.
By \autoref{lem:ded} and \autoref{lem:ded-eq}, we have $\nu \Names_{n'},\tilde
r.\Subst_{n'} \vdash t$. Therefore, $\nu \Names_{n'}.\Subst_{n'} \vdash t$.
 We chose the following transition:
\[
  \cdots \trans{F_n'} \apipstate{n'} \trans{K(t_1)}
 \apipstate{n'+1}
\]
with
$\Names_{n'+1}=\Names_{n'}$,
$\StoreA_{n'+1}=\StoreA_{n'}$,
$\StoreB_{n'+1}=\StoreB_{n'}$,
$\calP_{n'+1}=\calP_{n'}\msetminus \mset {\code{\piout$(t_1,t_2);Q'$}} \mcup
\mset{Q'}$,
$\Subst_{n'+1}=\Subst_{n'}\cup \set{^{t_2}/ _x}$ and
$\ActiveLocks_{n'+1}=\ActiveLocks_{n'}$ for a fresh $x$.

We define $f$  as follows:
\begin{equation*}
 f(i) :=
\begin{cases}
	 f_p(i) & \text{if $i\in\setN_{m}$} \\
	 n' & \text{if $m<i< n-1$} \\
	 n'+1 & \text{if $i=n$}
\end{cases}
\end{equation*}

Therefore,
Conditions~\ref{f:nonce} to \ref{e:actions} hold for $i<n-1$.
It is left to show that Conditions~\ref{f:nonce} to \ref{e:actions} hold for $n$.

\autoref{f:actions} holds since
$\mathit{hide}([E_1,\ldots,E_m])=_\ET[F_1,\ldots,n']$, and
$[E_{m+1},\ldots,E_{n-1}]=_\ET[F_{n'+1}]$, since $E_{n-1}=K(t_1)$.

\autoref{e:subst} holds since
$\Subst_{n'+1}=\Subst_{n'}\cup \set{^{t_2}/ _x}$, and therefore:
\begin{align*}
 \mset{x \Subst_{n'+1} \mid x\in\Dom(\Subst_{n'+1})}=&
\mset{x \Subst_{n'} \mid x\in\Dom(\Subst_{n'})} \mcup \mset{t_2} \\
=_\ET&\mset{\mathsf{Out}(t)\in \cup_{k \le n-2} S_k} \mcup \mset{t_2} \tag{by
\eqref{eq:subst}}\\
=&\mset{\mathsf{Out}(t)\in \cup_{k \le n} S_k}
\end{align*}

By definition of $\sem{P}$ and $\sem{P}_{=p}$, we have that
$Q' \bijp \state_{p\cdot 1}(\tilde t)$.
Therefore, and since
$\code{\piout$(t_1,t_2);Q'$} \bijp \state_{p}(\tilde t)$,
$\calP_{n'+1}=\calP_{n'} \msetminus \mset{\code{\piout$(t_1,t_2);Q'$}} \mcup
\mset{Q'}$,
and
$S_n=_\ET S_{n-1} \msetminus \mset{\In(a), \state_p(\tilde t)} \mcup \set{\state_{p\cdot
1}(\tilde t),\Out(t_2)}$,
\autoref{f:proc} holds.

Conditions
\autoref{f:nonce},
\ref{f:storea} and
\ref{f:storeb}
hold trivially.
\end{mycase} %}}}

\begin{mycase}%{{{ in
[$\ri=
\msr
{\state_p(\tilde t), \In(<t_1,t_2>)}
{\ChannelInEvent(\pair {t_1} {t_2})}
{\state_{p\cdot
1}(\tilde t, \tilde t')}$
 (for some $p,\tilde t,\tilde t'$ and $t_1,t_2\in \Mess$)]
Since the msr execution is normal, we have that
$S_{n-2}
\trans{K(t_1)}_\textsc{MDIn} S_{n-1}$. Since
$S_0 \trans{E_1}_{\sem{P}} \ldots \trans{E_n}_\sem{P} S_n$ is normal, so is
$S_0 \trans{E_1}_{\sem{P}} \ldots \trans{E_{n-1}}_\sem{P} S_{n-1}$, and
therefore
$S_0 \trans{E_1}_{\sem{P}} \ldots \allowbreak \trans{E_{n-2}}_\sem{P} S_{n-2}$. Hence
there is an $m< n-2$ such
$S_0 \trans{E_1}_{\sem{P}} \ldots \trans{E_{m}}_\sem{P} S_{m}$ is a normal
trace and $S_m \rightarrow^*_R
S_{n-1}$ for $R=\MDOutPubFreshApplFresh$.

By induction hypothesis, we have ${\cal P}_{n'} \bijp S_{m}$. Since
\MDOutPubFreshAppl, \textsc{Fresh} and \textsc{MDIn} do not add or
remove \state-facts, ${\cal P}_{n'} \bijp S_{n-2}$.
Let $Q\melem \calP_{n'}$ such that $Q \bijp \state_p(\tilde t)$.  Let
$\theta$ be a grounding substitution for $\state_p(\tilde
x)\in\prems(\sem{P}_{=p})$ such that $\tilde t=_E\tilde{x}\theta$.  Then
$\theta$ induces a substitution $\tau$ and a bijective renaming $\rho$ for
fresh, but not bound names (in $Q$) such that $P|_{p}\tau\rho=Q$ (see
\autoref{def:second-process-bijection}).
From the form of the rule $R$, and since $Q=P|_{p}\tau\rho$, we can deduce
that $Q=\code{\piin$(t_1,N);Q'$}$,  for $N$ a term that is not necessarily
ground, and a process $Q'=P|_{p\cdot 1}\tau \rho$. Since
$\ri\in_E\ginsts(R)$, we have that there is a substitution $\tau'$ such
that $N\tau'=_E t_2$.

% Why is it possible to execute this step
From the induction hypothesis, and since
$E_{m+1},\ldots,E_{n-2} = \emptyset$,
we have that
 \[ \Names_{n'} =  \set{ a \mid \ProtoNonce(a)\in\bigcup_{1 \leq j \le
n-2} E_j}. \]

From the induction hypothesis, and since no rule producing $\Out$-facts is
applied between step $m$ and step $n-2$, we have that
\begin{equation}
\label{eq:subst}
 \set{x \Subst_{n'} \mid x\in\Dom(\Subst_{n'})}^\#=\mset{\mathsf{Out}(t)\in \cup_{k
\le n-2} S_k}.
\end{equation}

Let $\tilde r= \set{ a:\freshsort \mid \RepNonce(a)\in\bigcup_{1\leq j \le
n-2} F_j}$.  Since $\K(<t_1,t_2>)\in \prems(\textsc{MDIn}\sigma)$ for
$\sigma(x)=<t_1,t_2>$, we have $\K(<t_1,t_2>)_\ET\in S_{n-2}$.  By
\autoref{lem:ded} and \autoref{lem:ded-eq}, we have $\nu \Names_{n'},\tilde
r.\Subst_{n'} \vdash <t_1,t_2>$. Therefore, $\nu \Names_{n'}.\Subst_{n'}
\vdash <t_1,t2>$. Using \textsc{DEq} and \textsc{DAppl} with the function
symbols $\mathit{fst}$ and $\mathit{snd}$, we have $\nu
\Names_{n'}.\Subst_{n'} \vdash t_1$ and $\nu \Names_{n'}.\Subst_{n'} \vdash
t2$. Therefore, we chose the following transition:
\[
  \cdots \trans{F_n'} \apipstate{n'} \trans{K(t_1)}
 \apipstate{n'+1}
\]
with
$\Names_{n'+1}=\Names_{n'}$,
$\StoreA_{n'+1}=\StoreA_{n'}$,
$\StoreB_{n'+1}=\StoreB_{n'}$,
$\calP_{n'+1}=\calP_{n'}\msetminus \mset {\code{\piin$(t_1,N);Q'$}} \mcup
\mset{Q'\tau'}$,
$\Subst_{n'+1}=\Subst_{n'}$ and
$\ActiveLocks_{n'+1}=\ActiveLocks_{n'}$.

We define $f$  as follows:
\begin{equation*}
 f(i) :=
\begin{cases}
	 f_p(i) & \text{if $i\in\setN_{m}$} \\
	 n' & \text{if $m<i< n-1$} \\
	 n'+1 & \text{if $i=n$}
\end{cases}
\end{equation*}

Therefore, Conditions~\ref{f:nonce} to \ref{e:actions} hold for $i<n-1$.
It is left to show that Conditions~\ref{f:nonce} to \ref{e:actions} hold
for $n$.

\autoref{f:actions} holds since
$\mathit{hide}([E_1,\ldots,E_m])=[F_1,\ldots,n']$, and
$[E_{m+1},\ldots,E_{n-1}]=[F_{n'+1}]$, since $E_{n-1}=K(t_1)$.

Let $\theta'$ such that $\ri=\theta'R$.
As established before, we have
$\tau'$ such
that $N\tau'=_E t_2$. By definition of $\sem{P}_{=p}$, we have
that $\state_{p\cdot 1}(\tilde t,\tilde t')\in_\ET\ginsts(P_{=p\cdot 1})$,
and that $\theta'= \theta\cdot \tau'$.
Since $\tau$ and $\rho$ are induced by $\theta$,
$\theta'$ induces $\tau \cdot \tau'$ and the same $\rho$. We have that
$Q'\tau'=(P|_{p\cdot 1}\tau \rho)\tau'=P|_{p}\tau\tau'\rho$ and therefore
$Q'\tau \bijp \state_{p\cdot 1}(\tilde t,\tilde t')$.
Thus, and since
$\code{\piin$(t_1,N);Q'$} \bijp \state_{p}(\tilde t)$,
$\calP_{n'+1}=\calP_{n'}\msetminus \mset {\code{\piin$(t_1,N);Q'$}} \mcup
\mset{Q'\tau'}$
and
$S_n=S_{n-1} \msetminus \mset{\In(<t_1,t_2>), \state_p(\tilde t)} \mcup
\mset{\state_{p\cdot 1}(\tilde t,\tilde t')}$,
\autoref{f:proc} holds.

Conditions
\autoref{f:nonce},
\ref{f:storea},
\ref{f:storeb} and
\ref{f:subst}
hold trivially.
\end{mycase}%}}}

\begin{mycase}%{{{ silent in/out
[$\ri=
\msr
{\semistate_p(\tilde s),\Ack(t_1,t_2)}
{}
{
\state_{p\cdot 1}(\tilde s)
}$
 (for some $p,\tilde t$ and $t_1,t_2\in \Mess$)]
Since the msr execution is normal, we have that there $p$,$q$,$\tilde
x,\tilde y,\tilde y'$ such that:
 \[ S_{n-3} \trans{~}_{R_1}  S_{n-2}
\trans{~}_{R_2}  S_{n-1}
\trans{~}_{R_3}  S_n \quad\text{, where:}\]
		\begin{itemize}
			\item $R_1 = [\state_p(\tilde x)] \to [\Msg(t_1,t_2),
\semistate_{p}(\tilde x)]$
		\item $R_2=[\state_q(\tilde y), \Msg(t_1,t_2)] \to
  [\state_{q\cdot 1}(\tilde y \cup \tilde y'), \Ack(t_1,t_2)]$
		\item $R_3=[\semistate_p(\tilde x),\Ack(t_1,t_2) ] \to [\state_{p\cdot
1}(\tilde x)]$
\end{itemize}.

Since in this case, there is a fact with symbol \Ack removed from $S_{n-1}$
to $S_n$, we have to apply a different argument to apply the induction
hypothesis.

Since
$\emptyset \trans{E_1}_{\sem{P}} \cdots
\trans{E_n}_{\sem{P}} S_n \in \execmsr(\sem{P})$
is normal, $n\ge 2$,
	 and $t_1,t_2\in\Mess$, $\Ack(t_1,t_2) \in (
S_{n-1} \msetminus S_n)$,
 there exists $m\le n-3$ such that $S_m \rightarrow^*_R
S_{n-3}$ for $R=\MDOutPubFreshAppl \cup  \textsc{Fresh}$ and
	$\emptyset \trans{E_1}_{\sem{P}} \cdots
\trans{E_m}_{\sem{P}} S_m \in \execmsr(\sem{P})$ is normal.
This allows us to apply the induction hypothesis on
$\emptyset \trans{E_1}_{\sem{P}} \cdots
\trans{E_m}_{\sem{P}} S_m \in \execmsr(\sem{P})$.
Hence there is a monotonically
increasing function from $\setN_{m} \to \setN_{n'}$  and an execution
such that Conditions~\ref{f:nonce} to~\ref{e:actions} hold. Let $f_p$ be
this function and note that $n'=f_p(m)$.

In the following case distinction, we extend
the previous execution by one step from $\apipstate{n'}$ to
$\apipstate{n'+1}$, and prove that Conditions~\ref{f:nonce}
to~\ref{f:actions} hold for $n'+1$.  By induction hypothesis, they hold for
all $i\le n'$. We define a function $f \colon \setN_n \to \setN_{n'+1}$ as
follows:
\begin{equation*}
\label{eq:definition-of-f}
 f(i) :=
\begin{cases}
	 f_p(i) & \text{if $i\in\setN_{m}$} \\
	 n' & \text{if $m<i\le n-3$} \\
	 n'+1 & \text{if $i=n$}
\end{cases}
\end{equation*}

Since, $S_m \rightarrow^*_R S_{n}$ for $R=\MDOutPubFreshApplFresh$, only
$S_n \msetminus S_m$ contains only \Fr-facts and \K-facts, and
$S_m \msetminus S_n$ contains only \Fr-facts and \Out-facts. Therefore,
\autoref{f:storeb}, \ref{f:proc} and \ref{f:subst} hold for all
$i\le n-3$. Since $E_{m+1},\ldots,E_{n-1}=\emptyset$, \autoref{f:nonce},
\ref{f:storea}, \ref{f:lock} and \ref{f:actions} hold for all $i\le n-3$.

Fix a bijection such that $\Processes_{n'} \bijp S_{m}$. We will abuse
notation by writing $P \bijp \state_p(\tilde t)$, if this bijection maps
$P$ to $\state_p(\tilde t)$.  Since \MDOutPubFreshAppl\  and \textsc{Fresh}
do not add or remove \state-facts, ${\cal P}_{n'} \bijp S_{n-3}$.
Let $P\melem \calP_{n'}$ such that $P \bijp \state_p(\tilde s)$.  Let
$Q\melem \calP_{n'}$ such that $Q \bijp \state_q(\tilde t)$.

  Let $\theta'$ be a grounding substitution for $\state_q(\tilde
y)\in\prems(\sem{P}_{=q})$ such that $\tilde t=_E\tilde{y}\theta'$.  Then
$\theta'$ induces a substitution $\tau'$ and a bijective renaming $\rho'$
for fresh, but not bound names (in $Q$) such that $P|_{q}\tau'\rho'=Q$ (see
\autoref{def:second-process-bijection}).

From the form of the rules $R_1$ and $R_3$,  and since $P=_\ET P|_{p}\tau\rho$,
for $\tau$ and $\rho$ induced by the grounding substitution for
$\state_p(\tilde x)$,
we can deduce that $P=_\ET \code{\piout~$t_1,t_2$; $P'$}$ for a process
$P'=P|_{p\cdot 1}\tau \rho$. Similarly, from
the form of $R_2$, we can deduce
$Q=_E\code{\piin$(t_1,N);Q'$}$,  for $N$ a term that is not necessarily
ground, and a process $Q'=P|_{q\cdot 1}\tau' \rho'$. Since
$S_{n-2} \trans{~}_{R_2} S_{n-1}$, we have that there is a substitution $\tau^*$ such
that $N\tau'\rho'\tau^*=_E t_2$ and $((\tilde y \cup \vars(N)) \setminus
\tilde{ y})\tau^*=_\ET\tilde t'$, where $\tilde t'$ such that $\state_{q\cdot 1}
(\tilde t,\tilde t')\in
S_{n-1} \msetminus S_{n-2}$.

Given that $Q=_E\code{\piin$(t_1,N);Q'$}$ and $P=_\ET
\code{\piout~$t_1,t_2$; $P'$}$, have that $\Processes_{n'}=\Processes' \mcup
\mset{\piout~t_1,t_2; P',\allowbreak \piin~(t_1',N);Q'}$ with $t_1=_E t_1'$ and
$t_2 =_\ET N\tau^*$.  Therefore, we chose the following transition:
\[
  \cdots \trans{F_n'} \apipstate{n'} \trans{K(t_1)}
 \apipstate{n'+1}
\]
with
$\Names_{n'+1}=\Names_{n'}$,
$\StoreA_{n'+1}=\StoreA_{n'}$,
$\StoreB_{n'+1}=\StoreB_{n'}$,
$\calP_{n'+1}=\Processes' \mcup \mset{P',Q'}$,
$\Subst_{n'+1}=\Subst_{n'}$ and
$\ActiveLocks_{n'+1}=\ActiveLocks_{n'}$.

Conditions~\ref{f:nonce} to \ref{e:actions} hold for $i\le n-3$.
It is left to show that Conditions~\ref{f:nonce} to \ref{e:actions} hold
for $n$.

As established before, we have
$\tau^*$ such
that $N\tau'\rho'\tau^* =_E t_2$. Let $\state_q(\tilde t,\tilde t')$ be the
state variable added to $S_{n-1}$. Then, $((\tilde y \cup \vars(N)) \setminus
\tilde y)\tau^*=\tilde t'$.
By definition of $\sem{P}_{=q}$, we have that $\state_{q\cdot 1}(\tilde
t,\tilde t')\in\prems(\ginsts(P_{=p\cdot 1}))$ for a grounding substitution
$\theta_{q\cdot 1}= \theta'\cdot \tau^*$. Since $\tau'$ and $\rho'$ are
induced by $\theta'$, $\theta_{q\cdot 1}$ induces $\tau \cdot \tau'$ and the same
$\rho$. We have that $Q'\tau'=(P|_{q\cdot 1}\tau'
\rho')\tau*=P|_{q\cdot 1}\tau\tau'\rho$ and therefore $Q'\tau^* \bijp \state_{q\cdot
1}(\tilde t,\tilde t')$.
Similarly, we have $P' \bijp \state_{q\cdot
1}(\tilde s)$. We conclude that \autoref{f:proc} holds.

Conditions
\autoref{f:nonce},
\ref{f:storea},
\ref{f:storeb},
\ref{f:subst},
\ref{f:lock}
and
\ref{f:actions}
hold trivially.
\end{mycase} %}}}

\begin{mycase}%{{{ if positive
[$\ri=
\msr
{\state_p(\tilde t)}
{\Eq(t_1,t_2)}
{\state_{p\cdot
1}(\tilde t)}$
 (for some $p,\tilde t$ and $t_1,t_2\in\Mess$)]
\inclusionBcaseintroHelper
{ \mathtt{if}~t_1=t_2~\mathtt{then}~Q_1~\mathtt{else}~ Q_2 }
{ for a process $Q'=P|_{p\cdot 1}\tau \rho$}
{Since, $[E_1,\ldots,E_n\vDash \alpha$, and thus $[E_1,\ldots,E_m\vDash
\AssEq$, we have that $t_1=_\ET t_2$.}
{$\Names_{n'+1}=\Names_{n'}$,
$\StoreA_{n'+1}=\StoreA_{n'}$,
$\StoreB_{n'+1}=\StoreB_{n'}$,
$\calP_{n'+1}=\calP_{n'}
\msetminus \mset{ \mathtt{if}~t_1=t_2 \mathtt{then}~Q_1~\mathtt{else}~ Q_2 }
\mcup \mset{Q_1}$,
$\Subst_{n'+1}=\Subst_{n'}$ and
$\ActiveLocks_{n'+1}=\ActiveLocks_{n'}$}
{{~}}

By definition of $\sem{P}$ and $\sem{P}_{=p}$, we have that
$Q_1 \leftrightarrow \state_{p\cdot 1}(\tilde t)$.
Therefore, and since
$ \mathtt{if}~t_1=t_2~\mathtt{then}~Q_1~\mathtt{else}~ Q_2  \leftrightarrow \state_{p}(\tilde t)$,
$\calP_{n'+1}=\calP_{n'} \msetminus \mset{\mathtt{if}~t_1=t_2~\mathtt{then}~Q_1~\mathtt{else}~ Q_2 } \mcup \mset{Q_1}$,
and $S_n=S_{n-1} \msetminus \mset{\state_p(\tilde t)} \mcup
\mset{\state_{p\cdot 1}(\tilde t)}$,
\autoref{f:proc} holds.
Conditions
\ref{f:nonce},
\ref{f:storea},
\ref{f:storeb},
\ref{f:subst},
\ref{f:lock} and
\ref{f:actions}
hold trivially.
\end{mycase}%}}}

\begin{mycase}%{{{ if negative
[$\ri=
\msr
{\state_p(\tilde t)}
{\mathrm{NotEq}(t_1,t_2)}
{\state_{p\cdot
1}(\tilde t)}$
 (for some $p,\tilde t$ and $t_1,t_2\in\Mess$)]
In this case, the proof is almost the same as in the previous case, except
that $\AssNotEq$ is the relevant axiom, $Q_2$ is chosen instead of
$Q_1$ and
and $S_n=S_{n-1} \msetminus \mset{\state_p(\tilde t)} \mcup
\mset{\state_{p\cdot 2}(\tilde t)}$.
\end{mycase}%}}}

\begin{mycase}%{{{ event
[$\ri=
\msr
{\state_p(\tilde t)}
{\mathrm{F},\Event()}
{\state_{p\cdot
1}(\tilde t)}$
 (for some $p,\tilde t$)]
This is a special case of the case where
$\ri= \msr {\state_p(\tilde t),l} {a} {\state_{p\cdot 1}(\tilde t),r}$ for
$l=r=\emptyset$ and $a=F$.
\end{mycase}%}}}

\begin{mycase}%{{{ insert
[$\ri=
\msr
{\state_p(\tilde t)}
{\mathit{Insert}(t_1,t_2)}
{\state_{p\cdot
1}(\tilde t)}$
 (for some $p,\tilde t$ and $t_1,t_2\in\Mess$)]
\inclusionBcaseintro
{ \mathtt{insert}~t_1,t_2; Q'}
{ for a process $Q'=P|_{p\cdot 1}\tau \rho$}
{}
{$\Names_{n'+1}=\Names_{n'}$,
$\StoreA_{n'+1}=\StoreA_{n'}[t_1 \mapsto t_2]$,
$\StoreB_{n'+1}=\StoreB_{n'}$,
$\calP_{n'+1}=\calP_{n'}
\msetminus \mset{\mathtt{insert}~t_1,t_2; Q'}
\mcup \mset{Q'}$,
$\Subst_{n'+1}=\Subst_{n'}$ and
$\ActiveLocks_{n'+1}=\ActiveLocks_{n'}$}
{{~}}
{ $S_n=S_{n-1} \msetminus \mset{\state_p(\tilde t)} \mcup
\mset{\state_{p\cdot 1}(\tilde t)}$}

\autoref{f:storea} holds, since $E_n=\mathit{Insert}(t_1,t_2)$ is the last
element of the trace.

Conditions
\ref{f:nonce},
\ref{f:storeb},
\ref{f:subst},
\ref{f:lock} and
\ref{f:actions}
hold trivially.
\end{mycase}%}}}

\begin{mycase}%{{{ delete
[$\ri=
\msr
{\state_p(\tilde t)}
{\mathit{Delete}(t_1,t_2)}
{\state_{p\cdot
1}(\tilde t)}$
 (for some $p,\tilde t$ and $t_1,t_2\in\Mess$)]
\inclusionBcaseintro
{ \mathtt{delete}~t_1; Q'}
{ for a process $Q'=P|_{p\cdot 1}\tau \rho$}
{}
{$\Names_{n'+1}=\Names_{n'}$,
$\StoreA_{n'+1}=\StoreA_{n'}[t_1 \mapsto t_2]$,
$\StoreB_{n'+1}=\StoreB_{n'}$,
$\calP_{n'+1}=\calP_{n'}
\msetminus \mset{\mathtt{delete}~t_1; Q'}
\mcup \mset{Q'}$,
$\Subst_{n'+1}=\Subst_{n'}$ and
$\ActiveLocks_{n'+1}=\ActiveLocks_{n'}$}
{{~}}
{ $S_n=S_{n-1} \msetminus \mset{\state_p(\tilde t)} \mcup
\mset{\state_{p\cdot 1}(\tilde t)}$}

\autoref{f:storea} holds, since $E_n=\mathit{Delete}(t_1,t_2)$ is the last
element of the trace.

Conditions
\ref{f:nonce},
\ref{f:storeb},
\ref{f:subst},
\ref{f:lock} and
\ref{f:actions}
hold trivially.
\end{mycase}%}}}

\begin{mycase}%{{{ lookup - positive
[$\ri=
\msr
{\state_p(\tilde t)}
{\mathit{IsIn}(t_1,t_2)}
{\state_{p\cdot
1}(\tilde t, t_2)}$
 (for some $p,\tilde t$ and $t_1,t_2\in\Mess$)]
\inclusionBcaseintroHelper
{ \mathtt{lookup}\allowbreak~t_1~\mathtt{as}~v\allowbreak~\mathtt{in}~Q_1~\mathtt{else}~Q_2 }
{ for some variable $V$, and two processes
$Q_1=P|_{p\cdot 1}\tau\rho$ and $Q_2=P|_{p\cdot 2}\tau\rho$}
{ Since $[E_1,\ldots,E_n]\vDash \AssSetIn$, there is an $i<n$ such that
$\mathit{Insert}(t_1,t_2)\in_\ET E_i$ and there is no $j$ such that $i<j<n$
and $\mathit{Delete}(t_1)\in_\ET E_j$ or
and $\mathit{Insert}(t_1,t_2)\in_ET E_j$. Since $E_m,\ldots,E_n=\emptyset$,
we know that $i<m$. Hence, by induction hypothesis,
	$\StoreA_{n'}(t_1)=t_2$.}
{$\Names_{n'+1}=\Names_{n'}$,
$\StoreA_{n'+1}=\StoreA_{n'}$,
$\StoreB_{n'+1}=\StoreB_{n'}$,
$\calP_{n'+1}=\calP_{n'}
\msetminus \mset{
\mathtt{lookup}~t_1~\mathtt{as}~v~\mathtt{in}~Q_1\allowbreak~\mathtt{else}~Q_2}
\mcup \mset{Q_1\set{^{t_2}/ _{v}} }$,
$\Subst_{n'+1}=\Subst_{n'}$ and
$\ActiveLocks_{n'+1}=\ActiveLocks_{n'}$}
{{~}}

By definition of $\sem{P}$ and $\sem{P}_{=p}$, we have that
$Q_1\set{ ^v / _{t_2}} \leftrightarrow \state_{p\cdot 1}(\tilde t, t_2)$
(for $\tau'=\tau[v \mapsto t_2]$ and $\rho'=\rho$).
Therefore, and since
$ \mathtt{lookup}\allowbreak~t_1~\mathtt{as}~v\allowbreak~\mathtt{in}~Q_1~\mathtt{else}~Q_2 \leftrightarrow \state_{p}(\tilde t)$,
$\calP_{n'+1}=\calP_{n'} \msetminus \mset{ \mathtt{lookup}\allowbreak~t_1~\mathtt{as}~v\allowbreak~\mathtt{in}~Q_1~\mathtt{else}~Q_2} \mcup \mset{Q'}$,
and $S_n=\S_{n-1} \msetminus \mset{\state_p(\tilde t)} \mcup \mset{
\state_{p\cdot 1}(\tilde t, t_2)}$,
\autoref{f:proc} holds.

Conditions
\ref{f:nonce},
\ref{f:storea},
\ref{f:storeb},
\ref{f:subst},
\ref{f:lock} and
\ref{f:actions}
hold trivially.
\end{mycase}%}}}

\begin{mycase}%{{{ lookup - negative
[$\ri=
\msr
{\state_p(\tilde t)}
{\mathit{IsNotSet}(t_1)}
{\state_{p\cdot
2}(\tilde t)}$
 (for some $p,\tilde t$ and $t_1\in\Mess$)]
\inclusionBcaseintroHelper
{ \mathtt{lookup}\allowbreak~t_1~\mathtt{as}~v\allowbreak~\mathtt{in}~Q_1~\mathtt{else}~Q_2 }
{ for a variable $v$ and two processes
$Q_1=P|_{p\cdot 1}\tau\rho$ and $Q_2=P|_{p\cdot 2}\tau\rho$}
{ Since $[E_1,\ldots,E_n]\vDash \AssSetNotIn$, there is no $i<n$ such that
$\mathit{Insert}(t_1,t_2)\in_\ET E_i$ and there is no $j$ such that $i<j<n$
and $\mathit{Delete}(t_1)\in_\ET E_j$ or
and $\mathit{Insert}(t_1,t_2)\in_ET E_j$. Since $E_m,\ldots,E_n=\emptyset$,
we know that holds $j<m$. Hence, by induction hypothesis,
	$\StoreA_{n'}(t_1)$ is undefined.}
{$\Names_{n'+1}=\Names_{n'}$,
$\StoreA_{n'+1}=\StoreA_{n'}$,
$\StoreB_{n'+1}=\StoreB_{n'}$,
$\calP_{n'+1}=\calP_{n'}
\msetminus \mset{
\mathtt{lookup}~t_1~\mathtt{as}~v~\mathtt{in}~Q_1\allowbreak~\mathtt{else}~Q_2}
\mcup \mset{Q_2}$,
$\Subst_{n'+1}=\Subst_{n'}$ and
$\ActiveLocks_{n'+1}=\ActiveLocks_{n'}$}
{{~}}

By definition of $\sem{P}$ and $\sem{P}_{=p}$, we have that
$Q_2 \leftrightarrow \state_{p\cdot 2}(\tilde t)$.
Therefore, and since
$ \mathtt{lookup}\allowbreak~t_1~\mathtt{as}~v\allowbreak~\mathtt{in}~Q_1~\mathtt{else}~Q_2 \leftrightarrow \state_{p}(\tilde t)$,
$\calP_{n'+1}=\calP_{n'} \msetminus \mset{
\mathtt{lookup}\allowbreak~t_1~\mathtt{as}~v\allowbreak~\mathtt{in}~Q_1~\mathtt{else}~Q_2}
\mcup \mset{Q_2}$,
and $S_n=\S_{n-1} \msetminus \mset{\state_p(\tilde t)} \mcup \mset{
\state_{p\cdot 2}(\tilde t)}$,
\autoref{f:proc} holds.

Conditions
\ref{f:nonce},
\ref{f:storea},
\ref{f:storeb},
\ref{f:subst},
\ref{f:lock} and
\ref{f:actions}
hold trivially.
\end{mycase}%}}}

\begin{mycase}%{{{ lock
[$\ri=
\msr
{\state_p(\tilde t),\Fr(\mathit{lock}_l)}
{\mathit{Lock}(\mathit{lock}_l,t)}
{\state_{p\cdot
1}(\tilde t, \mathit{lock}_l)}$
 (for some $p,\tilde t$, $\mathit{lock}_l\in\FN$ and $t\in\Mess$)]
\inclusionBcaseintroHelper
{ \mathtt{lock}^l~t;~Q' }
{ for $Q'=P|_{p\cdot 1}\tau\rho$}
{ Since $[E_1,\ldots,E_n]\vDash \AssLock$, for every $i<n$ such that
$\mathit{Lock}(l_p,t)\in_\ET E_i$, there a $j$ such that $i<j<n$ and
$\mathit{Unlock}(l_p,t)\in_\ET E_j$, and in between $i$ and $j$, there is
no lock or unlock, \ie, for all $k$ such that $i<k<j$, and all $l_i$,
$\mathit{Lock}(l_i,t)\notin_\ET E_k$ and
$\mathit{Unlock}(l_i,t)\notin_\ET E_k$.

Since $E_m,\ldots,E_n=\emptyset$, we know that this holds for $i<m$ and
$j<m$ as well. By induction hypothesis, \autoref{f:lock}, this implies that
$t\not\in_\ET \ActiveLocks_{n'}$.}
{$\Names_{n'+1}=\Names_{n'}$,
$\StoreA_{n'+1}=\StoreA_{n'}$,
$\StoreB_{n'+1}=\StoreB_{n'}$,
$\calP_{n'+1}=\calP_{n'}
\msetminus \mset{
 \mathtt{lock}^l~t;~Q'
} \mcup \mset{Q'}$,
$\Subst_{n'+1}=\Subst_{n'}$ and
$\ActiveLocks_{n'+1}=\ActiveLocks_{n'} \cup \set{t}$}
{{~}}

By definition of $\sem{P}$ and $\sem{P}_{=p}$, we have that
$Q' \leftrightarrow \state_{p\cdot 1}(\tilde t)$.
Therefore, and since
$  \mathtt{lock}^l~t;~Q' \leftrightarrow \state_{p}(\tilde t)$,
$\calP_{n'+1}=\calP_{n'}
\msetminus \mset{
 \mathtt{lock}^l~t;~Q'
} \mcup \mset{Q'}$,
and $S_n=\S_{n-1} \msetminus \mset{\state_p(\tilde t), \Fr(\mathit{lock}_l)} \mcup \mset{
\state_{p\cdot 1}(\tilde t, \mathit{lock}_l)}$,
\autoref{f:proc} holds.

\autoref{f:lock} holds since $E_n=\mset{\mathrm{Lock}(\mathit{lock}_l,t)}$ is added to
the end of the trace.

Conditions
\ref{f:nonce},
\ref{f:storea},
\ref{f:storeb},
\ref{f:subst}
and
\ref{f:actions}
hold trivially.
\end{mycase}%}}}

\begin{mycase}%{{{ unlock
[$\ri=
\msr
{\state_p(\tilde t)}
{\mathit{Unlock}(n_l,t)}
{\state_{p\cdot
1}(\tilde t)}$
 (for some $p,\tilde t$, $n_l\in\FN$ and $t\in\Mess$)]
\inclusionBcaseintroHelper
{ \mathtt{unlock}^l~t;~Q' }
{ for $Q'=P|_{p\cdot 1}\tau\rho$}
{}
{$\Names_{n'+1}=\Names_{n'}$,
$\StoreA_{n'+1}=\StoreA_{n'}$,
$\StoreB_{n'+1}=\StoreB_{n'}$,
$\calP_{n'+1}=\calP_{n'}
\msetminus \mset{
 \mathtt{unlock}^l~t;~Q'
} \mcup \mset{Q'}$,
$\Subst_{n'+1}=\Subst_{n'}$ and
$\ActiveLocks_{n'+1}=\ActiveLocks_{n'} \setminus \set{t}$}
{{~}}

By definition of $\sem{P}$ and $\sem{P}_{=p}$, we have that
$Q' \leftrightarrow \state_{p\cdot 1}(\tilde t)$.
Therefore, and since
$  \mathtt{unlock}^l~t;~Q' \leftrightarrow \state_{p}(\tilde t)$,
$\calP_{n'+1}=\calP_{n'}
\msetminus \mset{
 \mathtt{unlock}^l~t;~Q'
} \mcup \mset{Q'}$,
and $S_n=\S_{n-1} \msetminus \mset{\state_p(\tilde t)} \mcup \mset{
\state_{p\cdot 1}(\tilde t)}$,
\autoref{f:proc} holds.

We show that \autoref{f:lock} holds for
$\ActiveLocks_{n'+1}=\ActiveLocks_{n'} \setminus \set{t}$: For all
$t' \neq_\ET t$, $t' \in_E \ActiveLocks_{n'} \Leftrightarrow t'\in_E
\ActiveLocks_{n'+1}$ by induction hypothesis. If $t \not\in_\ET
\ActiveLocks_{n'}$, then
$\forall j\leq m,u .  \mathrm{Lock}(u,t) \in_\ET E_j
    \rightarrow \exists j < k \leq n. \mathrm{Unlock}(u,t) \in_\ET E_k$.
Since we have
$E_m,\ldots,E_{n-1}=\emptyset$ and $E_n= \mset{
\mathrm{Unlock}(n_l,t)}$, we can strengthen this to
$\forall j\leq n,u .  \mathrm{Lock}(u,t) \in_\ET E_j
    \rightarrow \exists j < k \leq n. \mathrm{Unlock}(u,t) \in_\ET E_k$,
which means that the condition holds in this case.
If $t \in_\ET \ActiveLocks_{n'}$, then
$\exists j\leq n,u .
    \mathrm{Lock}(u,t) \in_\ET E_j
    \wedge \forall j < k \leq n. \mathrm{Unlock}(u,t) \not \in_\ET E_k$
and since $E_m,\ldots,E_{n-1}=\emptyset$ and $E_n= \mset{
\mathrm{Unlock}(n_l,t)}$, a
contradiction to \autoref{e:lock} would constitute of $j$ and $u\neq_\ET
n_l$ such that
    $\mathrm{Lock}(u,t) \in_\ET E_j$ and
    $ \forall j < k \leq n. \mathrm{Unlock}(u,t) \not \in_\ET E_k$.

We will show that this leads to a contradiction with
$[E_1,\ldots,E_n]\vDash \alpha$.
Fix $j$
and $u$.
By definition of
$\sem{P}$ and well-formedness of $P$, there is a $p_l$ that is a prefix of
$p$ such that $P|_{q_l}=\mathtt{lock}^l t; Q''$ for the same annotation $l$
and parameter $t$. The form of the translation guarantees that if
$\state_{p}(\tilde t)\in S_n$, then for some $\tilde t'$ there is $i\le n$
such that
$\state_{p'}(\tilde t')\in S_i$, if $p'$ is a prefix of $p$. We therefore
have that there is $i<n$ such that $E_i=_E\mset{ \mathrm{Lock}(n_l,t)}$.
We proceed by case distinction:

\noindent \underline{Case 1:} $j<i$ (see
\autoref{fig:visualisation-b-of-case-1-}). Since
    $ \forall j < k \leq n. \mathrm{Unlock}(u,t) \not \in_\ET E_k$,
 $[E_1,\dots,E_n] \not\vDash \AssLock$.
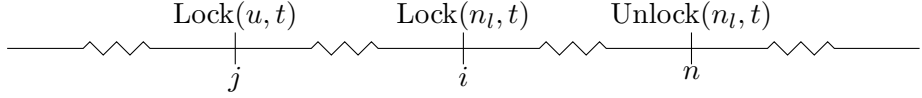
\begin{figure}[ht] % (fold)%{{{
\centering
\begin{tikzpicture}[scale=3 ]

%draw horizontal line
\draw[snake] (0,0) -- (1,0);
\draw[snake] (1,0) -- (2,0);
\draw[snake] (2,0) -- (3,0);
\draw[snake] (3,0) -- (4,0);

%draw vertical lines
\foreach \x in {1,2,3}
   \draw (\x cm,2pt) -- (\x cm,-2pt);

%draw nodes
\draw (1,0) node[below=3pt] {$ j $} node[above=3pt] {$\mathrm{Lock}(u,t)$};
\draw (2,0) node[below=3pt] {$ i $} node[above=3pt] {$\mathrm{Lock}(n_l,t)$};
\draw (3,0) node[below=3pt] {$ n $} node[above=3pt] {$\mathrm{Unlock}(n_l,t)$};
\end{tikzpicture}
\caption{ Visualisation of Case 1. }
\label{fig:visualisation-b-of-case-1-}
\end{figure}%}}}

\noindent \underline{Case 2:} $i<j$ (see
\autoref{fig:visualisation-b-of-case-2-}).
By definition of $\overline P$, there is no parallel and
no replication between $p_l$ and $p$. Note that
any rule in $\sem{P}$ that produces a state named
$\state_q$ for a non-empty $q$ is such that it requires a fact with name
$\state_{q'}$ for $q=q'\cdot 1$ or $q=q'\cdot 2$ (in case of the
translation of \code{out}, it might require $\semistate_{q'}$, which in
turn requires $\state_{q'}$). Therefore, there cannot be a second
$k\neq n$ such that
$\mathrm{Unlock}(n_l,t)\in_\ET E_k$ (since $n_l$ was added in a
\Fr-fact in to $S_i$).
This means in particular that there is not $k$ such that $i<k<n$ and
$\mathrm{Unlock}(n_l,t)\in_\ET E_k$. Therefore,
 $[E_1,\dots,E_n] \not\vDash \AssLock$.
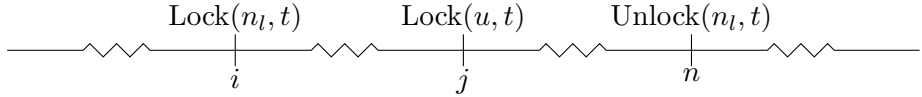
\begin{figure}[ht] % (fold)%{{{
\centering
\begin{tikzpicture}[scale=3
]

%draw horizontal line
\draw[snake] (0,0) -- (1,0);
\draw[snake] (1,0) -- (2,0);
\draw[snake] (2,0) -- (3,0);
\draw[snake] (3,0) -- (4,0);

%draw vertical lines
\foreach \x in {1,2,3}
   \draw (\x cm,2pt) -- (\x cm,-2pt);

%draw nodes
\draw (1,0) node[below=3pt] {$ i $} node[above=3pt] {$\mathrm{Lock}(n_l,t)$};
\draw (2,0) node[below=3pt] {$ j $} node[above=3pt] {$\mathrm{Lock}(u,t)$};
\draw (3,0) node[below=3pt] {$ n $} node[above=3pt] {$\mathrm{Unlock}(n_l,t)$};
\end{tikzpicture}
\caption{ Visualisation of Case 2. }
\label{fig:visualisation-b-of-case-2-}
\end{figure}%}}}

Conditions
\ref{f:nonce},
\ref{f:storea},
\ref{f:storeb},
\ref{f:subst}
and
\ref{f:actions}
hold trivially.
\end{mycase}%}}}

\begin{mycase}%{{{ embedded msr
[$\ri=
\msr
{\state_p(\tilde t),l'}
{a',\Event()}
{\state_{p\cdot
1}(\tilde t,\tilde t'),r'}$
 (for some $p,\tilde t,\tilde t'$ and $l',r',a'\in\GroundFacts^*$)]
\inclusionBcaseintroHelper
{ l \msrewrite a r; Q'}
{ for a process $Q'=P|_{p\cdot 1}\tau \rho$}
{ Since $\ri\in_E\ginsts(R)$, we have that there is a substitution $\tau^*$ such
that $(l \msrewrite a r)\tau^*=l' \msrewrite a' r'$, $\lfacts(l')\msubset S_{n-1}$, $\pfacts(l')\subset_E
S_{n-1}$ and, from the definition of $\sem{P}$
for embedded msr rules, $\vars(l)\tau^*=\tilde t'$. Since $P$ is
well-formed, no fact in \ReservedFacts appears in neither $l$ nor $r$, so from
\autoref{e:storeb} in the induction hypothesis, we have that
$\lfacts(l')\msubset \StoreB_m=\StoreB_{n'}$ and
$\pfacts(l')\subset \StoreB_m=\StoreB_{n'}$.}
{$\Names_{n'+1}=\Names_{n'}$,
$\StoreA_{n'+1}=\StoreA_{n'}$,
$\StoreB_{n'+1}=\StoreB_{n'}\setminus \lfacts(l') \mcup r' $,
$\calP_{n'+1}=\calP_{n'}
\msetminus \mset{ l \msrewrite a r; Q'}
\mcup \mset{Q'\tau^*}$,
$\Subst_{n'+1}=\Subst_{n'}$ and
$\ActiveLocks_{n'+1}=\ActiveLocks_{n'}$}
{a'}

\autoref{f:actions} holds since
$\mathit{hide}([E_1,\ldots,E_m])=[F_1,\ldots,n']$, and
$E_{m+1},\ldots,E_{n-1}=\emptyset$, while $E_n=F_n'\setminus \Event()=a'$
(note that $\Event() \in \ReservedFacts$).

As established before, we have
$\tau^*$ such
that $(l\msrewrite a r )\tau^*=_E l \msrewrite a r$. By the definition of $\sem{P}_{=p}$, we have
that $\state_{p\cdot 1}(\tilde t,\tilde t')\in_\ET\ginsts(P_{=p\cdot 1})$,
and a $\theta'= \theta\cdot\tau^*$ that is grounding for $\state_{p\cdot
1}(\tilde t,\tilde t')$. Since $\tau$ and $\rho$ are induced by $\theta$,
$\theta'$ induces $\tau \cdot \tau^*$ and the same $\rho$. We have that
$Q'\tau^*=(P|_{p\cdot 1}\tau \rho)\tau^*=P|_{p}\tau\tau^*\rho$ and therefore
$Q'\tau \bijp \state_{p\cdot 1}(\tilde t,\tilde t')$.
Thus, and since
$l \msrewrite a r; Q' \bijp \state_{p}(\tilde t)$,
$\calP_{n'+1}=\calP_{n'}\msetminus \mset {l \msrewrite a r; Q'} \mcup
\mset{Q'\tau^*}$
and
$S_n=S_{n-1} \msetminus \lfacts(l') \mcup r' \msetminus \mset{\state_{p}(\tilde t)}
\mcup \mset{\state_{p\cdot 1}(\tilde t,\tilde t')}
$,
\autoref{f:proc} holds.

\autoref{f:storeb}, holds since
\begin{align*}
S_n \msetminus \ReservedFacts =&
(S_{n-1} \msetminus \lfacts(l') \mcup r' \msetminus \mset{\state_{p}(\tilde t)}
\mcup \mset{\state_{p\cdot 1}(\tilde t,\tilde t')} ) \msetminus
\ReservedFacts \\
=& (S_{n-1} \msetminus \lfacts(l') \mcup r' ) \msetminus \ReservedFacts \\
=& S_{n-1} \msetminus \ReservedFacts  \msetminus \lfacts(l') \mcup r' \\
=& \StoreB_{n'}  \msetminus \lfacts(l') \mcup r' \\
=& \StoreB_{n'+1}
\end{align*}

Conditions
\ref{f:nonce},
\ref{f:storea},
\ref{f:subst},
\ref{f:lock} and
\ref{f:actions}
hold trivially.
\end{mycase}%}}}

\end{component}
\end{proof}

\traceequivalence*
%TODO maybe elaborate more?
\begin{proof}
  From \autoref{lem:inclusion-apip-in-msr}, we can conclude that
  \[ \tracespi(P) \subseteq \set{\hide(t) | \tr\in\tracesmsr(\sem{P}) \text{ and } \tr
    \vDash \alpha}=\hide(\filter(\tracesmsr(\sem{P}))). \]
  From \autoref{lem:msr-normalisation}, we have that  
  \[ \hide(\filter(\tracesmsr(\sem{P}))) = \set{ \tr \in
    \hide(\filter(\tracesmsr(\sem{P}))) \mid \tr\text{
      is normal}}. \]
  From \autoref{lem:inclusion-msr-in-apip}, we can conclude that 
  \[ \set{ \tr \in
    \hide(\filter(\tracesmsr(\sem{P}))) \mid \tr\text{
      is normal}} \subseteq \tracespi(P).\]
  Hence 
  \[ \hide(\filter(\tracesmsr(\sem{P}))) \subseteq \tracespi(P). \]
\end{proof}

% vim:fdm=marker tw=70

\fi

\end{document}